\documentclass[cmp]{svjour}
\usepackage{latexsym}
\usepackage[free-standing-units]{siunitx}
\usepackage{circuitikz}
\usetikzlibrary{arrows}
\usepackage{tikz}
\usepackage{times}
\usepackage{amsmath}
\usepackage{amssymb}
\usepackage{amscd}
\usepackage{amsfonts}
\usepackage{amstext}
\usepackage{graphicx}
\usepackage[all]{xy}
\usepackage{float}
\usepackage{nicefrac}

\newcommand{\hq }{{/\kern -.185em/}}

\newtheorem{Fact}{Fact}

\newcommand{\bbZ}{{\mathbb Z}}
\newcommand{\bbR}{{\mathbb R}}

\newcommand{\vx}{\mathbf{x}}

\newcommand\Vcalbar{\bar{\cal V}}
\newcommand\Vbar{{\bar V}}

\newcommand{\half}{{\textstyle \frac{1}{2}}}
\newcommand\Cbar{{\bar P}}
\newcommand\Gbar{{\bar \Gamma}}
\newcommand\Omegabar{{\bar \Omega}}

\begin{document}

\title{n-particle quantum statistics on graphs}

\author{J.M. Harrison\inst{1}, J.P. Keating\inst{2}, J.M. Robbins\inst{2}, A. Sawicki\inst{2,}\inst{3}}

\institute{Department of Mathematics, Baylor University, One Bear Place,
Waco, TX 76798-7328 USA \and School of Mathematics, University of Bristol, University Walk, Bristol BS8 1TW, UK \and Center for Theoretical Physics, Polish Academy of
Sciences, Al. Lotnik\'ow 32/46, 02-668 Warszawa, Poland}

\date{Received: date / Accepted: date}
\communicated{name}
\maketitle

\begin{abstract}
We develop a full characterization of abelian quantum statistics on graphs. We explain how the number of anyon phases is related to connectivity.  For $2$-connected graphs the independence of quantum statistics with respect to the number of particles is proven. For non-planar $3$-connected graphs we identify bosons and fermions as the only possible statistics, whereas for planar  $3$-connected graphs we show that one anyon phase exists. Our approach also yields an alternative proof of the structure theorem for the first homology group of $n$-particle graph configuration spaces. Finally, we determine the topological gauge potentials for $2$-connected graphs.
\end{abstract}

\section{Introduction}
In classical mechanics, particles are considered distinguishable. Therefore, the $n$-particle configuration space is the Cartesian product, $M^{\times n}$, where $M$ is the one-particle configuration space. By contrast, in quantum mechanics elementary particles may be considered indistinguishable.  This conceptual difference in the description of many-body systems prompted Leinaas and Myrheim \cite{LM77} (see also \cite{S70,W90}) to study classical configuration spaces of indistinguishable particles, $C_n(M)$, and led to the discovery of anyon statistics.

Indistinguishability of classical particles places constraints on the usual configuration space, $M^{\times n}$. Configurations that differ by particle exchange must be identified.  One also assumes that two classical particles cannot occupy the same configuration. Consequently, the classical configuration space of $n$ indistinguishable particles is the orbit space $C_n(M)=(M^{\times n}-\Delta)/S_{n}$, where $\Delta$ corresponds to the configurations for which at least two particle are at the same point in $M$, and $S_n$ is the permutation group.

Significantly, the space $C_n(M)$ may have non-trivial topology. One can, for example, easily calculate that for $n$ particles in $M=\mathbb{R}^{N}$ the first homology group $H_1(C_n(\mathbb{R}^N))$ is $\mathbb{Z}$ if $N=2$ and $\mathbb{Z}_{2}$ when $N\geq3$ \cite{D85,S12}. This fact, combined with the standard quantization procedure on topologically non-trivial configuration spaces, explains, at a kinematic level, the existence of anyons in two dimensions and only bosons or fermions in higher dimensions.  It also raises the question of what  quantum statistics are possible on spaces with richer topology.

In order to explore how the quantum statistics picture depends on topology, the case of two indistinguishable particles on a graph was  studied in \cite{JHJKJR} (see also \cite{BE92}). A graph $\Gamma$ is a network consisting of vertices (or nodes) connected by edges.  Quantum mechanically, one can either consider the one-dimensional Schr\"{o}dinger operator acting on the edges, with matching conditions for the wave functions at the vertices, or a discrete Schr\"{o}dinger operator acting on connected vertices (i.e.~a tight-binding model on the graph).  Such systems are of considerable independent interest and their single-particle quantum mechanics has been studied extensively in recent years \cite{Berkolaiko13}.  The extension of this theory to many-particle quantum graphs was another motivation for \cite{JHJKJR} (see also \cite{Bolte13}).  The discrete case turns out to be significantly easier to analyse, and in this situation it was found that a rich array of anyon statistics are kinematically possible.   Specifically, certain graphs were found to support anyons while others can only support fermions or bosons.  This was demonstrated by analysing the topology of the corresponding configuration graphs $C_2(\Gamma)=(\Gamma^{\times 2}-\Delta)/S_{2}$ in various examples.  It opens up the problem of determining general relations between the quantum statistics of a graph and its topology.

As noted above, mathematically the determination of quantum statistics reduces to finding the first homology group $H_{1}$ of the appropriate classical configuration space, $C_n(M)$.  Although the calculation for $C_n(\mathbb{R}^N)$ is relatively elementary, it becomes a non-trivial task when $\mathbb{R}^{N}$ is replaced by a general graph $\Gamma$. One possible route  is to use discrete Morse theory, as developed by Forman \cite{Forman98}. This is a combinatorial counterpart of classical Morse theory, which applies to cell complexes. In essence, it reduces the problem of finding $H_{1}(M)$, where $M$ is a cell complex, to the construction of certain discrete Morse functions, or equivalently discrete gradient vector fields. Following this line of reasoning
Farley and Sabalka \cite{FS05} defined the appropriate discrete vector fields and gave a formula for the first homology groups of tree graphs.  Recently, making extensive use of discrete Morse theory and some graph invariants, Ko and Park \cite{KP11} extended the results of \cite{FS05} to an arbitrary graph $\Gamma$.  However, their approach relies on a suite of relatively elaborate techniques -- mostly connected to a proper ordering of vertices and choices of trees to reduce the number of critical cells -- and the relationship to, and consequences for, the physics of quantum statistics are not easily identified.

In the current paper we give a full characterization of all possible abelian quantum statistics on graphs. In order to achieve this we develop a new set of ideas and methods which lead to an alternative proof of the structure theorem for the first homology group of the $n$-particle configuration space obtained by Ko and Park \cite{KP11}. Our reasoning, which is more elementary in that it makes minimal use of discrete Morse theory, is based on a set of simple combinatorial relations which stem from the analysis of some canonical small graphs.   The advantage for us of this approach is that it is explicit and direct.  This makes the essential physical ideas much more transparent and so enables us to identify the key topological determinants of the quantum statistics.  It also enables us to develop some further physical consequences. In particular we give a full characterization of the topological gauge potentials on $2$-connected graphs, and to identify some examples of particular physical interest, in which the quantum statistics have features that are subtle.

The paper is organized as follows. We start with a discussion, in section
\ref{sec:Examples}, of some physically interesting examples of quantum
statistics on graphs, in order to motivate the general theory that follows. In section \ref{sec:Graph-configuration-spaces}
we define some basic properties of graph configuration
spaces. In section \ref{sec:Two-particle-quantum-statistics}
we develop a full characterization of the first homology group for $2$-particle graph
configuration spaces. In section \ref{sec:N-particle-statistics-for}
we give a simple argument for the stabilization of quantum statistics
with respect to the number of particles for $2$-connected graphs.
Using this we obtain the desired result for $n$-particle graph configuration
spaces when $\Gamma$ is $2$-connected. In order to generalize the
result to $1$-connected graphs we consider star and fan graphs.
The  main result is obtained at the end of section \ref{sec:N-particle-statistics-on}.
The last part of the paper is devoted to the characterization
of topological gauge potentials for $2$-connected graphs.

\section{Quantum statistics on graphs\label{sec:Examples}}

In this section we discuss several examples which illustrate  some interesting and surprising aspects of quantum statistics on graphs.  A determining factor turns out to be the {\it connectivity} of a graph.  We recall (cf \cite{tutte01}) that a graph is {\it $k$-connected} if it remains connected after removing any  $k-1$ vertices.
According to Menger's theorem \cite{tutte01}, a graph is $k$-connected if and only if every pair of distinct vertices can be joined by at least $k$ disjoint paths.
A $k$-connected graph can be decomposed into $(k+1)$-connected components, unless it is complete \cite{Holberg92}.  Thus, a graph may be regarded as being built out of more highly connected components.
Quantum statistics, as we shall see, depends on $k$-connectedness up to $k = 3$.

\subsection{$3$-connected graphs}

Quantum statistics for a $3$-connected graph depends only on whether the graph is planar, and not on any additional structure.
 We recall that a graph is planar if it can be drawn in the plane without  crossings.
For planar $3$-connected graphs we will show that the statistics is characterised by a single anyon phase associated with cycles in which a pair of particles exchange positions.  For non-planar $3$-connected graphs, the statistics is either Bose or Fermi -- in effect, the anyon phase is restricted to be $0$ and $\pi$.  Thus, as far as quantum statistics is concerned, three- and higher-connected graphs behave like $\bbR^2$ in the planar case and $\bbR^d$, $d > 2$, in the nonplanar case. A new aspect for graphs is the possibility of combining planar and nonplanar components. The graph shown in  figure~\ref{fig:The-large-almost} consists of a large square lattice in which four cells have been replaced by a defect in the form of a $K_5$ subgraph, the (nonplanar) fully connected graph on five vertices. This local substitution makes the full graph nonplanar, thereby excluding anyon statistics.
\begin{figure}[h]
\begin{center}\includegraphics[scale=0.5]{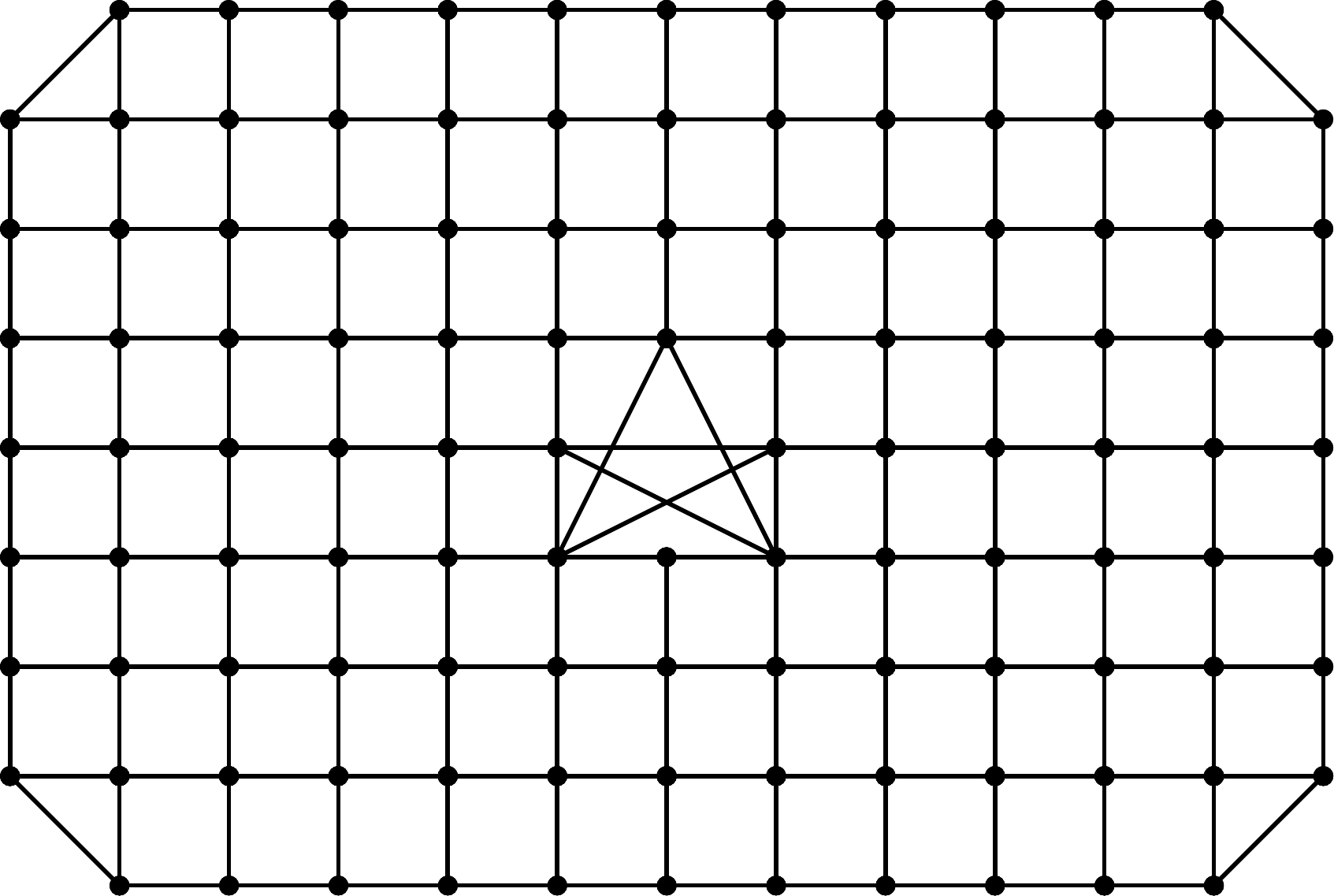}\end{center}
\caption{\label{fig:The-large-almost}The large almost planar $3$-connected
graph. }
\end{figure}

One of the simplest examples of this phenomenon is provided by the graph $G$ shown in figure~\ref{fig: GK5}. $G$ is planar $3$-connected, and therefore supports an anyon phase.  However, if  an additional edge $e$ is added, the resulting graph is $K_5$, and therefore supports only Bose or Fermi statistics.  One can continuously interpolate from a quantum Hamiltonian defined on $K_5$ to one defined by $G$ by introducing  an amplitude coefficient $\epsilon$ for transitions along $e$.  For $\epsilon = 0$, the edge $e$ is effectively absent, and the resulting Hamiltonian is  defined on $G$.  This situation might appear to be paradoxical; how could anyon statistics, well defined for $\epsilon = 0$, suddenly disappear for $\epsilon \neq 0$? The resolution lies in the fact that an anyon phase defined for $\epsilon = 0$ introduces, for $\epsilon \neq 0$, physical effects that cannot be attributed to quantum statistics (unless the phase is $0$ or $\pi$).  The transition between planar and nonplanar geometries, which is easily effected with quantum graphs, merits further study.
\begin{figure}[H]
\begin{center}\includegraphics[scale=0.15]{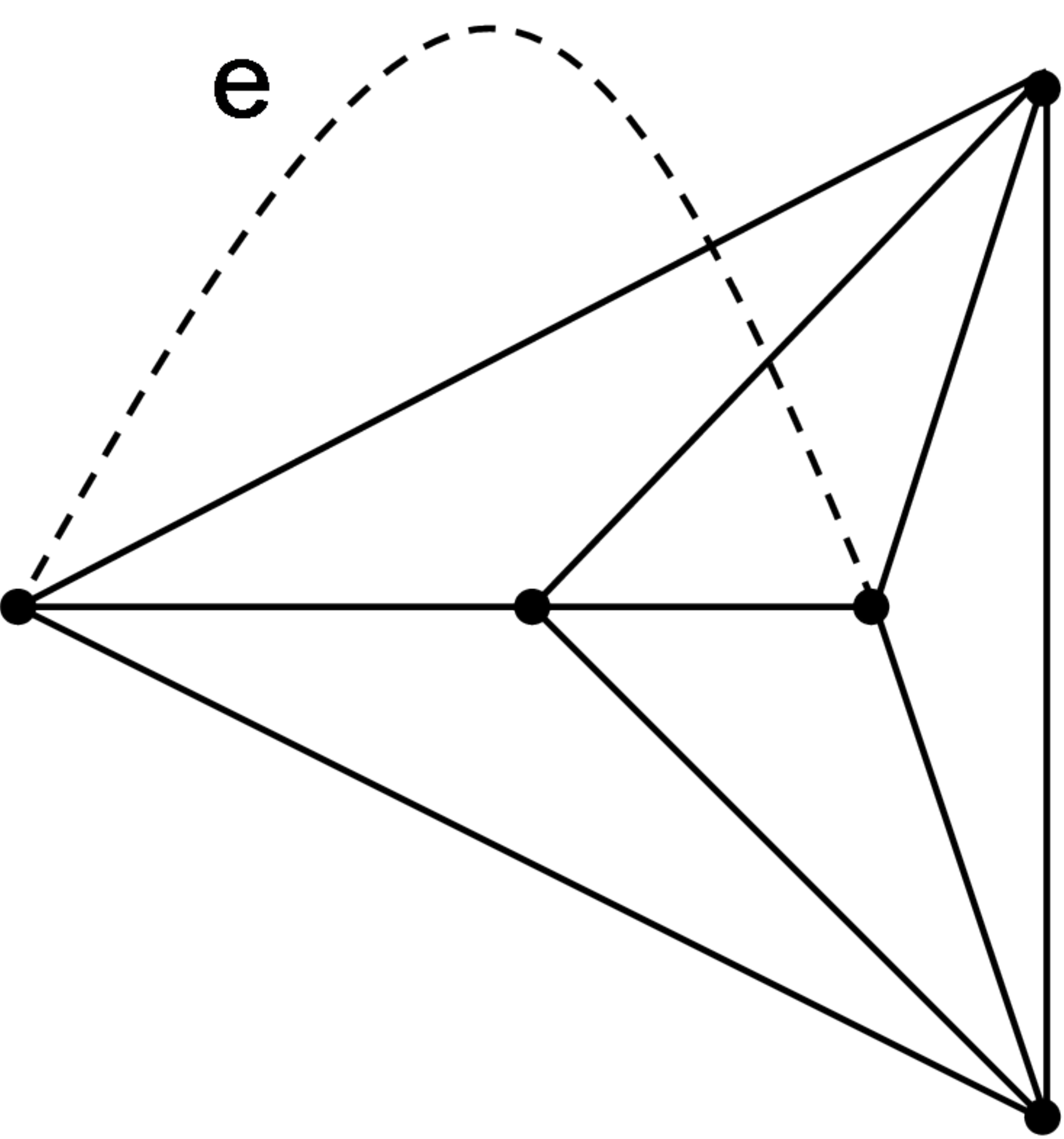}\end{center}
\caption{\label{fig: GK5}The graph $G$ (without  the edge $e$) is planar $3$-connected.  With $e$, the graph is $K_5$. }
\end{figure}

\subsection{$2$-connected graphs}

Quantum statistics on $2$-connected graphs is more complex, and depends on the decomposition of individual  graphs into cycles and $3$-connected components (see Section~\ref{subsec: two-connected}).
There may be multiple anyon and $\bbZ_2$ phases. But $2$-connected graphs share the following important property: their quantum statistics do not depend on the number of particles, and therefore can be regarded as a characteristic of the particle species.  This property is important physically; it means that there is a building-up principle for increasing the number of particles in the system.  This is described in detail in Section~\ref{sec: gauge potential},
 where we show how to construct an $n$-particle Hamiltonian from a two-particle Hamiltonian. Interesting examples are also obtained by building $2$-connected graphs out of higher-connected components.  Figure~\ref{fig: BF_chain} shows a chain of identical non-planar 3-connected components. The links between components, represented by lines in figure~\ref{fig: BF_chain}, consist of at least two edges, so that resulting graph is $2$-connected. In this case, the quantum statistics is in fact independent of the number of particles, and may be determined by specifying exchange phases ($0$ or $\pi$) for each component in the chain.  Thus, particles can act as bosons or fermions in different parts of the graph.
\begin{figure}[h]
\begin{center}\includegraphics[scale=0.25]{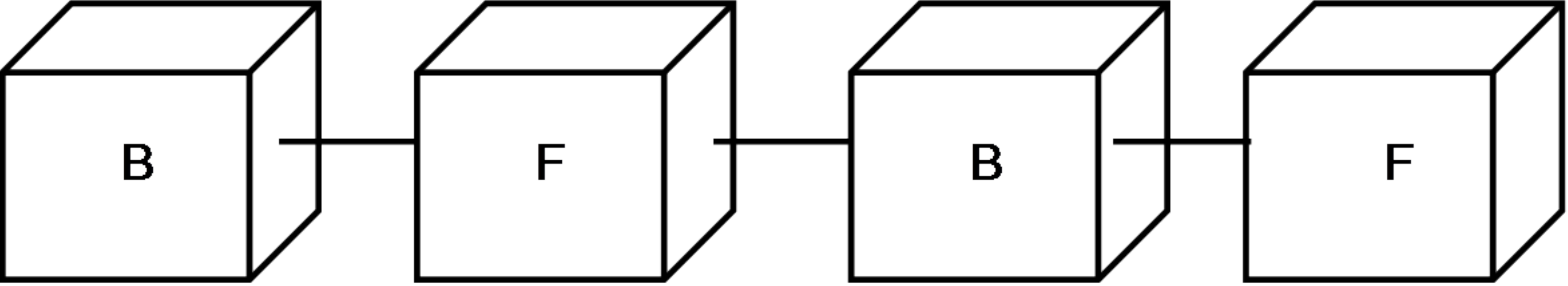}\end{center}
\caption{\label{fig: BF_chain} Linear chain of $3$-connected nonplanar components with alternating Bose and Fermi statistics.}
\end{figure}

\subsection{$1$-connected graphs}

Quantum statistics on graphs achieves its full complexity for 1-connected graphs, in which case it also depends on the number of particles $n$.  A representative example, treated in detail in Section~\ref{sub:The-star-graphs},
is a star graph with $E$ edges, for which the number of anyon phases is given by
\begin{gather*}
\beta_{n}^{E}={n+E-2 \choose E-1}\left(E-2\right)-{n+E-2 \choose E-2}+1,
\end{gather*}
and therefore depends on both $E$ and $n$.

\subsection{Aharonov-Bohm phases}

Configuration-space cycles on which one particle moves around a circuit $C$ while the others remain fixed play an important role in the analysis of quantum statistics which follows.  We call these Aharonov-Bohm cycles, and the corresponding phases Aharonov-Bohm phases, because they correspond physically to magnetic fluxes threading $C$.  In many-body systems, Aharonov-Bohm phases and quantum statistics phases can interact in interesting ways.  In particular, Aharonov-Bohm phases can depend on the positions of the stationary particles.  An example is shown in the two-particle  octahedron graph (see figure~\ref{fig: oct_AB}), in which the  Aharonov-Bohm phase associated with one particle going around the equator depends on whether the second particle is at the north or south pole.  For $3$-connected non-planar graphs, it can be shown that Aharonov-Bohm phases are independent of the positions of the stationery particles.  (The octahedron graph, despite appearances, is planar.)
\begin{figure}[h]
\begin{center}~~~~~\includegraphics[scale=0.2]{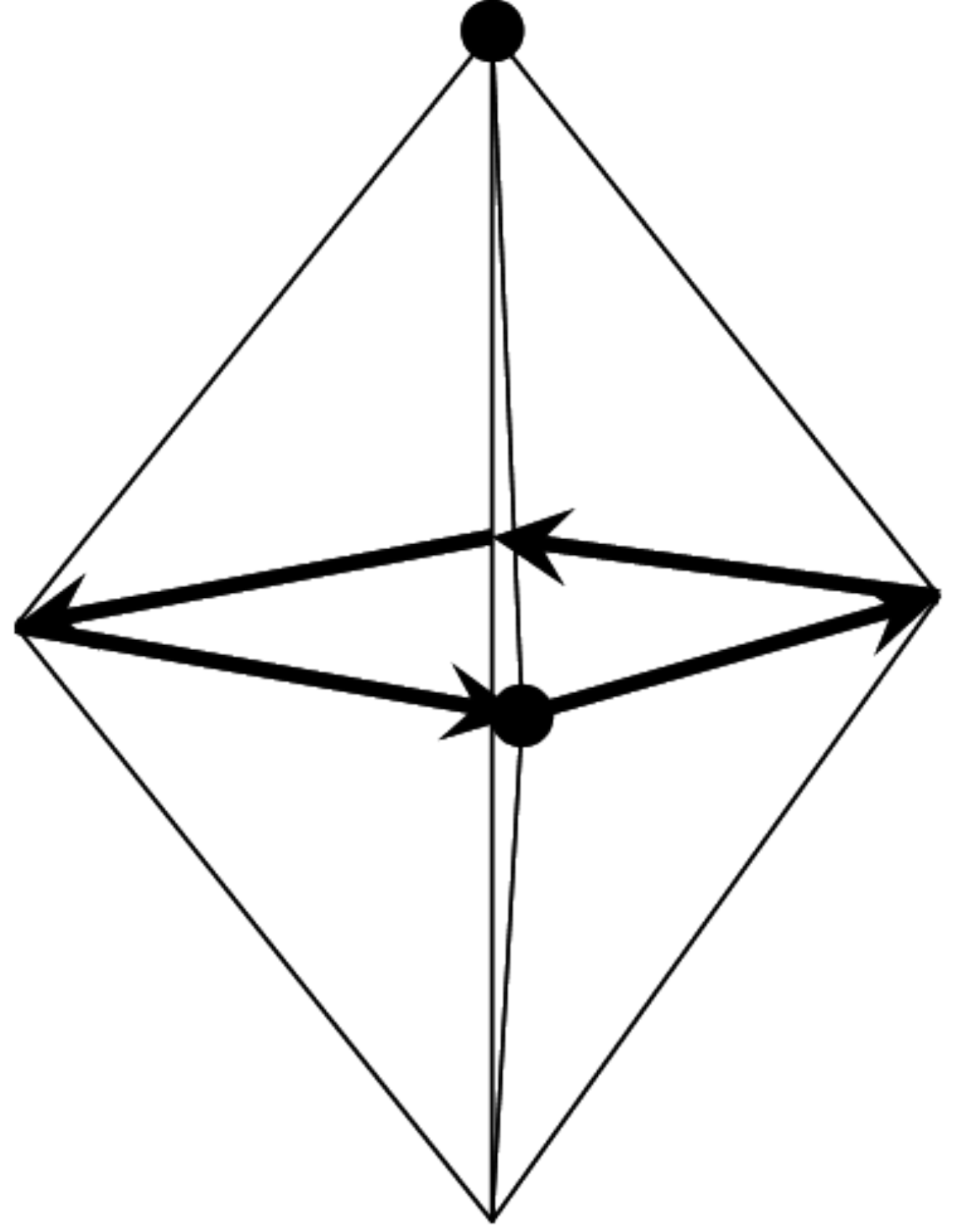}~~~~~~~~~~~~~~~~~~\includegraphics[scale=0.2]{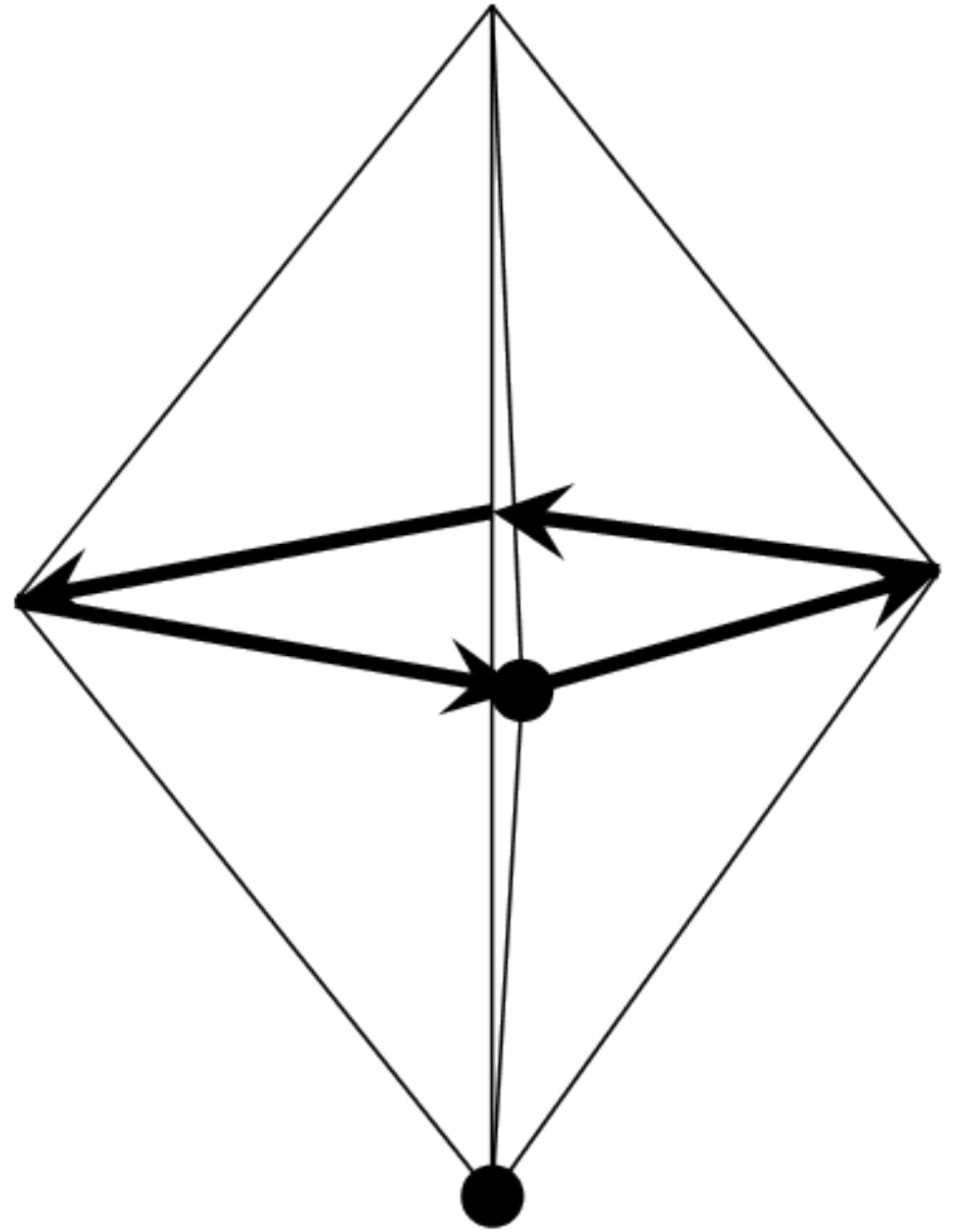}\end{center}
\caption{\label{fig: oct_AB} The Aharonov-Bohm phase for the equatorial cycle depends on whether the second particle is at the north or south pole.}
\end{figure}

\section{Graph configuration spaces\label{sec:Graph-configuration-spaces}}

Let $\Gamma$ be a metric connected simple graph with $V$
vertices and $E$ edges. In a metric graph edges correspond to finite closed intervals of $\mathbb{R}$.  However, as we will be interested in the topology of the graph, the length of the edges will not play a role in the discussion.
An undirected edge between vertices $v_{1}$
and $v_{2}$ will be denoted by $v_{1}\leftrightarrow v_{2}$.  It will also be convenient to be able to label directed edges
so $v_{1}\rightarrow v_{2}$ and $v_{2}\rightarrow v_{1}$ are the directed edges associated with $v_{1}\leftrightarrow v_{2}$.
A path joining two vertices $v_1$ and $v_m$ is then specified by a sequence of $m-1$ directed edges, written $v_1 \rightarrow v_2 \rightarrow \dots \rightarrow v_m$.

We define the \emph{$n$-particle configuration space} as the quotient space
\begin{eqnarray}
C_{n}(\Gamma)=\left(\Gamma^{\times n}-\Delta\right)/S_{n},
\end{eqnarray}
where $S_{n}$ is the permutation group of $n$ elements and
\begin{eqnarray}
\Delta=\{(x_{1},x_{2},\ldots,x_{n}):\exists_{i,j}\, x_{i}=x_{j}\},
\end{eqnarray}
is the set of coincident configurations.
We are interested in the calculation of the first homology group,
$H_{1}(C_{n}(\Gamma))$ of $C_{n}(\Gamma)$. The space $C_{n}(\Gamma)$
is not a cell complex. However, it is homotopy equivalent to the space
$\mathcal{D}^{n}(\Gamma)$, which is a cell complex, defined below.

Recall that a cell complex $X$ is a nested sequence of topological spaces
\begin{eqnarray}
X^{0}\subseteq X^{1}\subseteq\dots\subseteq X^{n},
\end{eqnarray}
where the $X^{k}$'s are the so-called $k$-skeletons defined as follows:
\begin{itemize}
\item The $0$ - skeleton $X^{0}$ is a finite set of points.
\item For $\mathbb{N}\ni k>0$, the $k$ - skeleton $X^{k}$ is the result
of attaching $k$ - dimensional balls $B_{k} = \{x\in\mathbb{R}^{k}\,:\,\|x\|\leq1\}$ to $X^{k-1}$ by gluing
maps
\begin{eqnarray}
\sigma:S^{k-1}\rightarrow X^{k-1},
\end{eqnarray}
where $S^{k-1}$ is the unit-sphere $S^{k-1}=\{x\in\mathbb{R}^{k}\,:\,\|x\|=1\}$.
\end{itemize}
A $k$-cell is the interior of the ball $B_{k}$ attached to the $(k-1)$-skeleton $X^{k-1}$.

Every simple graph $\Gamma$ is naturally a cell complex; the vertices are $0$-cells (points)
and edges are $1$-cells ($1$-dimensional balls whose boundaries are the $0$-cells). The product $\Gamma^{\times n}$ then naturally
inherits a cell complex structure. The cells of $\Gamma^{n}$ are
Cartesian products of cells of $\Gamma$. It is clear
that the space $C_{n}(\Gamma)$ is not a cell complex as
points belonging to $\Delta$ have been deleted. Following \cite{Abrams}
we define an \emph{$n$-particle combinatorial configuration space} as
\begin{eqnarray}
\mathcal{D}^{n}(\Gamma)=(\Gamma^{\times n}-\tilde{\Delta})/S_{n},
\end{eqnarray}
where $\tilde{\Delta}$ denotes all cells whose closure intersects
with $\Delta$. The space $\mathcal{D}^{n}(\Gamma)$ possesses a natural
cell complex structure. Moreover,
\begin{theorem}
\label{Abrams_thm}\cite{Abrams} For any graph $\Gamma$ with at
least $n$ vertices, the inclusion $\mathcal{D}^{n}(\Gamma)\hookrightarrow C_{n}(\Gamma)$
is a homotopy equivalence iff the following hold:
\begin{enumerate}
\item Each path between distinct vertices of valence not equal to two passes
through at least $n-1$ edges.
\item Each closed path in $\Gamma$ passes through at least $n+1$ edges.
\end{enumerate}
\end{theorem}
\noindent Following \cite{Abrams,FS05} we refer to a graph $\Gamma$ with properties 1 and 2 as \emph{sufficiently subdivided}. For $n=2$ these conditions are automatically satisfied
(provided $\Gamma$ is simple). Intuitively, they can be understood
as follows:
\begin{enumerate}
\item In order to have homotopy equivalence between $\mathcal{D}^{n}(\Gamma)$
and $C_{n}(\Gamma)$, we need to be able to accommodate $n$ particles
on every edge of graph $\Gamma$. This is done by introducing $n-2$ trivial vertices of degree $2$ to make a line subgraph between every adjacent pair of non-trivial vertices in the original graph $\Gamma$.
\item For every cycle there is at least one free (not occupied) vertex which
enables the exchange of particles around this cycle.
\end{enumerate}
For a sufficiently subdivided graph $\Gamma$ we can now effectively treat $\Gamma$ as a combinatorial graph where particles are accommodated at vertices and hop between adjacent unoccupied vertices along edges of $\Gamma$.

Using Theorem \ref{Abrams_thm}, the problem of finding $H_{1}(C_{n}(\Gamma))$
is reduced to the problem of computing $H_{1}(\mathcal{D}^{n}(\Gamma))$.
In the next sections we show how to determine $H_{1}(\mathcal{D}^{n}(\Gamma))$
for an arbitrary simple graph $\Gamma$. Note, however, that by the
structure theorem for finitely generated modules \cite{Nakahara}

\begin{gather}
H_{1}(\mathcal{D}^{n}(\Gamma))=\mathbb{Z}^{k}\oplus T_{l}\label{eq:finite-module}
\end{gather}
where $T_{l}$ is the torsion, i.e.
\begin{gather}
T_{l}=\mathbb{Z}_{n_1}\oplus\ldots\oplus\mathbb{Z}_{n_l},\label{eq:finite-module-torsion}
\end{gather}
and $n_{i}|n_{i+1}$. In other words $H_{1}(\mathcal{D}^{n}(\Gamma))$
is determined by $k$ free parameters $\{\phi_{1},\ldots,\phi_{k}\}$
and $l$ discrete parameters $\{\psi_{1},\ldots,\psi_{l}\}$ such that for
each $i\in\{1,\ldots l\}$

\begin{gather}
n_{i}\psi_{i}=0\,\,\mbox{mod}\,2\pi,\,\,n_i\in \mathbb{N}\,\,\,\:\mbox{and}\,\,n_{i}|n_{i+1}.\label{eq:finite-module-1}
\end{gather}
Taking into account their physical interpretation we will call the
parameters $\phi$ and $\psi$ continuous and discrete phases respectively.

\section{Two-particle quantum statistics\label{sec:Two-particle-quantum-statistics}}

In this section we fully describe the first homology group $H_{1}(\mathcal{D}^{2}(\Gamma))$
for an arbitrary connected simple graph $\Gamma$. We start with three simple examples: a cycle, a Y-graph and a lasso. The $2$-particle discrete configuration
space of the lasso reveals an important relation between the exchange
phase on the Y-graph and on the cycle. Combining this relation with
an ansatz for a perhaps over-complete spanning set of the cycle space of
$\mathcal{D}^{2}(\Gamma)$ and some combinatorial properties of $k$-connected
graphs, we give a formula for $H_{1}(\mathcal{D}^{2}(\Gamma))$. Our
argument is divided into three parts; corresponding to $3$-, $2$- and $1$-connected graphs respectively.

\paragraph{Three examples }
\begin{itemize}
\item Let $\Gamma$ be a triangle graph shown in figure \ref{fig:The-triangle}(a).
Its combinatorial configuration space $\mathcal{D}^{2}(\Gamma)$ is shown
in figure 1(b). The cycle $(1,2)\rightarrow(1,3)\rightarrow(2,3)\rightarrow(1,2)$
is not contractible and hence $H_{1}(\mathcal{D}^{2}(\Gamma))=\mathbb{Z}$.
In other words we have one free phase $\phi_{c}$ and no torsion.
\end{itemize}
\begin{figure}[h]
\begin{center}\includegraphics[scale=0.45]{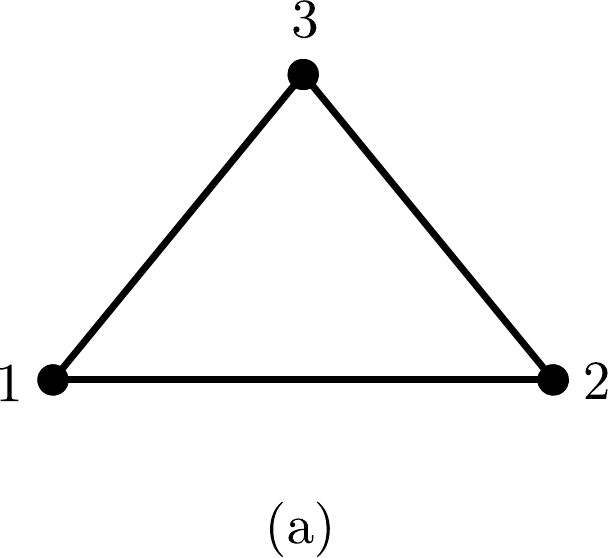}~~~~~~~~~~~~~~~~~~~~~~~~~~~~\includegraphics[scale=0.45]{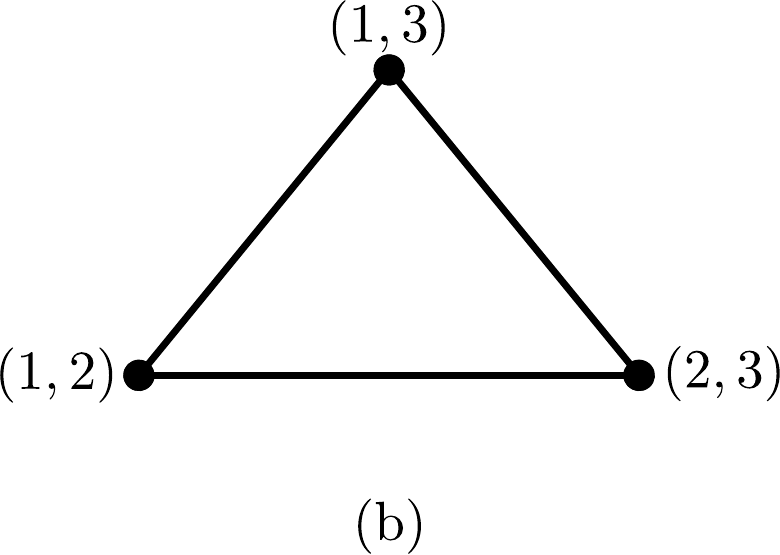}\end{center}

\caption{\label{fig:The-triangle}(a) The triangle graph $\Gamma$ (b) The
$2$-particle configuration space $\mathcal{D}^{2}(\Gamma)$}
\end{figure}

\begin{itemize}
\item Let $\Gamma$ be a Y-graph shown in figure \ref{fig:(a)-The-Y}(a).
Its combinatorial configuration space $\mathcal{D}^{2}(\Gamma)$ is shown
in figure 2(b). The cycle $(1,2)\rightarrow(1,3)\rightarrow(2,3)\rightarrow(3,4)\rightarrow(2,4)\rightarrow(1,4)\rightarrow(1,2)$
is not contractible and $H_{1}(\mathcal{D}^{2}(\Gamma))=\mathbb{Z}$.
Hence we have one free phase $\phi_{Y}$ and no torsion.
\end{itemize}
\begin{figure}[h]
\begin{center}~~~~~~~~~~~~~~~~\includegraphics[scale=0.4]{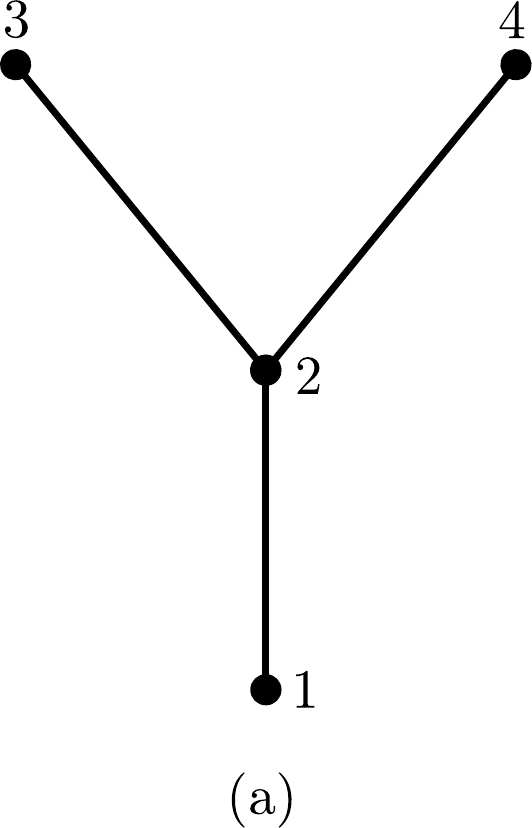}~~~~~~~~~~~~~~~~~~\includegraphics[scale=0.4]{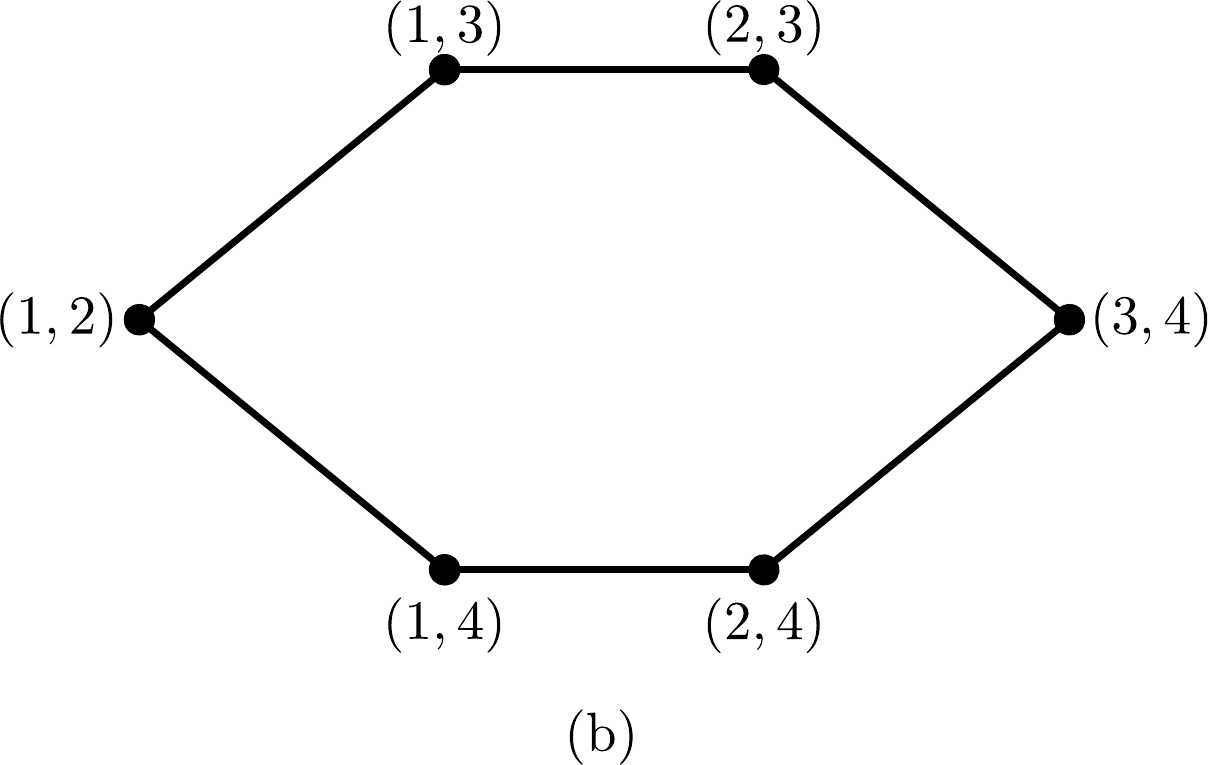}\end{center}

\caption{\label{fig:(a)-The-Y}(a) The Y-graph $\Gamma$ (b) The $2$-particle
configuration space $\mathcal{D}^{2}(\Gamma)$}
\end{figure}

\begin{itemize}
\item Let $\Gamma$ be a lasso graph shown in figure \ref{fig:The-lasso}(a).
It is a combination of Y and triangle graphs. Its combinatorial configuration
space $\mathcal{D}^{2}(\Gamma)$ is shown in figure 3(b). The shaded
rectangle is a $2$-cell and hence the cycle $(1,3)\rightarrow(2,3)\rightarrow(2,4)\rightarrow(1,4)\rightarrow(1,3)$
is contractible. The cycle $(1,2)\rightarrow(1,3)\rightarrow(1,4)\rightarrow(1,2)$
corresponds to the situation when one particle is sitting at the vertex
$1$ and the other moves along the cycle $c=2\rightarrow3\rightarrow4\rightarrow2$
of $\Gamma$. We will call this cycle an Aharonov-Bohm cycle (AB-cycle) and
denote its phase $\phi_{c,1}^1$. The cycle $(2,3)\rightarrow(3,4)\rightarrow(2,4)\rightarrow(2,3)$
represents the exchange of two particles around $c$. The corresponding phase will be denoted by $\phi_{c,2}$. Finally, for the cycle
$(1,2)\rightarrow(1,3)\rightarrow(2,3)\rightarrow(3,4)\rightarrow(2,4)\rightarrow(1,4)\rightarrow(1,2)$, corresponding to
exchange of two particles along a Y-graph, the phase is $\phi_Y$. There is no torsion in $H_{1}(\mathcal{D}^{2}(\Gamma))$.
Moreover,
\begin{gather}
\phi_{c,2}=\phi_{c,1}^{1}+\phi_{Y}.\label{eq:lasso-relation}
\end{gather}
Notice that knowing $\phi_{Y}$ and the AB-phases determines the phases $\phi_{c,2}$. As we shall see, \ref{eq:lasso-relation} plays an important role in relating Y-phases and AB-phases for general graphs.
\end{itemize}
\begin{figure}[h]
\begin{center}~~~~~~~~~~~~~~\includegraphics[scale=0.4]{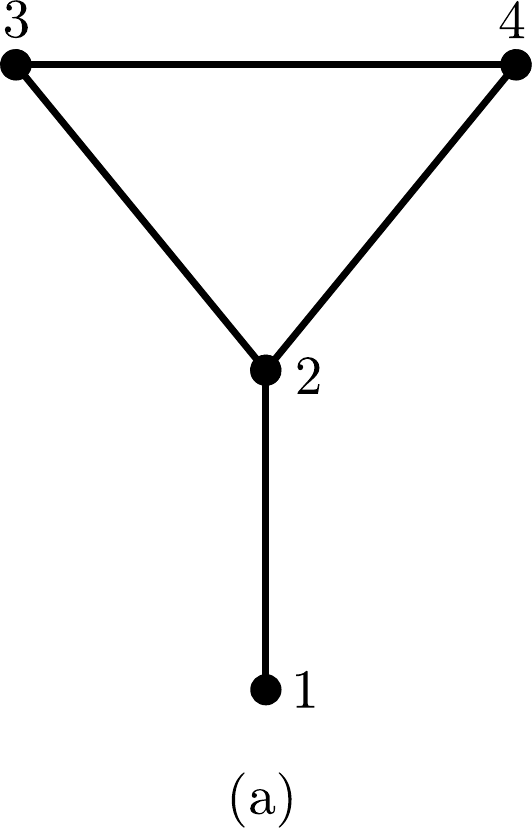}~~~~~~~~~~~~~~~~~~\includegraphics[scale=0.4]{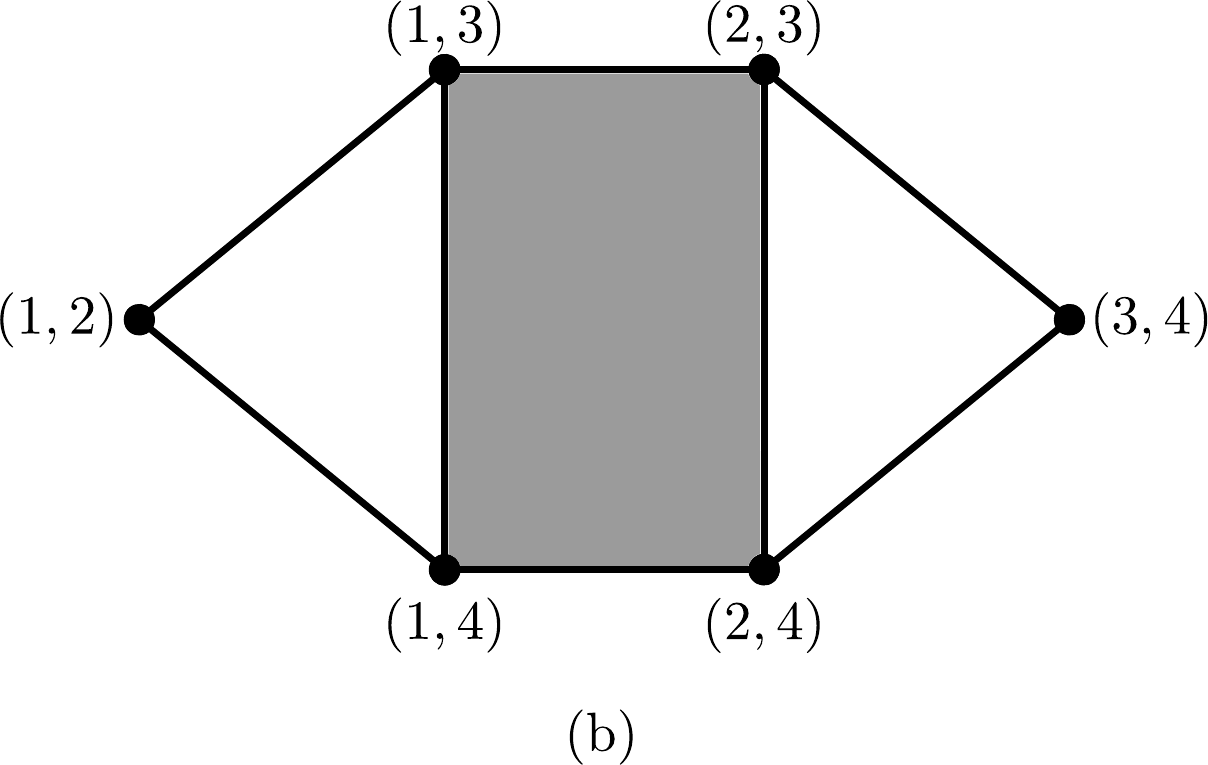}\end{center}

\caption{\label{fig:The-lasso}(a) The lasso graph $\Gamma$ (b) The $2$-particle
configuration space $\mathcal{D}^{2}(\Gamma)$}
\end{figure}

\subsection{A spanning set of $H_1(\mathcal{D}^{2}(\Gamma))$\label{sub:An-over-complete-basis}}

In order to proceed with the calculation of $H_{1}(\mathcal{D}^{2}(\Gamma))$ for arbitrary $\Gamma$
we need a spanning set of $H_{1}(\mathcal{D}^{2}(\Gamma))$. Before we give
one, let us discuss the dependence of the AB-phase on the position
of the second particle. Suppose there is a cycle $c$ in $\Gamma$ with two vertices $v_{1}$
and $v_{2}$ not on the cycle. We want to know the relation between $\phi_{c,1}^{v_{1}}$ and $\phi_{c,1}^{v_{2}}$.
There are two possibilities to consider. The first is shown in
figure \ref{fig:The-dependence-of}(a) and represents the situation
when there is a path $P_{v_{1},v_{2}}$ which joins $v_{1}$ and $v_{2}$
and is disjoint with $c$. In this case both AB-cycles are homotopy
equivalent as they belong to the cylinder $c\times P_{v_{1},v_{2}}$. Therefore,
\begin{Fact}\label{fact1} Assume there is a cycle $c$ in $\Gamma$ with two vertices $v_{1}$
and $v_{2}$ not on the cycle. Suppose there is a path $P_{v_{1},v_{2}}$ which joins $v_{1}$ and $v_{2}$
and is disjoint with $c$. Then $\phi_{c,1}^{v_{1}}=\phi_{c,1}^{v_{2}}$.
\end{Fact}
Assume now that there is a path joining $v_1$ and $v_2$ which passes through the cycle $c$ (see figure \ref{fig:The-dependence-of}(b)). Using
relation (\ref{eq:lasso-relation}) we get

\begin{gather}
\phi_{c,2}=\phi_{c,1}^{v_{1}}+\phi_{Y_{1}},\,\,\,\,\phi_{c,2}=\phi_{c,1}^{v_{2}}+\phi_{Y_{2}},\label{eq:A-B-1}
\end{gather}
and hence
\begin{gather}
\phi_{c,1}^{v_{1}}-\phi_{c,1}^{v_{2}}=\phi_{Y_{2}}-\phi_{Y_{1}}.\label{eq:AB-2}
\end{gather}
The relations between different AB-phases for a fixed cycle $c$
of $\Gamma$ are therefore encoded in the phases $\phi_{Y}$.

\begin{figure}[h]
\begin{center}~~~~~~\includegraphics[scale=0.5]{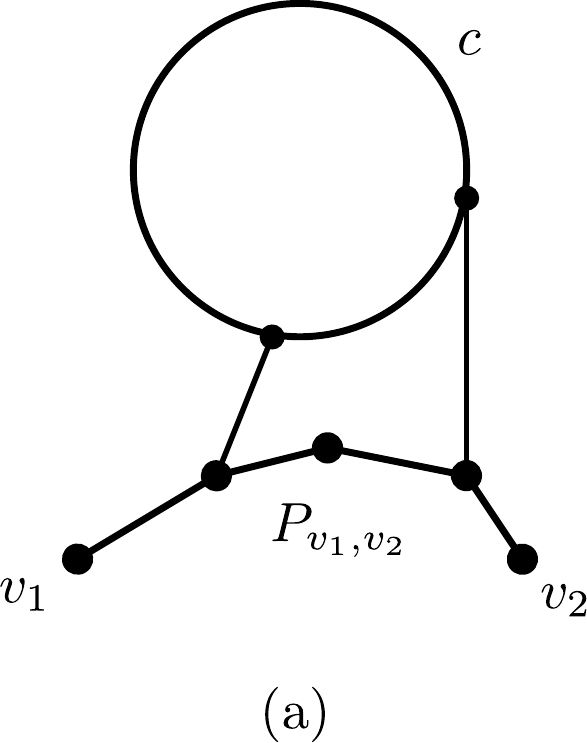}~~~~~~~~~~~~~~~~~~\includegraphics[scale=0.5]{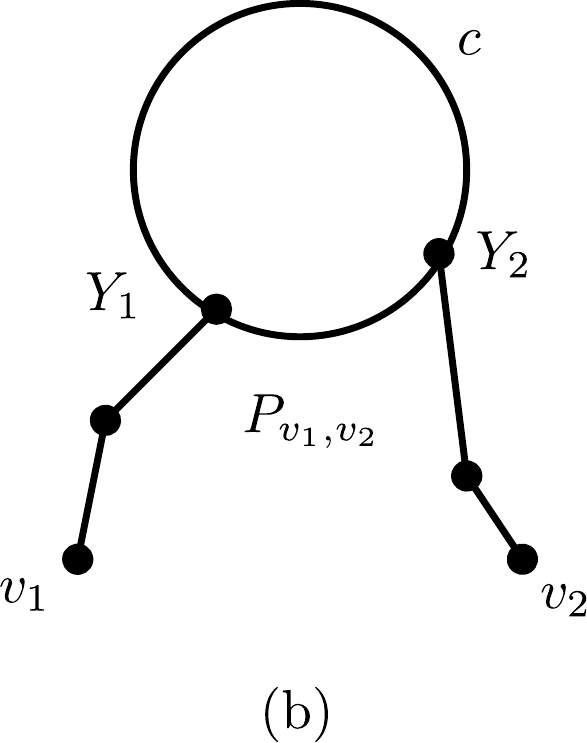}\end{center}

\caption{\label{fig:The-dependence-of}The dependence of the AB-phase for
cycle $c$ on the position of the second particle when (a) there is
a path between $v_{1}$ and $v_{2}$ disjoint with $c$, (b) every
path joining $v_{1}$ and $v_{2}$ passes through $c$.}
\end{figure}

As we show in the appendix, a spanning set of $H_{1}(\mathcal{D}^{2}(\Gamma))$ is given by all Y and AB-cycles. Note that
relations (\ref{eq:lasso-relation}) and (\ref{eq:AB-2}) reduce the number of relevant AB-cycles to one per independent cycle of $\Gamma$, i.e. to
\begin{gather}
\beta_{1}(\Gamma)=E-V+1,
\end{gather}
cycles. As a result, we will use a spanning set (which in general is over-complete) containing the following:
\begin{enumerate}
\item All $2$-particle cycles corresponding to the exchanges on Y subgraphs
of $\Gamma$. There can be dependencies between these cycles.
\item A set of $\beta_{1}(\Gamma)$ AB-cycles, one for each independent cycle in $\Gamma$.
\end{enumerate}
Thus, $H_{1}(\mathcal{D}^{2}(\Gamma))=\mathbb{Z}^{\beta_1(\Gamma)}\oplus A$, where $A$ is determined by Y-cycles. Consequently, in order to determine $H_{1}(\mathcal{D}^{2}(\Gamma))$ one has to
study the relations between Y-cycles.

\subsection{$3$-connected graphs}

In this section we determine $H_{1}(\mathcal{D}^{2}(\Gamma))$
for $3$-connected graphs. Let $\Gamma$ be a connected graph. We define an $m$-separation of
$\Gamma$ \cite{tutte01}, where $m$ is a positive integer, as an
ordered pair $(\Gamma_{1},\Gamma_{2})$ of subgraphs of $\Gamma$
such that
\begin{enumerate}
\item The union $\Gamma_{1}\cup\Gamma_{2}=\Gamma$.
\item $\Gamma_{1}$ and $\Gamma_{2}$ are edge-disjoint and have exactly
$m$ common vertices, $V_{m}=\{v_{1},\ldots,v_{m}\}$.
\item $\Gamma_{1}$ and $\Gamma_{2}$ have each a vertex not belonging to
the other.
\end{enumerate}
It is customary to say that the $V_{m}$ separates vertices of $\Gamma_{1}$
and $\Gamma_{2}$ different from $V_{m}$.
\begin{definition}
\label{n-connected-def}A connected graph $\Gamma$ is $n$-connected
iff it has no $m$-separation for any $m<n$.
\end{definition}
The following theorem of Menger \cite{tutte01} gives
an additional insight into graph connectivity:
\begin{theorem}
\label{Menger} For an $n$-connected graph $\Gamma$ there are at
least $n$ internally disjoint paths between any pair of vertices.
\end{theorem}
The basic examples of $3$-connected graphs are wheel graphs. A wheel
graph $W^{n}$ of order $n$ consists of a cycle with $n$ vertices
and a single additional vertex which is connected to each vertex of
the cycle by an edge. Following Tutte \cite{tutte01} we denote the
middle vertex by $h$ and call it the hub, and the cycle that does not include $h$ by $R$ and call
it the rim. The edges connecting the hub to the rim will be called
spokes. The importance of wheels in the theory of $3$-connected graphs
follows from the following theorem;
\begin{theorem}
\label{thm-Wheel-thorem}(Wheel theorem \cite{tutte01}) Let $\Gamma$ be a simple
$3$-connected graph different from a wheel. Then for some edge $e\in E(\Gamma)$
either $\Gamma\setminus e$ or $\Gamma/e$ is simple and
$3$-connected.
\end{theorem}
Here $\Gamma\setminus e$ is constructed from $\Gamma$ by removing the edge $e$, and $\Gamma/e$ is obtained by contracting edge $e$ and identifying its vertices. These two operations will be called edge removal and edge contraction. The inverses will be called edge addition and vertex expansion. Note that that vertex expansion requires specifying which edges are connected to which vertices after expansion. As we deal with $3$-connected graphs we will apply the vertex expansion only to vertices of degree at least four and split the edges between new vertices in a such way that they are at least $3$-valent.

As a direct corollary of Theorem \ref{thm-Wheel-thorem} any simple $3$-connected graph can be constructed in a finite number of steps
starting from a wheel graph $W^{k}$, for some $k$
\begin{gather*}
W_{k}=\Gamma_{0}\mapsto\Gamma_{1}\mapsto\ldots\mapsto\Gamma_{n-1}\mapsto\Gamma_{n}=\Gamma
\end{gather*}
where $\Gamma_{i}$ is constructed from $\Gamma_{i-1}$ by either
\begin{enumerate}
\item Adding an edge between non-adjacent vertices or
\item Expanding at the vertex of the valency at least four.
\end{enumerate}
Moreover, each $\Gamma_{i}$ is simple and $3$-connected. In order
to prove inductively some feature of a $3$-connected graph it is
therefore enough to show it for an arbitrary wheel and consider what
happens when an edge between two non-adjacent vertices is added or
a vertex of valency at least four is expanded.
\begin{lemma}
For wheel graphs $W^{n}$ all phases $\phi_{Y}$ are equal up to the
sign.
\end{lemma}
\begin{proof}
The Y subgraphs of $W^{n}$ can be divided into two groups: (i) the central vertex of Y is on the rim
(ii) the center vertex of Y is the hub. For (i) let $v_{1}$ and $v_{2}$ be two adjacent vertices belonging to the
rim, $R$. Let $Y_{v_{1}}$ and $Y_{v_{2}}$ be the corresponding
Y-graphs whose central vertices are $v_{1}$ and $v_{2}$ respectively
and one edge is a spoke. Evidently two edges of $Y_{v_{1}}$ and $Y_{v_{2}}$ belong
to the same triangle cycle, $C$ i.e the one spanned by $v_{1}$,
$v_{2}$ and the hub (see figure \ref{fig:wheel}(a)). Moreover, $b_{1}$
is connected to $b_{2}$ by a path which is disjoint with $\mbox{\ensuremath{C}}$. Using Fact \ref{fact1} and  relation (\ref{eq:AB-2})
we get $\phi_{Y_{v_{1}}}=\phi_{Y_{v_{2}}}$. Repeating this reasoning
we obtain that all $\phi_{Y_{v_{i}}}$, with $v_{i}$ belonging to
the rim are equal (perhaps up to the sign). We are left with the Y-graphs whose central vertex is the hub. Similarly (see figure \ref{fig:wheel}(b))
we take a cycle, $C$, with two edges belonging to the chosen Y. But
there is always a Y-graph with two edges belonging to $C$ and center to the rim.
Therefore, by Fact \ref{fact1} and  relation (\ref{eq:AB-2}) the phase on the
Y subgraph whose center vertex is the hub is the same as on the Y subgraphs
whose center vertex is on the rim.\qed
\end{proof}
\begin{figure}[h]
\begin{center}~~~~~~~~~~\includegraphics[scale=0.45]{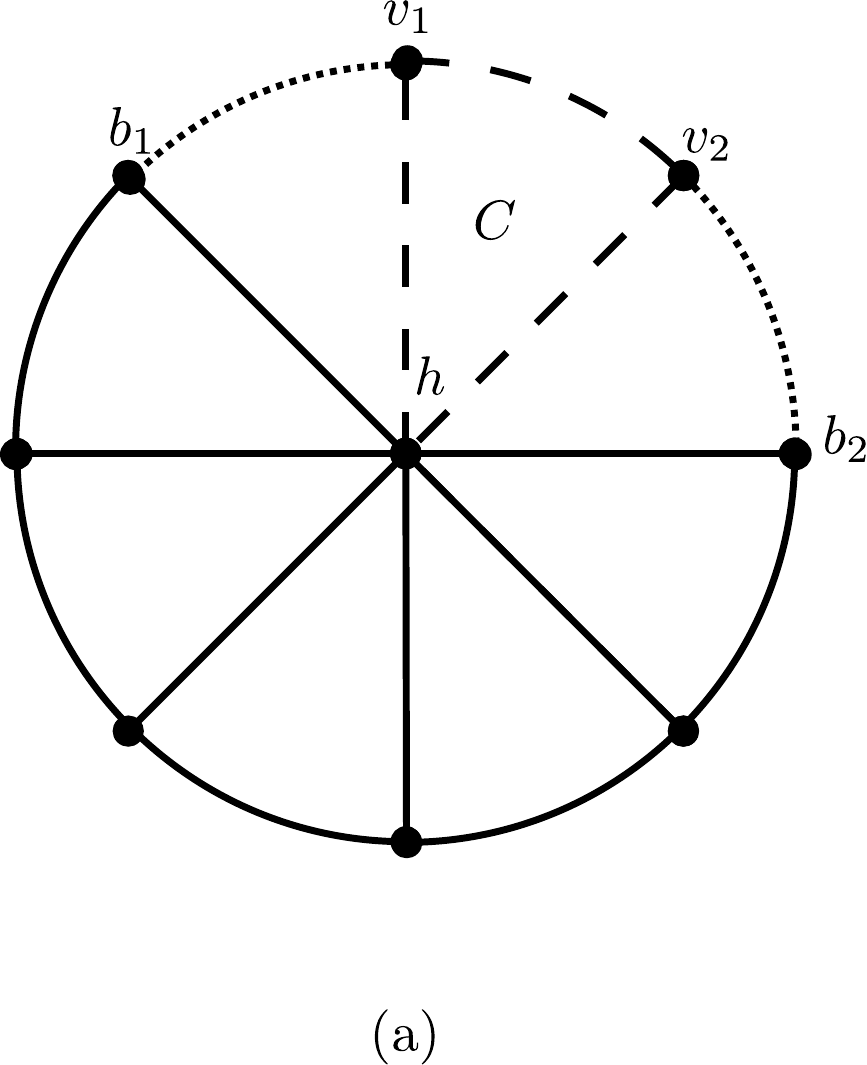}~~~~~~~~~~~~~~~\includegraphics[scale=0.45]{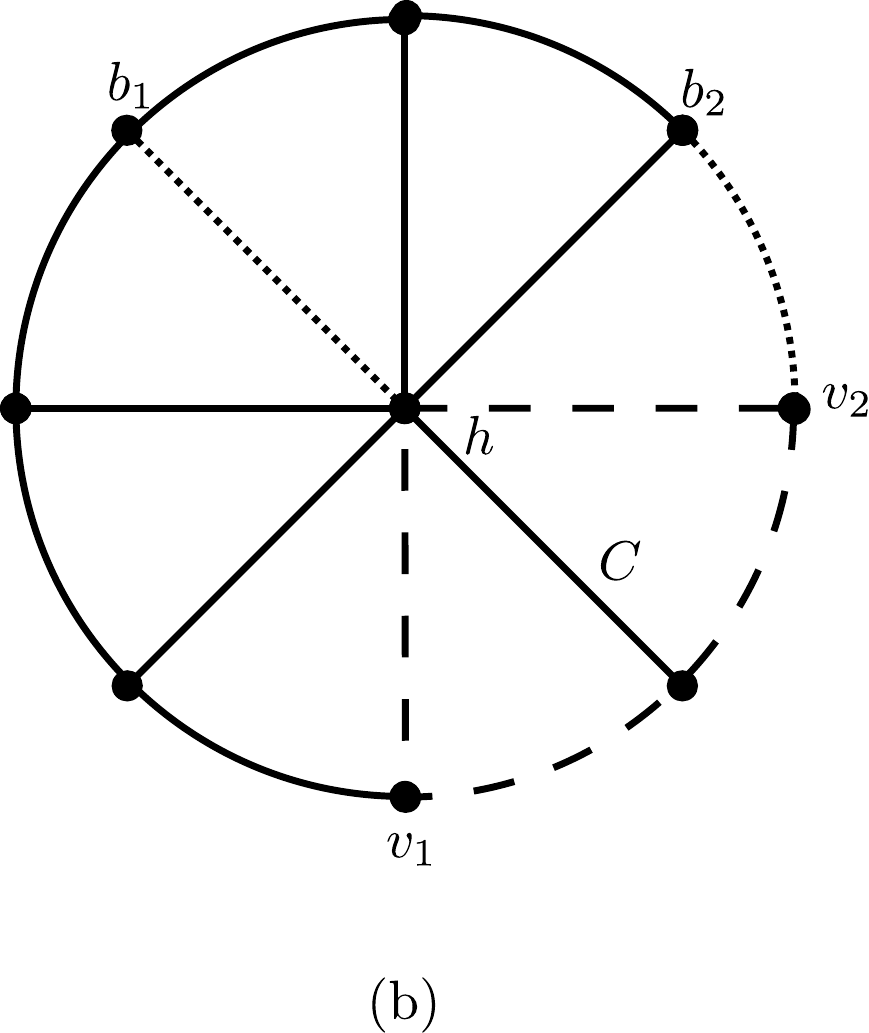}\end{center}

\caption{\label{fig:wheel}(a) Graphs $Y_{v_{1}}$ and $Y_{v_{2}}$(b) Graphs
$Y_{h}$ and $Y_{v_{2}}$ }
\end{figure}

\begin{lemma}
For $3$-connected simple graphs all phases $\phi_{Y}$ are equal
up to the sign.\end{lemma}
\begin{proof}
We prove by induction. By Lemma 1 the statement is true for all
wheel graphs.

Adding an edge: Assume that $v_{1}$ and $v_{2}$ are non-adjacent
vertices of the $3$-connected graph $\Gamma$ and all $\phi_{Y}$
phases are equal (up to the sign). By adding an edge between the vertices
$v_{1}$ and $v_{2}$ we do not change the relations between the $\phi_{Y}$
phases of $\Gamma$. However, the graph $\Gamma\cup e$ contains
new Y-graphs, whose middle vertices are $v_{1}$ or $v_{2}$ and
one of the edges is $e$. We need to show that the phase $\phi_{Y}$
for these new Y's is the same as on the old ones. Let $\{e,f_{1},f_{2}\}$
be such a Y-graph (see figure \ref{fig:Adding-an-edge}(a)). Let $\alpha_1$ and $\alpha_2$ be endpoints of $f_1$ and $f_2$. By $3$-connectedness, there is a path between $\alpha_1$ and $\alpha_2$ which does not contain $v_1$ or $v_2$.  In this way we obtain a cycle $C$, as shown in figure \ref{fig:Adding-an-edge}(a). Again by $3$-connectedness, there is a path $P$ from $v_2$ to a vertex $\beta$ in $C$. Let $Y'$ be the Y-graph with $\beta$ as its centre and edges along $C$ and $P$, as shown in figure \ref{fig:Adding-an-edge}(a). Then $Y'$ belongs to $\Gamma$. Applying Fact \ref{fact1} and  relation (\ref{eq:AB-2}) to the cycle $C$ and
the two Y-graphs discussed, the result follows.


Vertex expansion: Let $\Gamma$ be a $3$-connected simple graph and
let $v$ be a vertex of degree at least four. Let $\tilde{\Gamma}$
be a graph derived from $\Gamma$ by expanding at the vertex $v$
and assume that the new vertices, $v_{1}$ and $v_{2}$, are at least
$3$-valent. These assumptions are necessary for $\tilde{\Gamma}$
to be $3$-connected \cite{tutte01}. Note that $\Gamma$ and $\tilde{\Gamma}$ have
the same number of independent cycles. Moreover, by splitting at the
vertex $v$ we do not change the relations between the $\phi_{Y}$
phases of $\Gamma$. This is simply because if the equality of some
of the $\phi_{Y}$ phases required a cycle passing through $v$, one can now use the cycle with one more edge passing through $v_{1}$ and $v_{2}$
in $\tilde{\Gamma}$. The graph $\tilde{\Gamma}$ contains new
Y-graphs, whose middle vertices are $v_{1}$ or $v_{2}$ and
one of the edges is $e=v_{1}\leftrightarrow v_{2}$. We need to show
that the phase $\phi_{Y}$ on these new Ys is the same as on the old
ones. Let $\{e,f_{1},f_{2}\}$ be such a graph and let $\alpha_1$ and $\alpha_2$ be endpoints of $f_1$ and $f_2$. By $3$-connectedness, there is a path between $\alpha_1$ and $\alpha_2$ which does not contain $v_1$ or $v_2$.  In this way we obtain a cycle $C$, as shown in figure \ref{fig:Adding-an-edge}(b).  Again by $3$-connectedness, there is a path $P$ from $v_2$ to a vertex $\beta$ in $C$.  Let $Y'$ be the Y-graph with $\beta$ as its centre and edges along $C$ and $P$, as shown in figure \ref{fig:Adding-an-edge}(b).  Then $Y'$ belongs to $\Gamma$. Applying Fact \ref{fact1} and  relation (\ref{eq:AB-2})
to the cycle $C$ and the two Y-graphs discussed, the result follows.\qed
\end{proof}

\begin{figure}[h]
\begin{center}~~~~~\includegraphics[scale=0.4]{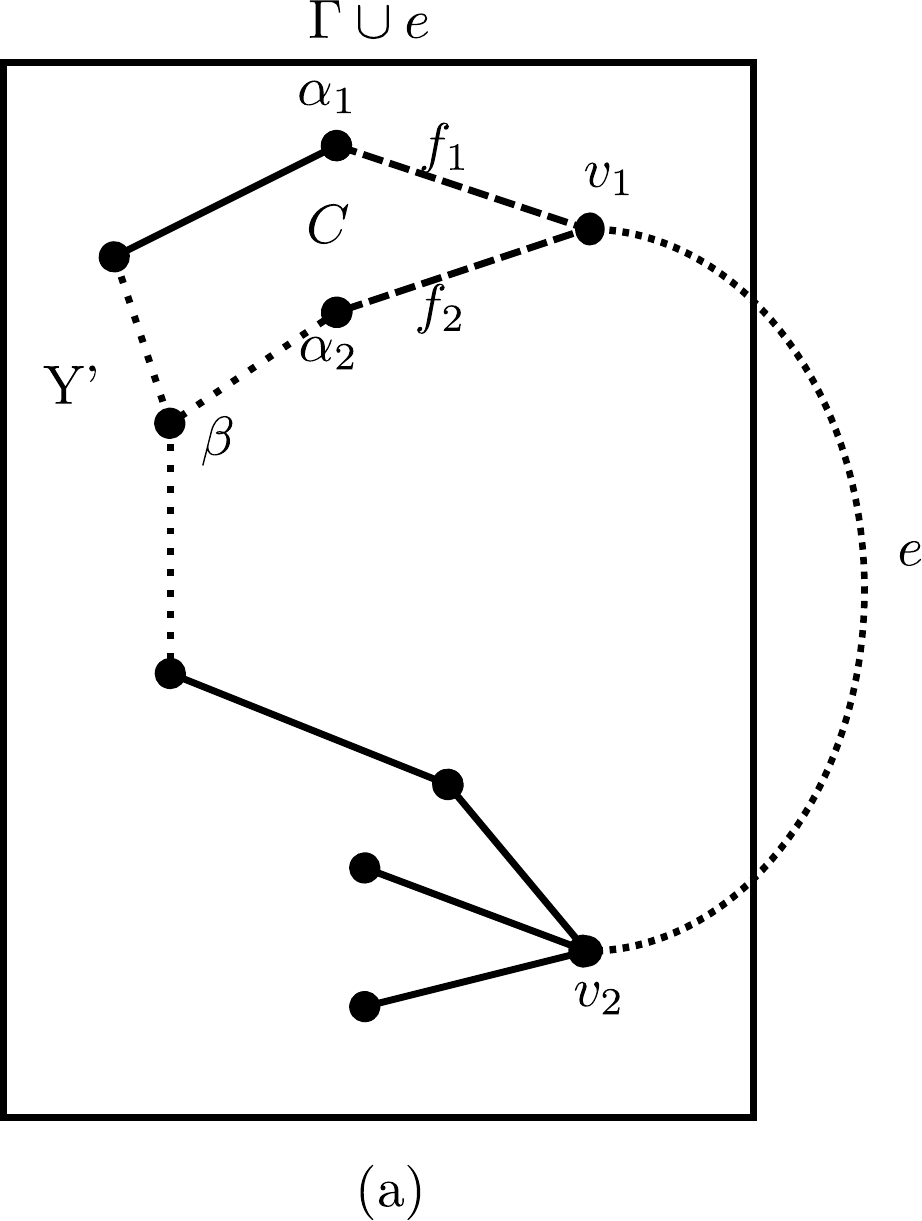}~~~~~~~~~~~~~~~~~~\includegraphics[scale=0.4]{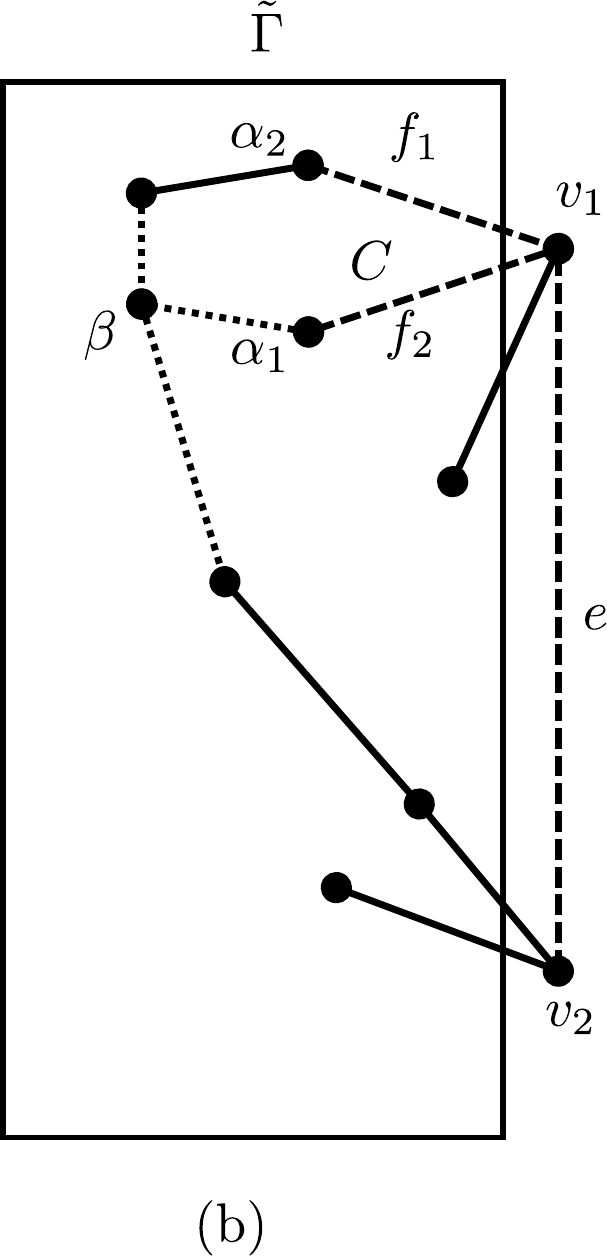}\end{center}

\caption{\label{fig:Adding-an-edge}(a) Adding an edge (b) Expanding at the vertex.}
\end{figure}

\begin{theorem}\label{3-final}
For a $3$-connected simple graph, $H_{1}(\mathcal{D}^{2}(\Gamma))=\mathbb{Z}^{\beta_1(\Gamma)}\oplus A$,
where $A=\mathbb{Z}_{2}$ for non-planar graphs and $A=\mathbb{Z}$
for planar graphs. \end{theorem}
\begin{proof}
By Lemmas 1 and 2 we only need to determine the phase $\phi_{Y}$. It
was shown in \cite{JHJKJR} that for the graphs $K_{5}$ and $K_{3,3}$
$H_1(\mathcal{D}^2(\Gamma))=\mathbb{Z}^{\beta_1(\Gamma)}\oplus \mathbb{Z}_2$. Therefore the phase $\phi_{Y}=0$ or $\pi$. By Kuratowski's
theorem \cite{Kuratowski30} every non-planar graph contains a subgraph
which is isomorphic to $K_{5}$ or $K_{3,3}$. This proves the statement
for non-planar graphs. For planar graphs we have the anyon phase and
hence $A=\mathbb{Z}$. This is because for planar graphs, one can introduce anyon phases by drawing the graph in the plane and integrating the anyon vector potential $\frac{\alpha}{2\pi} \hat{z}\times \frac{r_1 - r_2}{|r_1 - r_2|^2}$ along the edges of the two-particle graph, where $r_1$ and $r_2$ are the positions of the particles.\qed
\end{proof}

Finally, note that as a direct consequence of relation (\ref{eq:AB-2}) and Theorem \ref{3-final}, AB-phases for $3$-connected nonplanar graphs are independent of the positions of stationary particles. This is not necessarily the case for $3$-connected planar graphs, as $\phi_Y$ phases are equal only up to the sign.

\subsection{$2$-connected graphs}\label{subsec: two-connected}

In this subsection we discuss $2$-connected graphs. First, by
considering a simple example we show that in contrast to $3$-connected
graphs it is possible to have more than one $\phi_{Y}$ phase. Using
a decomposition procedure of a $2$-connected graph into $3$-connected
graphs and topological cycles we provide the formula for $H_{1}(\mathcal{D}^{2}(\Gamma))$.

\begin{figure}[h]
\begin{center}\includegraphics[scale=0.5]{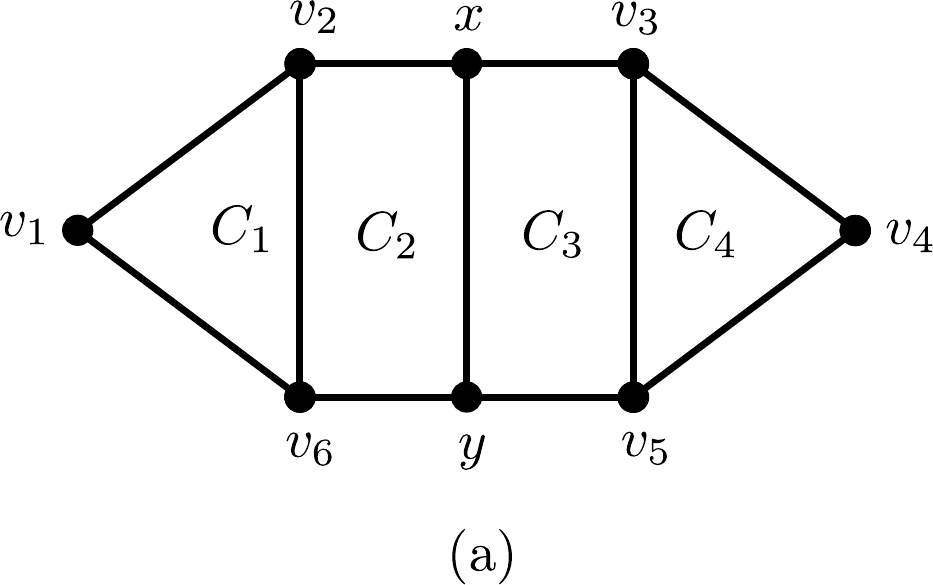}~~~\includegraphics[scale=0.5]{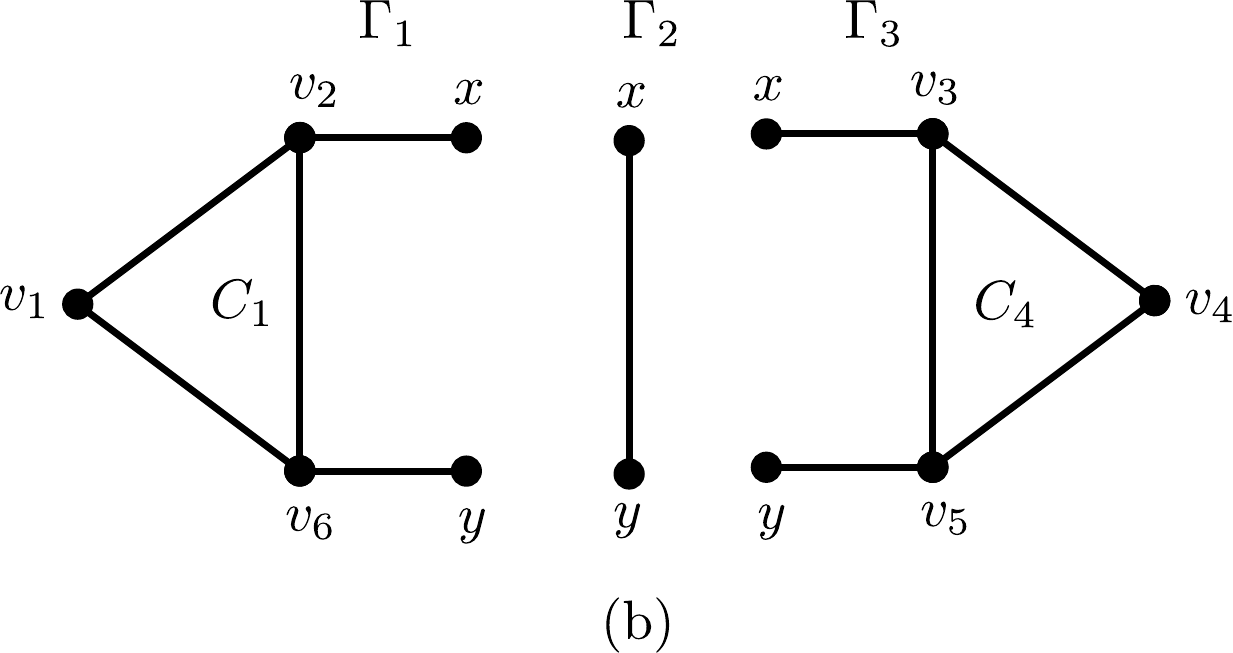}\end{center}

\begin{center}~~~~~~~~~~~~~~\includegraphics[scale=0.5]{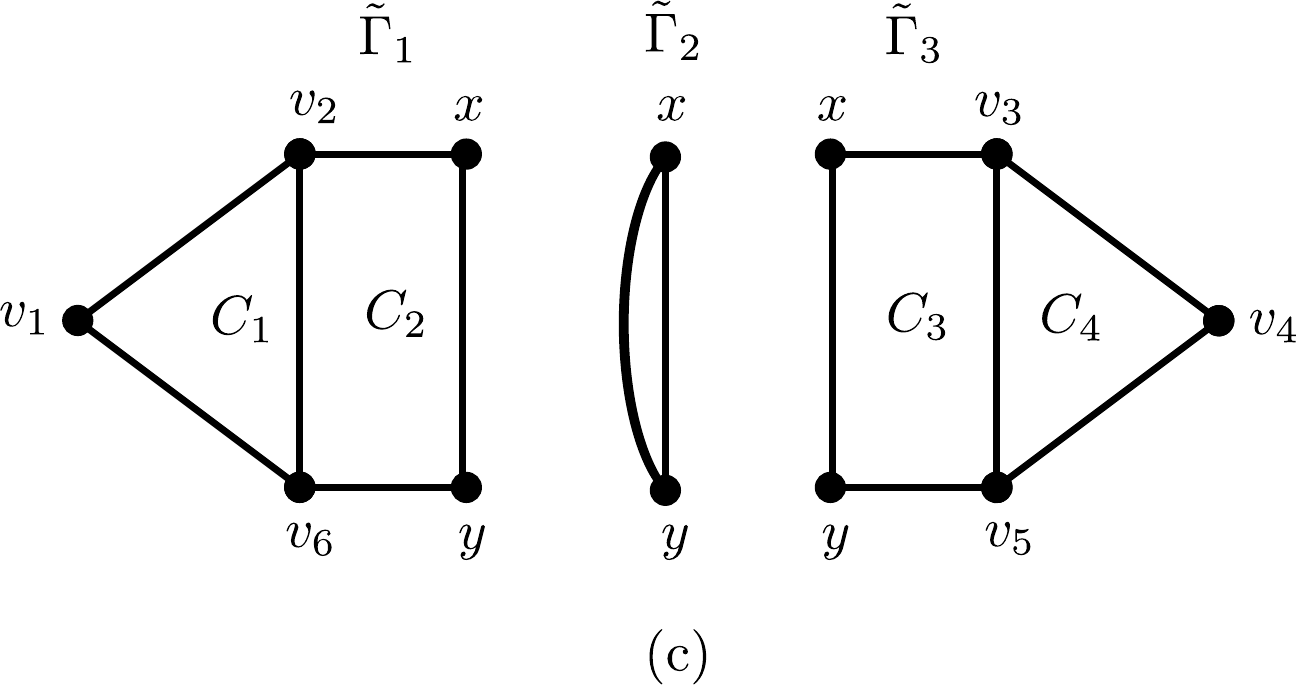}\end{center}

\caption{\label{fig:2-conn-ex}(a) An example of a $2$-connected graph, (b) the
components of the $2$-cut $\{x,y\}$, (c) the marked components. }
\end{figure}

\begin{example}
Let us consider graph $\Gamma$ shown in figure \ref{fig:2-conn-ex}(a).
Since vertices $v_{1}$ and $v_{4}$ are $2$-valent, $\Gamma$ is
not $3$-connected. It is however $2$-connected. Note that $\beta_1(\Gamma)=4$
and that there are six Y-graphs, with middle vertices $v_{2}$, $v_{3}$,
$v_{5}$, $v_{6}$, $x$ and $y$ respectively. Using Fact \ref{fact1} and  relation (\ref{eq:AB-2})
we verify that
\begin{gather}
\phi_{Y_{v_{2}}}=\phi_{Y_{v_{6}}},\,\,\phi_{Y_{v_{3}}}=\phi_{Y_{v_{5}}},\,\,\phi_{Y_{x}}=\phi_{Y_{y}}.\label{eq:example4}
\end{gather}
It is also straightforward to show that the phases $\phi_{Y_{v_{2}}}$, $\phi_{Y_{v_{3}}}$ and $\phi_{Y_{x}}$ are independent. Therefore we have three independent $\phi_{Y}$ phases and four AB-phases, and so
\begin{gather}
H_{1}(\mathcal{D}^{2}(\Gamma))=\mathbb{Z}^{6}.
\end{gather}
Vertices $\{x,y\}$ constitute a $2$-vertex cut of $\Gamma$, i.e.
after their deletion $\Gamma$ splits into three connected components
$\Gamma_{1}$, $\Gamma_{2}$, $\Gamma_{3}$ (see figure \ref{fig:2-conn-ex}(b)).
They are no longer $2$-connected. Moreover, for example, the two Y-subgraphs $Y_{v_{2}}$ and $Y_{v_{6}}$ for which $\phi_{Y_{v_{2}}}=\phi_{Y_{v_{6}}}$
in $\Gamma$ no longer satisfy this condition in $\Gamma_{1}$, i.e. $\phi_{Y_{v_{2}}}\neq\phi_{Y_{v_{6}}}$ in $\Gamma_1$. This
is because the AB-phases $\phi_{C_{1},1}^{x}$ and $\phi_{C_{1},1}^{y}$
are not equal. To make components $\Gamma_i$ $2$-connected and at the same time keep the correct relations between $\phi_{Y_{v_{i}}}$, it is enough to add to each
component $\Gamma_{i}$ an additional edge between vertices $x$ and $y$ (see figure \ref{fig:2-conn-ex}(c)).
The resulting graphs, which we call the marked components and denote by $\tilde{\Gamma}_{i}$ \cite{KP11}, are $2$-connected
and the relations between the Y-graphs in each $\tilde{\Gamma}_{i}$
are the same as in $\Gamma$. The union of the three marked components
has, however, $\beta_1(\Gamma)+1$ independent cycles. On the other
hand, by splitting $\Gamma$ into marked components the Y-cycles $Y_{x}$
and $Y_{y}$ have been lost. Since $\phi_{Y_{x}}=\phi_{Y_{y}}$ we
have lost one $\phi_{Y}$ phase. Summing up we can write $H_{1}(\mathcal{D}^{2}(\Gamma))\oplus\mathbb{Z}=\bigoplus_{i=1}^{3}H_{1}(\mathcal{D}^{2}(\tilde{\Gamma}_{i}))\oplus\mathbb{Z}$.
\end{example}

\paragraph{2-vertex cut for an arbitrary $2$-connected graph $\Gamma$}

In figure \ref{fig:2-vertex-cut}(a) a more general $2$-vertex
cut is shown together with components $\Gamma_{i}$. It is easy to
see that the marked components $\tilde{\Gamma}_{i}$ are $2$-connected
and the relations between the $\phi_{Y}$ phases in each $\tilde{\Gamma}_{i}$
are the same as in $\Gamma$.
Let $\mu(x,y)$ be the number of $\tilde{\Gamma}_{i}$ components into which $\Gamma$ splits after removal of vertices $x$ and $y$. By Euler's formula the union $\{\tilde{\Gamma}_{i}\}_{i=1}^{\mu(x,y)}$ of $\mu(x,y)$ marked components has
\begin{gather}
\beta=\# \mathrm{edges}-\# \mathrm{vertices}+\mu(x,y)\nonumber\\=E(\Gamma)+\mu(x,y)-\left(V(\Gamma)+2(\mu(x,y)-1)\right)+\mu(x,y)\nonumber \\
=E(\Gamma)-V(\Gamma)+2=\beta_1(\Gamma)+1,\label{eq:cycle-number-2-conn}
\end{gather}
independent cycles. By splitting $\Gamma$ into the marked components
we possibly lose $\phi_{Y}$ phases corresponding to the Y-graphs
with the middle vertex $x$ or $y$. However
\begin{enumerate}
\item If three edges of a Y-graph are connected to the same component we do
not lose $\phi_{Y}$.
\item If two edges of a Y-graph are connected to the same component we do
not lose $\phi_{Y}$. This can be understood by looking at figure
\ref{fig:2-vertex-cut}(b). The phase of the dashed Y-graph which
is lost is the same as the phase of the dotted Y-graph inside $\Gamma_{2}$.
\end{enumerate}
Hence the $\phi_{Y}$ phases we lose correspond to the Y-graphs for which
each edge is connected to a different component. First we want to show that any two Y-graphs with the central vertex $x$ (or $y$) whose edges are connected to three fixed components have the same phase. It is enough to show this for Y-graphs which share the same center and two edges. Let us consider two such Y-graphs (see figure \ref{fig:2-vertex-cut}(c) - the dashed egdes are common for both Y-graphs). Let $\alpha_{1}$ and $\alpha_{2}$ be endpoints of edges which are not shared by both Y-graphs. Since there is a path between $\alpha_{1}$ and $\alpha_{2}$ in $\Gamma_{2}$
and paths $P_{a_{1},a_{2}}$, $P_{b_{1},b_{2}}$ in $\Gamma_{1}$ and $\Gamma_{3}$ respectively we can apply Fact \ref{fact1} and  relation (\ref{eq:AB-2})
to the cycle $x\rightarrow a_{1}\cup P_{a_{1},a_{2}}\cup a_{2}\rightarrow y\rightarrow b_{2}\cup P_{b_{1},b_{2}}\cup b_{2}\rightarrow x$ and the two considered Y-graphs
obtaining equality of the two respective $\phi_{Y}$ phases. Therefore, the choice of three components gives only one $\phi_{Y}$ phase. Moreover,
note that after choosing three components the phase for the Y-graph
with the middle vertex $x$ is the same as for the Y-graph with the
middle vertex $y$ (see figure \ref{fig:2-vertex-cut}(d) where the considered Y-graphs are denoted by dashed and dotted lines). This is once again due to Fact \ref{fact1} and  relation (\ref{eq:AB-2}) applied to the cycle $x\rightarrow a_{1}\cup P_{a_{1},a_{2}}\cup a_{2}\rightarrow y\rightarrow \alpha_{2}\cup P_{\alpha_{1},\alpha_{2}}\cup \alpha_{2}\rightarrow x$ and the two considered Y-graphs. Summing up, the number of phases we lose when splitting $\Gamma$ into $\mu(x,y)$ marked components, $N_{2}(x,y)$, is equal to the number
of independent Y-graphs in the star graph with $\mu(x,y)$ edges. This can be calculated (see for example \cite{JHJKJR}) to be $N_{2}(x,y)=\frac{1}{2}\left(\mu(x,y)-2\right)\left(\mu(x,y)-1\right)$.
Hence
\begin{gather}
H_{1}(\mathcal{D}^{2}(\Gamma))=\bigoplus_{i=1}^{\mu(x,y)}H_{1}(\mathcal{D}^{2}(\tilde{\Gamma}_{i}))\oplus\mathbb{Z}^{N_{2}(x,y)-1}.\label{eq:2-connected}
\end{gather}
Note that the $-1$ in the exponent here is to get rid of the additional AB-phase stemming from the calculation (\ref{eq:cycle-number-2-conn}).
\begin{figure}[h]
\begin{center}\includegraphics[scale=0.35]{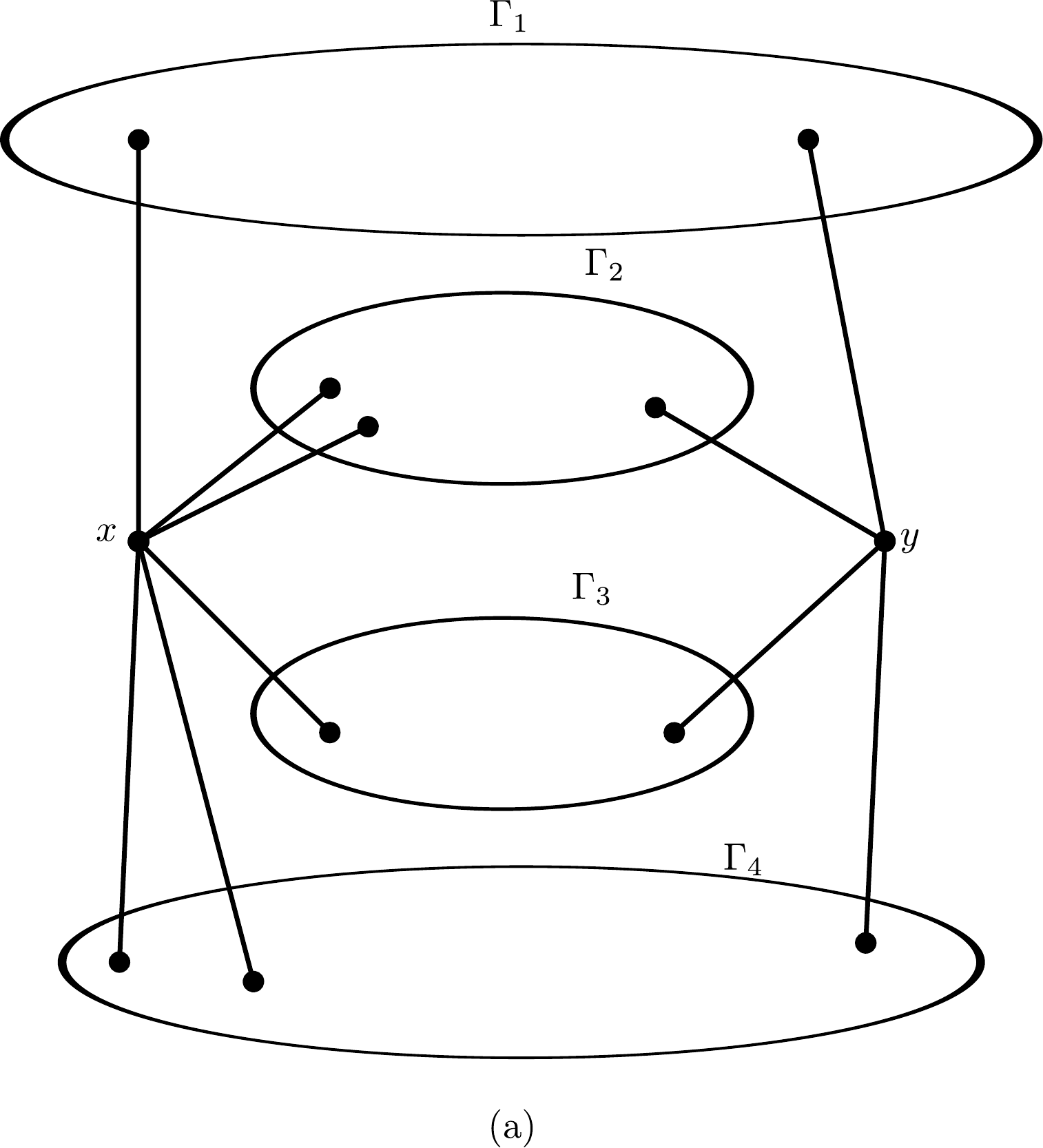}~~~\includegraphics[scale=0.35]{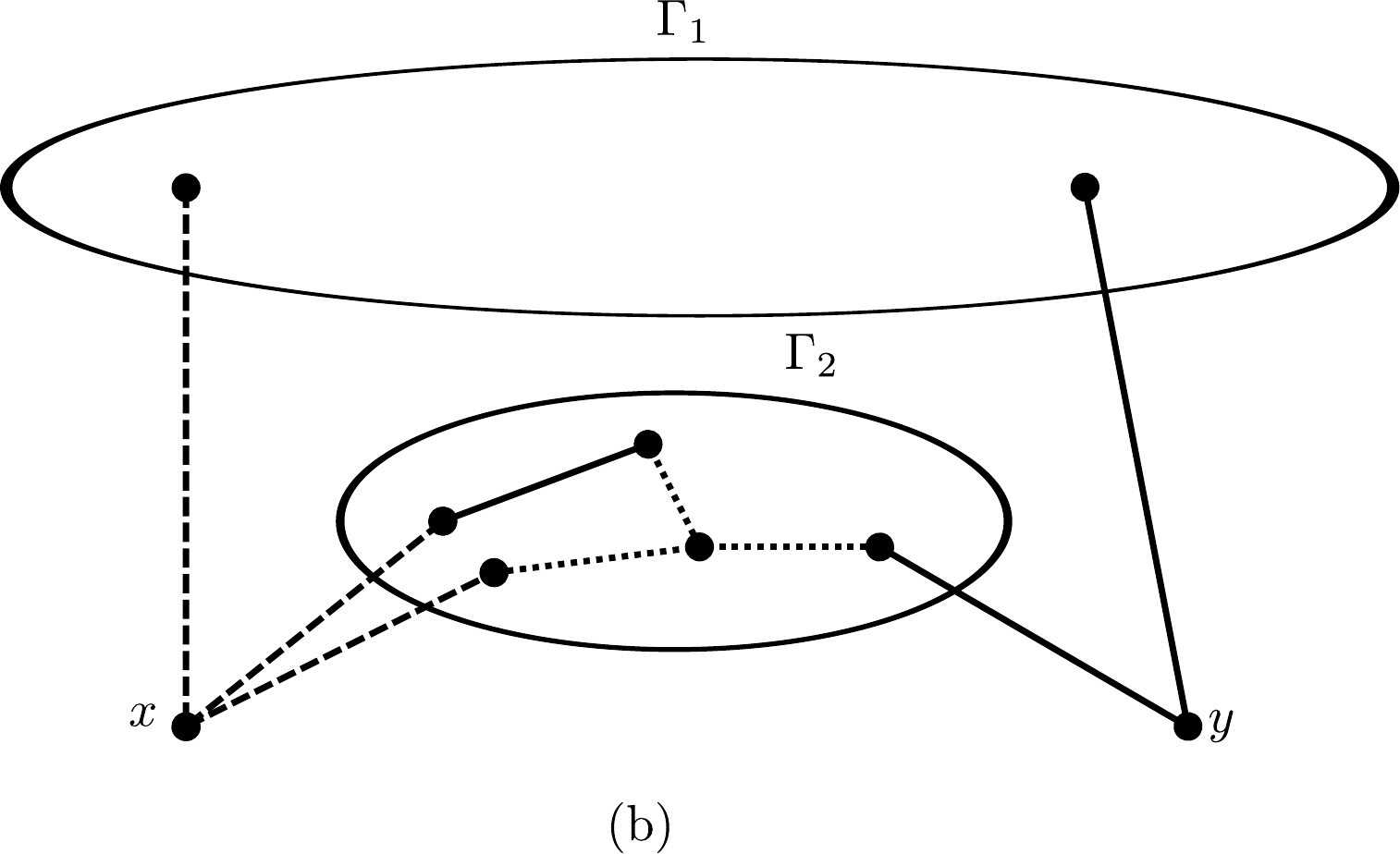}\end{center}
\bigskip{}
\begin{center}\includegraphics[scale=0.35]{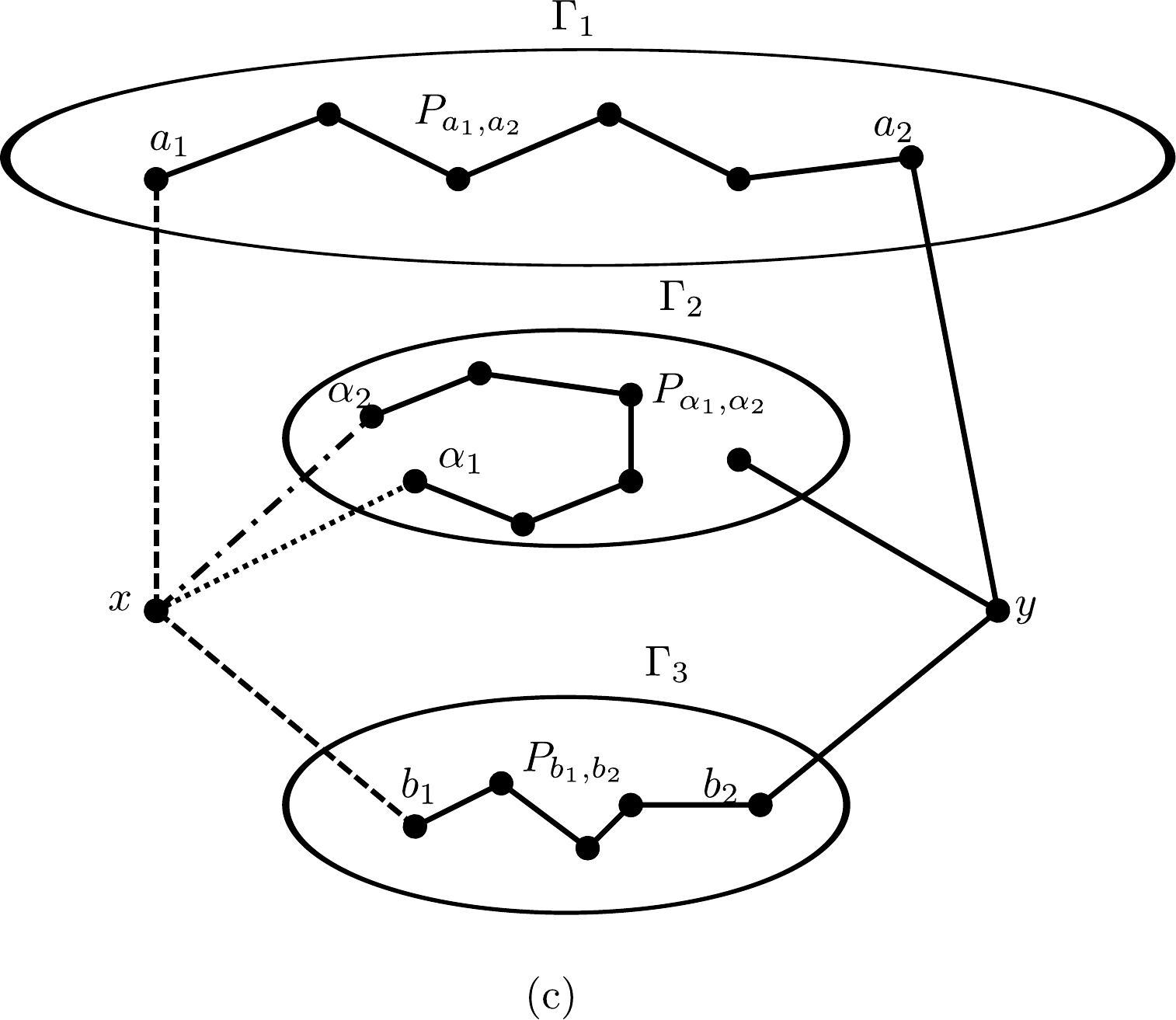}~~\includegraphics[scale=0.35]{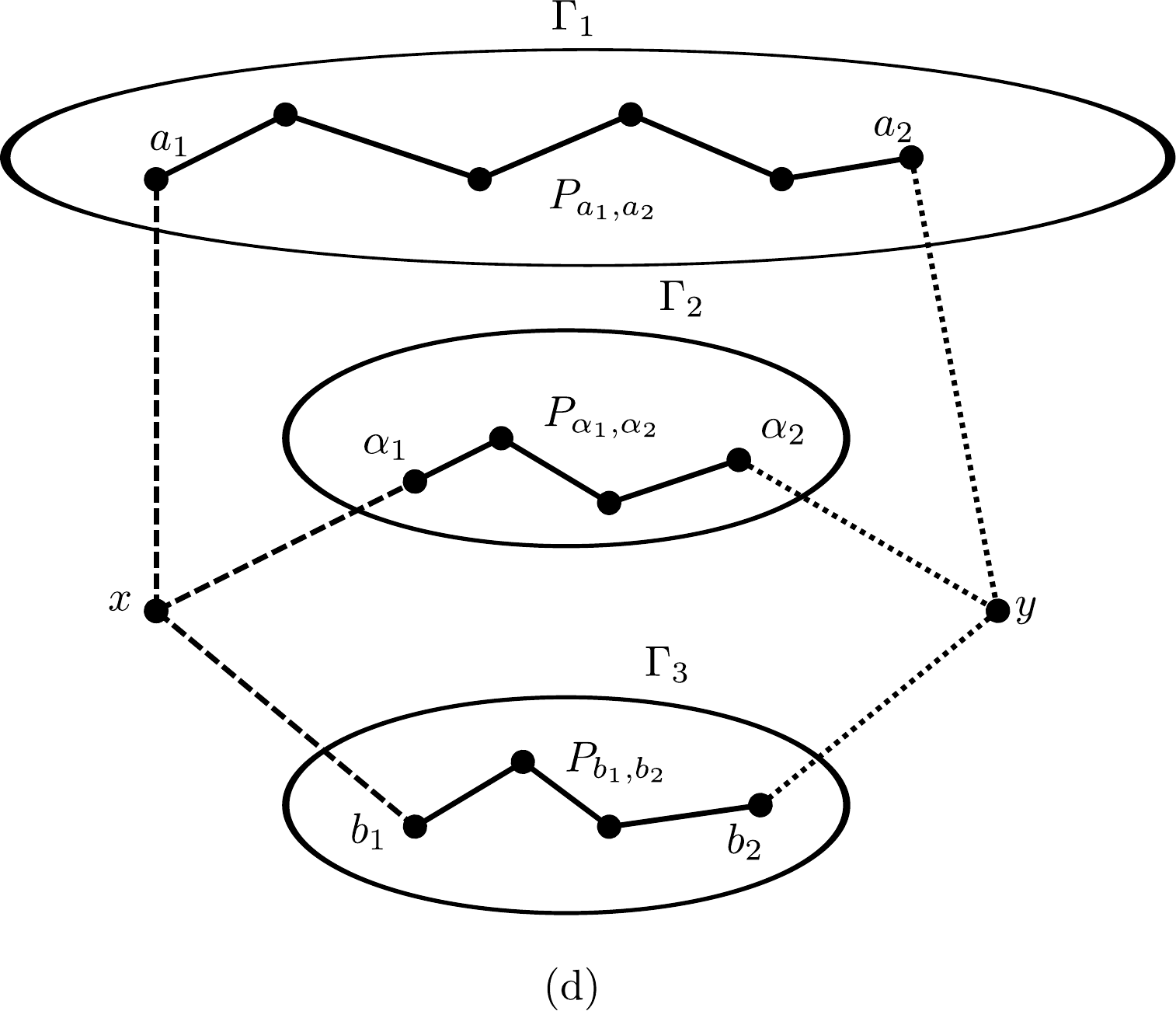}\end{center}
\caption{\label{fig:2-vertex-cut}(a) 2-vertex cut of $\Gamma$ (b) $Y_{x}$
with two edges connected to $\Gamma_{2}$ (c) two Y-cycles with three
edges in three different components (d) the equality of $\phi_{Y_{x}}$
and $\phi_{Y_{y}}$}
\end{figure}
Finally, it is known in graph theory that by the repeated application
of the above decomposition procedure the resulting marked components
are either topological cycles or $3$-connected graphs \cite{tutte01}. Let $n$ be
the number of $2$-vertex cuts which is needed to get such a decomposition,
$N_{2}=\sum_{\{x_{i},y_{i}\}}N_{2}(x_{i},y_{i})$, $N_{3}$ the number
of planar $3$-connected components, $N_{3}^{\prime}$ the number of
non-planar $3$-connected components and $N_{3}^{''}$ the number of
the topological cycles. Let $\mu=N_{3}+N_{3}^{'}+N_{3}^{''}$. Then
\begin{gather}
H_{1}(\mathcal{D}^{2}(\Gamma))=\bigoplus_{i=1}^{\mu}H_{1}(\mathcal{D}^{2}(\tilde{\Gamma}_{i}))\oplus\mathbb{Z}^{N_{2}-n},\label{eq:2-conn-intermidiate}
\end{gather}
where
\begin{gather}
H_{1}(\mathcal{D}^{2}(\tilde{\Gamma}_{i}))=\mathbb{Z}^{\beta_1(\tilde{\Gamma}_{i})}\oplus\mathbb{Z},\,\,\, \tilde{\Gamma}_{i}-\mathrm{planar}\\\nonumber
H_{1}(\mathcal{D}^{2}(\tilde{\Gamma}_{i}))=\mathbb{Z}^{\beta_1(\tilde{\Gamma}_{i})}\oplus\mathbb{Z}_{2},\,\,\,\tilde{\Gamma}_{i}-\mathrm{nonplanar}\\\nonumber
H_{1}(\mathcal{D}^{2}(\tilde{\Gamma}_{i}))=\mathbb{Z},\,\,\,\tilde{\Gamma}_{i}-\mathrm{topological}\,\,\mathrm{cycle}\\\nonumber
\end{gather}
Note that $\sum_{i}\beta_1(\tilde{\Gamma}_{i})+N_{3}^{''}=\beta_1(\Gamma)+n$
and therefore
\begin{gather}
H_{1}(\mathcal{D}^{2}(\Gamma))=\mathbb{Z}^{\beta_1(\Gamma)+N_{2}+N_{3}}\oplus\mathbb{Z}_{2}^{N_{3}^{\prime}}.\label{eq:2-c}
\end{gather}

\subsection{$1$-connected graphs\label{sub:One-connected-graphs}}
In this subsection we focus on $1$-connected graphs. Assume that $\Gamma$ is $1$-connected but not $2$-connected. There exists a vertex $v\in V(\Gamma)$ such that after its deletion
$\Gamma$ splits into at least two connected components. Denote by
$\Gamma_{1},\ldots,\Gamma_{\mu(v)}$ these components. Assume that
$\Gamma_{i}$ is attached to $v$ by $E_{i}$ edges and put $\nu(v)=\sum_{i}E_{i}$,
so that $\nu(v)$ is the total number of edges at $v$. By Euler's formula the union of components $\{\Gamma_{i}\}_{i=1}^{{\mu(v)}}$ has
\begin{gather}
E(\Gamma)-\left(V(\Gamma)+\mu(v)-1\right)+\mu(v)=\beta_1(\Gamma),\label{eq:cycles-1-conn}
\end{gather}
independent cycles, hence the number of independent cycles does not change
compared to $\Gamma$. Moreover, the phases $\phi_{Y}$ inside each
of the components are the same as in $\Gamma$. Note, however, that
by splitting we lose Y-graphs whose three edges do not belong to
one fixed component $\Gamma_{i}$. Consequently, there are two cases to consider:

\begin{figure}[h]
\begin{center}~~\includegraphics[scale=0.35]{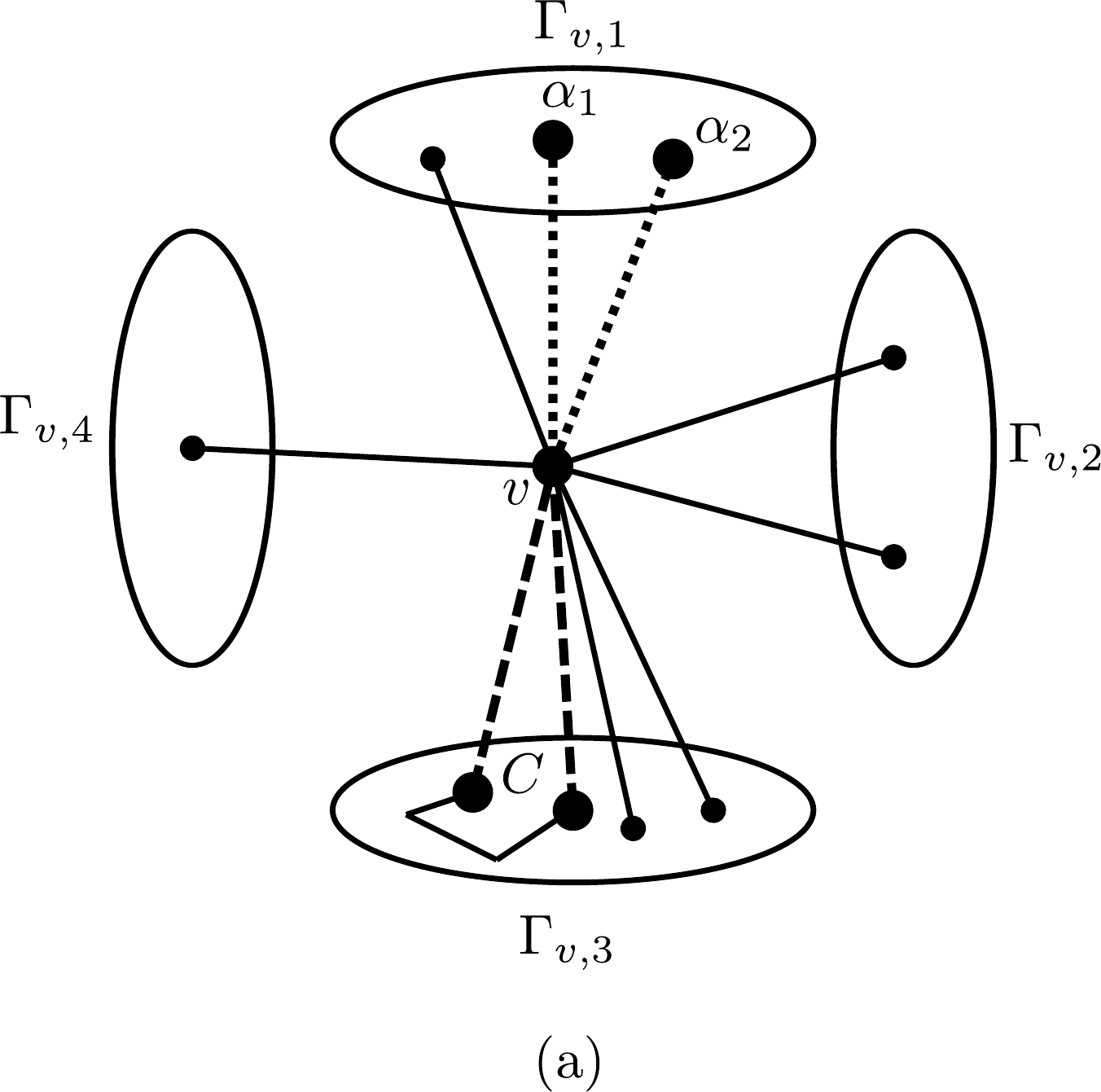}~~~~~~~~~~~~~~\includegraphics[scale=0.35]{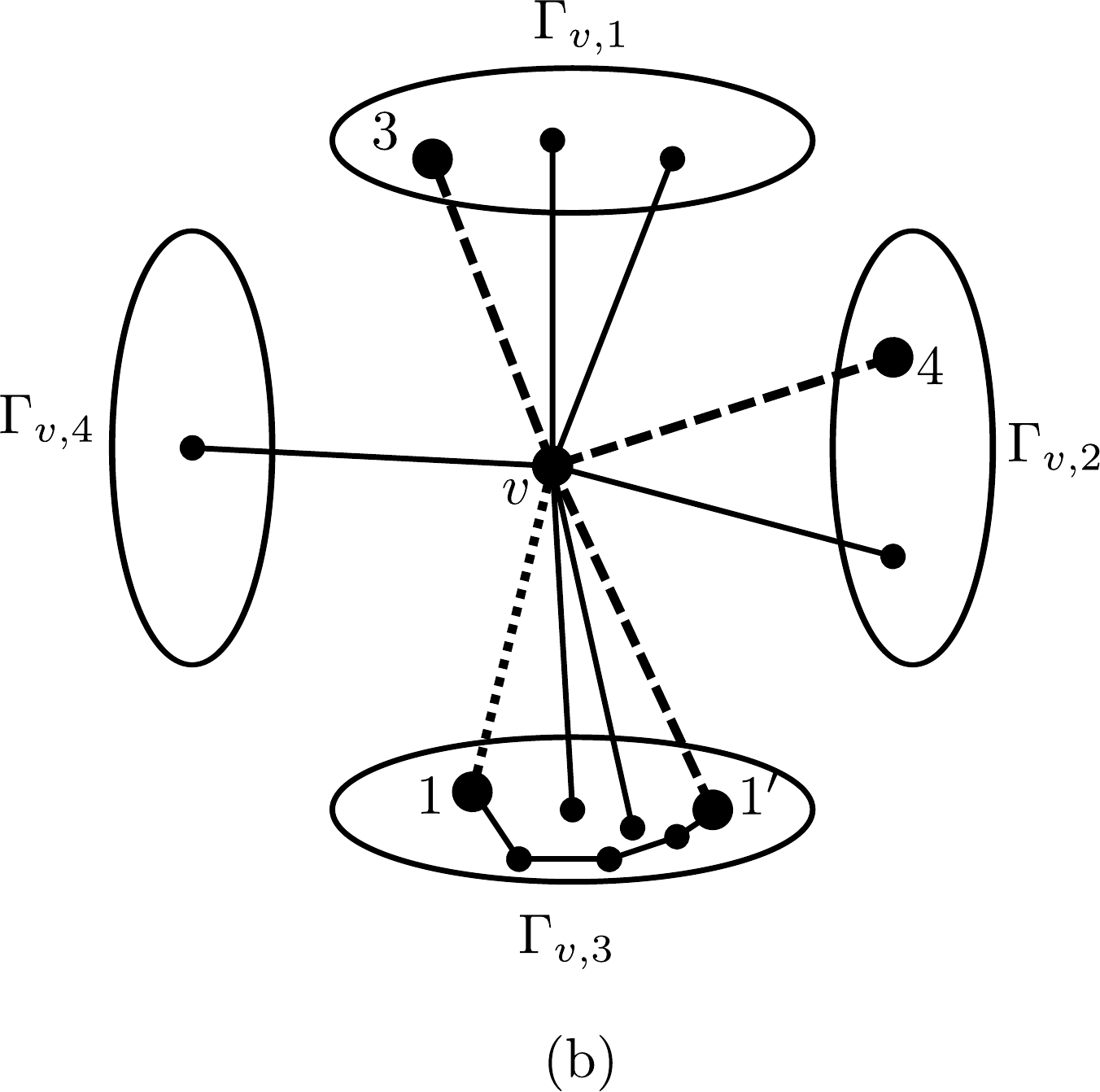}~\end{center}


\begin{center}\includegraphics[scale=0.6]{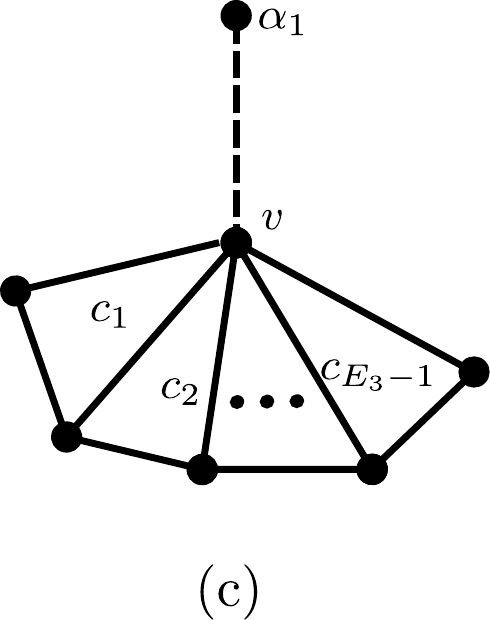}~~~\includegraphics[scale=0.55]{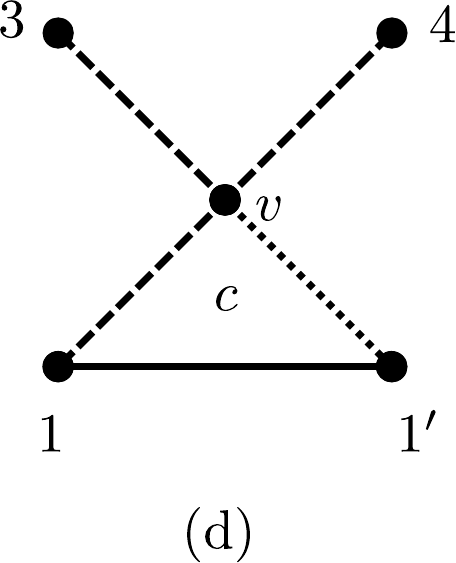}~~~~\includegraphics[scale=0.5]{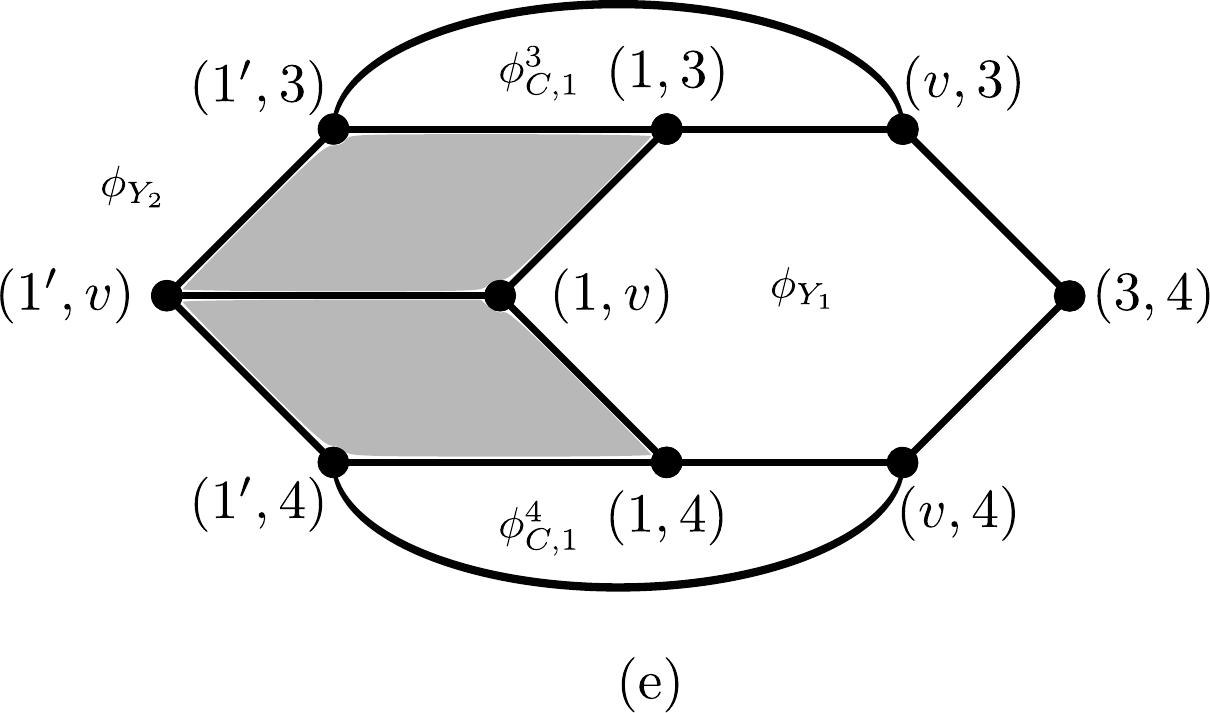}
\end{center}

\caption{\label{fig:1-vertex-cut}(a)Two edges of $Y$ are attached to $\Gamma_{v,3}$
while the third one to $\Gamma_{v,1}$ (b) Each edge of Y is attached
to a different component (c) Y-graphs with two edges in the same component (d) Two Y-graphs (d) The relevant part of $2$-particle configuration space of (d)}
\end{figure}

\begin{enumerate}
\item Two edges of the Y-graph are attached to one component, for example $\Gamma_{v,3}$,
while the third one is attached to another component, $\Gamma_{v,1}$. We claim that the phase $\phi_{Y}$ does
not depend on the choice of the third edge, provided
it is attached to $\Gamma_{v,1}$. To see this consider two Y-graphs,
$Y_{1}$ and $Y_{2}$ shown in figure \ref{fig:1-vertex-cut}(a).
Since vertices $\alpha_{1}$ and $\alpha_{2}$ are connected by a
path, by Fact \ref{fact1} $\phi_{C,1}^{\alpha_{1}}=\phi_{C,1}^{\alpha_{2}}$. Next, relation (\ref{eq:AB-2}) applied to cycle $C$ and the two considered Y  graphs gives $\phi_{Y_{1}}=\phi_{Y_{2}}$. \\

After choosing one edge of Y in component $\Gamma_{v,1}$ (by the above argument it does not matter which), we can choose the  two other edges in $\Gamma_{v,3}$ in $E_{3}\choose{2}$ ways. Therefore, {\it a priori}, we have $E_{3}\choose{2}$ Y-graphs to consider. There are, however,  relations between them. In order to find the relevant relations consider the graph shown in figure \ref{fig:1-vertex-cut}(c).
We are interested in  Y-graphs with one edge given by $\alpha_1 \leftrightarrow v$ (dashed line) and two edges joining $v$ to vertices in
 $\Gamma_{v,3}$, say  $j$ and $k$.
 Each such Y-graph determines a cycle $c$ in $\Gamma_{v,3}$ containing vertices $v$, $j$ and $k$ (since $\Gamma_{v,3}$ is connected).
We have that
\begin{gather}\label{eq: c2 c1 Y for 1 connected}
\phi_{c,2}=\phi_{c,1}^{\alpha_1}+\phi_{Y}.
\end{gather}
Therefore, the $E_{3}\choose{2}$ $Y$-phases under consideration are determined by the AB- and two-particle phases, $\phi_{c,2}$ and $\phi_{c,1}^{\alpha_1}$, of the associated cycles $c$.  These cycles may be expressed as linear combinations of a basis  of $E_3-1$ cycles, denoted $c_1, \ldots, c_{E_3 - 1}$, 
as  in figure \ref{fig:1-vertex-cut}(c).  It is clear that if $c = \sum_{i=1}^{E_3} r_i c_i$, then 
\begin{gather}
\phi_{c,1}^{\alpha_1}=\sum_{i=1}^{E_3-1} r_i \phi_{c_i,1}^{\alpha_1}, \quad \,\,\,\phi_{c,2}=\sum_{i=1}^{E_3-1} r_i \phi_{c_i,2}.
\end{gather}
Thus, the $Y$-phases under consideration may be expressed in terms of the $2(E_3-1)$ phases $\phi_{c_i,2}$ and $\phi_{c_i,1}^{\alpha_1}$. \\

Let $Y_i$ be the $Y$-graph which determines the cycle $c_i$.  We may turn the preceding argument around; from  \eqref{eq: c2 c1 Y for 1 connected}, the AB-phase $\phi_{c_i,1}^{\alpha_1}$ can be expressed in terms of $\phi_{Y_i}$ and $\phi_{c_i,2}$.  Combining the preceding observations, we deduce that the $\binom{E_3}{2}$  Y-phases lost when the vertex $v$ is removed may be expressed in terms of the phases $\phi_{c_i,2}$ and $\phi_{Y_i}$.  The phases  $\phi_{c_i,2}$ remain when $v$ is removed.   It follows that phases $\phi_{Y_i}$ suffice to determine all of the lost phases, so that the number of independent $Y$-phases lost is $E_3-1$.
Repeating this argument for each component, the total  number of Y-phases lost is $\sum_{i=1}^{\mu(v)}(E_{i}-1)(\mu(v)-1)=(\mu(v)-1)(\nu(v)-\mu(v))$.\\

\item Each edge of Y is attached to a different component. We will show now that once
three different components have been chosen it does not matter which
of the edges attaching $\Gamma_{v,i}$ to $v$ we choose. To see this
let us consider two Y-graphs shown in figure \ref{fig:1-vertex-cut}(b).
The first one consists of the three dashed edges and the second of
two dashed edges attached to $\Gamma_{v,1}$ and $\Gamma_{v,2}$ respectively and
the dotted edged attached to $\Gamma_{v,3}$. We will show
that the phase corresponding to Y-graph $\{1^{\prime}\leftrightarrow v,\,v\leftrightarrow3,\,v\leftrightarrow4\}$
is determined by the phase corresponding to Y-graph  $\{1\leftrightarrow v,\,v\leftrightarrow3,\,v\leftrightarrow4\}$ (see figure \ref{fig:1-vertex-cut}(d)) and phases added in the previously considered step. It is clear by figure \ref{fig:1-vertex-cut}(e) that
\begin{gather}
\phi_{Y_2}=\phi_{Y_1}+\phi_{c,1}^3+\phi_{c,1}^4.
\end{gather}
But phases $\phi_{c,1}^3$ and $\phi_{c,1}^4$ are known, as they have been added in the previous step. Thus, the number of the independent
Y-phases we lose is equal to the number of independent
Y-cycles in the two-particle configuration space of the star graph
with $\mu(v)$ edges, that is, $(\mu(v)-1)(\mu(v)-2)/2$.
\end{enumerate}
Summing up we can write
\begin{gather}
H_{1}(\mathcal{D}^{2}(\Gamma))=\bigoplus_{i=1}^{\mu(v)}H_{1}(\mathcal{D}^{2}(\Gamma_{v,i}))\oplus\mathbb{Z}^{N_{1}(v)},\label{eq:1-cut-forula}
\end{gather}
where $N_{1}(v)=(\mu(v)-1)(\mu(v)-2)/{2}+(\mu(v)-1)(\nu(v)-\mu(v))$.
It is known in graph theory \cite{tutte01} that by the repeated application
of the above decomposition procedure the resulting components become finally
$2$-connected graphs. Let $v_{1},\ldots,v_{l}$ be the set of cut
vertices such that components $\Gamma_{v_{i},k}$ are $2$-connected.
Making use of formula (\ref{eq:2-c}) we can write

\begin{gather}
H_{1}(\mathcal{D}^{2}(\Gamma))=\mathbb{Z}^{\beta(\Gamma)+N_{1}+N_{2}+N_{3}}\oplus\mathbb{Z}_{2}^{N_{3}^{\prime}},\label{eq:2-particle-final}
\end{gather}
where $N_{1}=\sum_{i}N_{1}(v_{i})$.

\section{n-particle statistics for $2$-connected graphs\label{sec:N-particle-statistics-for}}

Having discussed $2$-particle configuration spaces, we switch to the
$n$-particle case, $\mathcal{D}^{n}(\Gamma)$, where $n>2$. We proceed in a
similar manner to the previous section. First we give a spanning set of $H_1(\mathcal{D}^n(\Gamma))$. Next we show that if $\Gamma$ is $2$-connected
the first homology group stabilizes with respect to $n$, that is, $H_{1}(\mathcal{D}^{n}(\Gamma))=H_{1}(\mathcal{D}^{2}(\Gamma))$.
Making use of formula (\ref{eq:2-particle-final})

\begin{gather*}
H_{1}(\mathcal{D}^{n}(\Gamma))=\mathbb{Z}^{\beta(\Gamma)+N_{2}+N_{3}}\oplus\mathbb{Z}_{2}^{N_{3}^{\prime}}.
\end{gather*}

\subsection{A spanning set of $H_1(\mathcal{D}^{n}(\Gamma))$\label{sub:An-over-complete-basis-2}}

In order to calculate $H_1(\mathcal{D}^{n}(\Gamma))$ we first need to
subdivide the edges of $\Gamma$ appropriately. By Theorem \ref{Abrams_thm} each edge of $\Gamma$
must be able to accommodate $n$ particles and each cycle needs to
have at least $n+1$ vertices, that is, $\Gamma$ needs to be sufficiently subdivided. Before we specify a
spanning set of $H_1(\mathcal{D}^{n}(\Gamma))$ we first discuss
two interesting aspects of this space. The first one concerns the
relation between the exchange phase of $k$ particles, $k\leq n$
on the cycle $C$ of the lasso graph and its $\phi_{Y}$ phases (see
Lemma \ref{aspect1-1} ). The second gives the relation between the
AB-phases for fixed cycle $c$ of $\Gamma$ and the different possible positions
of the $n-1$ stationary particles.
\begin{lemma}
\label{aspect1-1}The exchange phase, $\phi_{C,n}$, of $n$ particles
on the cycle $c$ of the lasso graph is the sum of the exchange phase,
$\phi_{C,n-1}^{1}$, of $n-1$ particles on the cycle $C$ with the
last particle sitting at the vertex not belonging to $C$, e.g. vertex
$1$, and the phase $\phi_Y$ associated with the exchange of two particles on the $Y$ subgraph with $n-2$ particles
placed in the vertices $v_1,\ldots ,v_{n-2}$ of $C$ not belonging to the Y
\begin{gather*}
\phi_{C,n}=\phi_{C,n-1}^{1}+\phi_{Y}^{v_1,\ldots ,v_{n-2}}.
\end{gather*}
\end{lemma}
\begin{proof}
By (\ref{eq:lasso-relation}) the above lemma is true for $n=2$.
For the proof in the general case it is enough to consider the lasso
graphs with $3$ and $4$ particles shown in figures \ref{fig:-The-subdivided}(a)
and \ref{fig:-The-subdivided}(b). It is easy to see that they are
indeed sufficiently subdivided. The Y-graphs we consider are $\{2\leftrightarrow3,\,3\leftrightarrow4,\,3\leftrightarrow6\}$
and $\{3\leftrightarrow4,\,4\leftrightarrow5,4\leftrightarrow8\}$
respectively. The relevant parts of the $3$ and $4$-particle configuration
spaces are shown in figures \ref{fig:The-relevant-parts}(a) and \ref{fig:The-relevant-parts}(b).
The statement follows immediately from these figures.\qed
\end{proof}
\begin{figure}[h]
\begin{center}~~~~~~\includegraphics[scale=0.5]{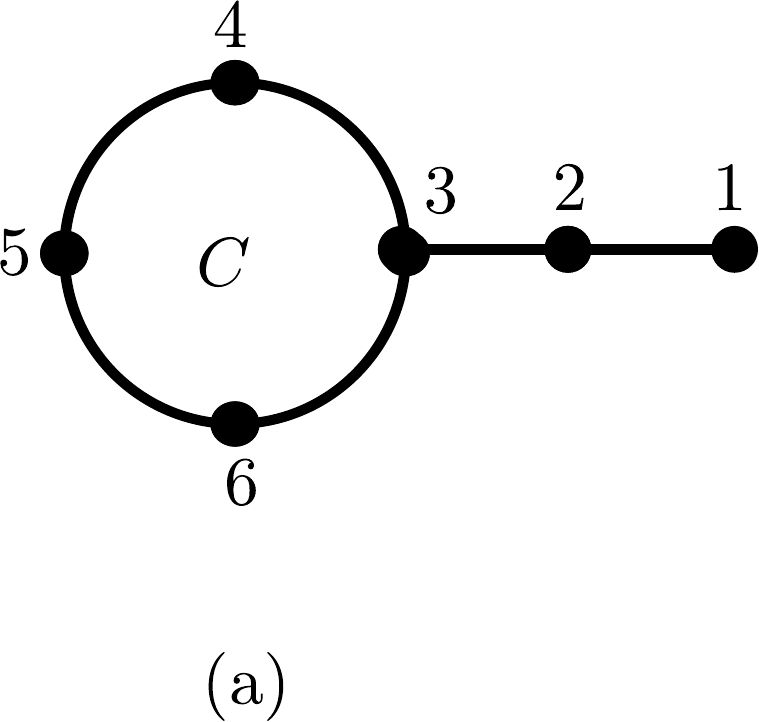}~~~~~~~~~~~\includegraphics[scale=0.5]{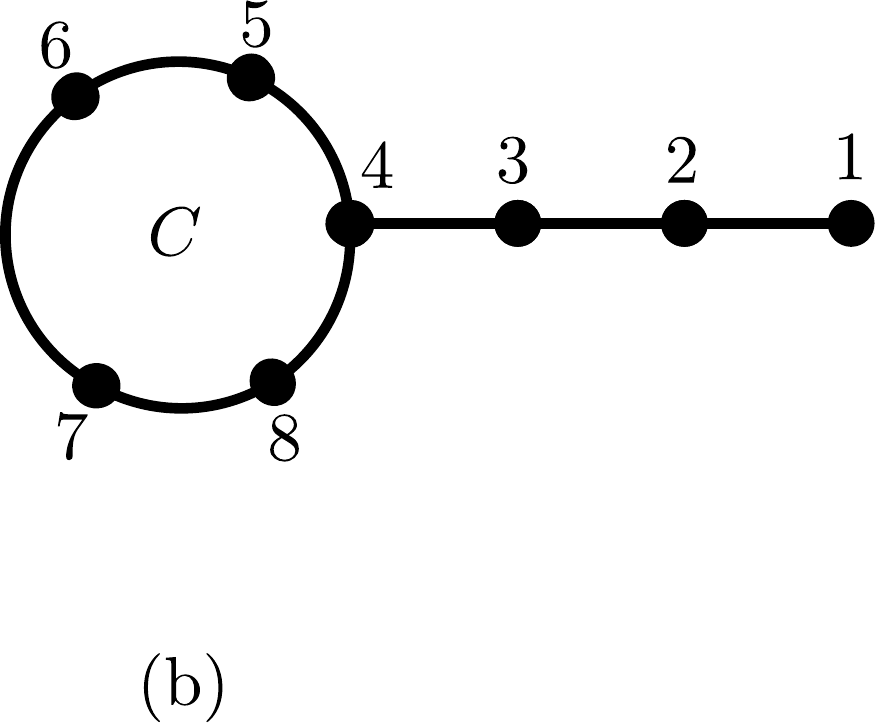}\end{center}

\caption{ \label{fig:-The-subdivided}The subdivided lasso for (a) $3$ particles,
(b) $4$ particles. }
\end{figure}

\begin{figure}[h]
\begin{center}\includegraphics[scale=0.5]{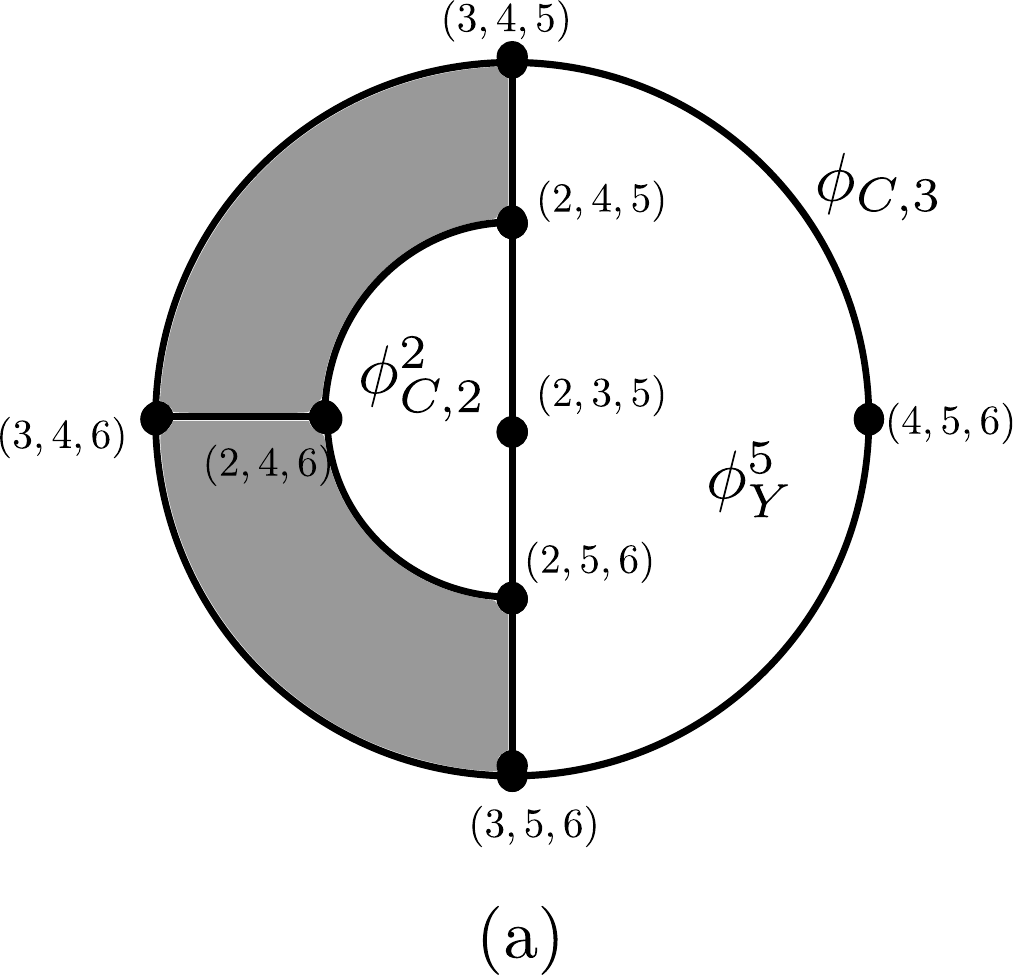}~~~\includegraphics[scale=0.5]{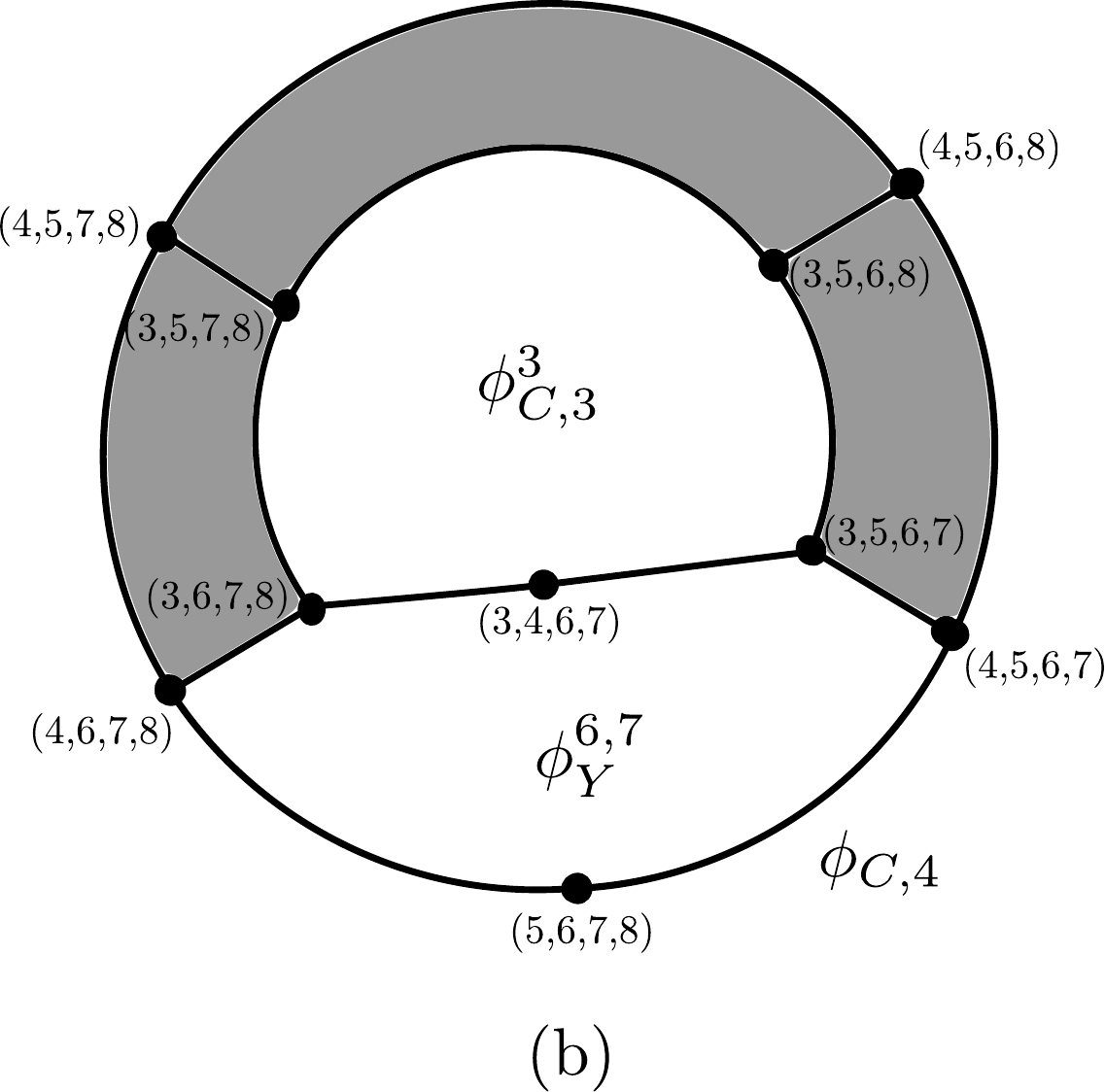}\end{center}

\caption{\label{fig:The-relevant-parts}The relevant parts of the configurations
spaces for the lasso graphs with (a) $3$ particles: $\phi_{C,3} = \phi^2_{C,2} + \phi^5_Y$, (b) $4$ particles: $\phi_{C,4}  = \phi^3_{C,3} + \phi^{6,7}_Y$. }
\end{figure}

\noindent By repeated application of Lemma \ref{aspect1-1} we see
that $\phi_{C,n}$ can be expressed as a sum of an AB-phase and the
Y-phases corresponding to different positions of $n-2$ particles.
For example in the case of the graphs from figure \ref{fig:-The-subdivided}(a)
and \ref{fig:-The-subdivided}(b) we get

\begin{gather*}
\phi_{C,3}=\phi_{Y}^{5}+\phi_{C,2}^{2}=\phi_{Y}^{5}+\phi_{Y}^{1}+\phi_{C,1}^{1,2}\,,\\
\phi_{C,4}=\phi_{Y}^{6,7}+\phi_{C,3}^{3}=\phi_{Y}^{6,7}+\phi_{C,3}^{1}=\phi_{Y}^{6,7}+\phi_{Y}^{1,6}+\phi_{Y}^{1,2}+\phi_{C,1}^{1,2,3}\,.
\end{gather*}

\paragraph{Aharonov-Bhom phases}

Assume now that we have $n$ particles on $\Gamma$. Let $C$ be a
cycle of $\Gamma$ and $e_{1}$ and $e_{2}$ two sufficiently subdivided
edges attached to $C$ (see figure \ref{fig:stabilization}(a)). We
denote by $\phi_{C,1}^{k_{1},k_{2}}$ the AB-phase corresponding
to the situation where one particle goes around the cycle $C$ while $k_{1}$
particles are in the edge $e_{1}$ and $k_{2}$ particles are in the
edge $e_{2}$, $k_{1}+k_{2}=n-1$. For each distribution $(k_{1},k_{2})$
of the $n-1$ particles between the edges $e_{1}$ and $e_{2}$ we
get a (possibly) different AB-cycle and AB-phase in $\mathcal{D}^{n}(\Gamma)$.
We want to know how they are related. To this end notice that

\begin{gather}
\phi_{C,2}^{k_{1},k_{2}}=\phi_{C,1}^{k_{1}+1,k_{2}}+\phi_{Y_{1}}^{k_{1},k_{2}},\,\,\,\,\phi_{C,2}^{k_{1},k_{2}}=\phi_{C,1}^{k_{1},k_{2}+1}+\phi_{Y_{2}}^{k_{1},k_{2}},\label{eq:AB-1-1}
\end{gather}
and hence
\begin{gather}
\phi_{C,1}^{k_{1}+1,k_{2}}-\phi_{C,1}^{k_{1},k_{2}+1}=\phi_{Y_{2}}^{k_{1},k_{2}}-\phi_{Y_{1}}^{k_{1},k_{2}}.\label{eq:AB-2-1}
\end{gather}
The relations between different AB-phases for a fixed cycle $C$
of $\Gamma$ are therefore encoded in the $2$-particle phases $\phi_{Y}$,
albeit these phases can depend on the positions of the remaining $n-2$
particles.

\begin{figure}[h]
\begin{center}~~~~~~\includegraphics[scale=0.5]{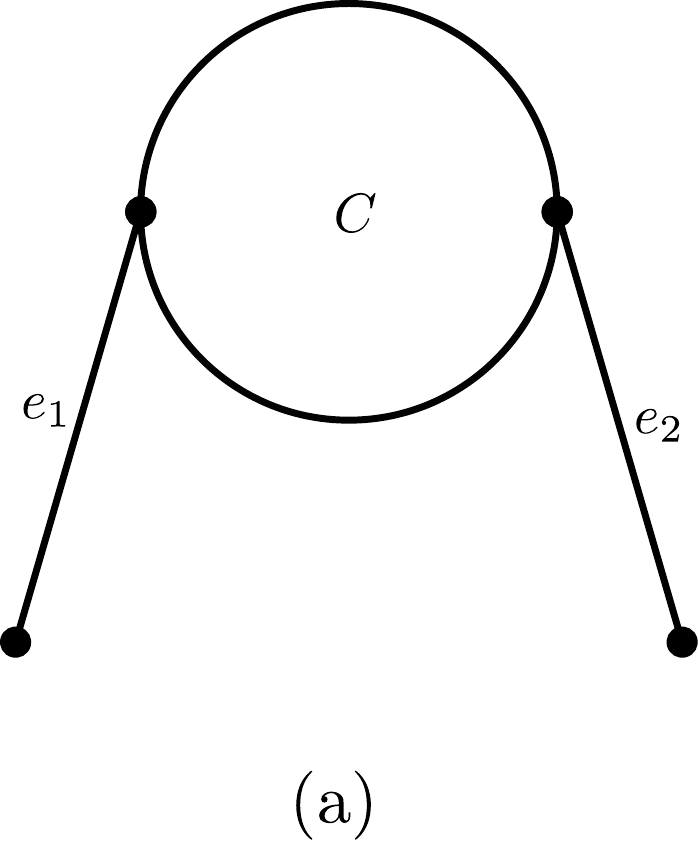}~~~~~~~~~~~~~~~~~~~~~~~~~~\includegraphics[scale=0.5]{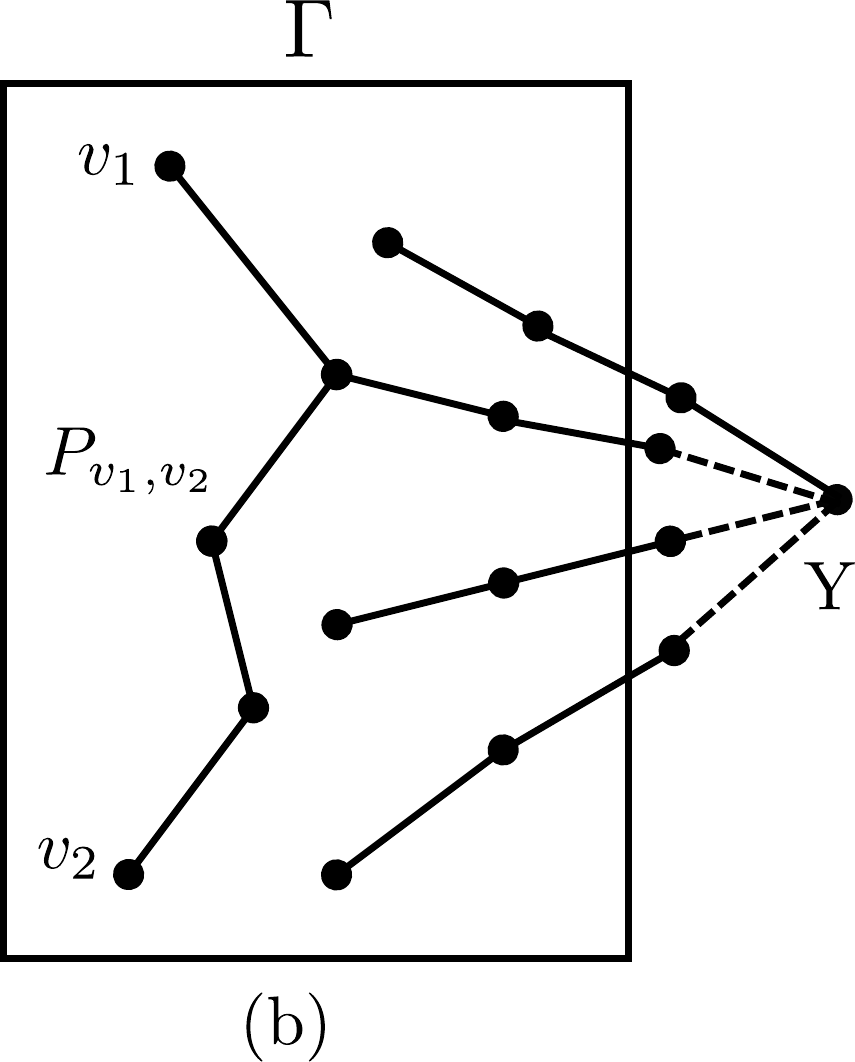}\end{center}

\caption{\label{fig:stabilization}(a) The relation between AB-phases, (b)
the stabilization of the first homology group.}
\end{figure}

\noindent A spanning set of
$H_1(\mathcal(D)^n(\Gamma))$ is by the following (see appendix for proof):
\begin{enumerate}
\item All $2$-particle cycles corresponding to the exchange of two particles
on the Y subgraph while $n-2$ particles are at vertices not belonging
to the considered Y-graph. In general the phases $\phi_{Y}$ depend
on the position of the remaining $n-2$ particles.
\item The set of $\beta_{1}(\Gamma)$ AB-cycles, where $\beta_{1}(\Gamma)$
is the number of the independent cycles of $\Gamma$. \end{enumerate}
\begin{theorem}\label{thm: 5}
\label{stabilization}For a $2$-connected graph $\Gamma$ the the
first homology group stabilizes with respect to the number of particles,
i.e. $H_{1}(\mathcal{D}^{n}(\Gamma))=H_{1}(\mathcal{D}^{2}(\Gamma))$.\end{theorem}
\begin{proof}
Using our spanning set it enough to show that phases on the
Y-cycles do not depend on the position of the remaining $n-2$ particles.
Notice that if any pair of the vertices not belonging to the chosen
Y-graph is connected by a path then clearly the corresponding Y-phases have this property. Since the graph $\Gamma$ is $2$-connected
it remains at least $1$-connected after removal of a vertex. Removing the central vertex of the Y (see figure
\ref{fig:stabilization}(b)), the theorem follows.\qed
\end{proof}

\section{N-particle statistics on $1$-connected graphs\label{sec:N-particle-statistics-on}}

By Theorem~\ref{thm: 5}, in order to fully characterize the first homology group of $\mathcal{D}^{n}(\Gamma)$ for an arbitrary graph $\Gamma$ we are left to calculate $H_1(\mathcal{D}^{n}(\Gamma))$ for graphs which are 1-connected but not 2-connected. This is achieved by considering $n$-particle star and fan graphs.

\subsection{Star graphs\label{sub:The-star-graphs}}

In the following we consider a particular family of $1$-connected graphs,
namely the star graphs $S_{E}$ with $E$ edges (see figure \ref{fig:The-star and fan}(a)).
Our aim is to provide a formula for the dimension of the first homology
group, $\beta_{n}^{E}$, of the $n$-particle configuration space $\mathcal{D}^{n}(S_{E})$. Let us recall that a graph $\Gamma$ is $1$-connected iff after deletion of one vertex it splits into at least two connected components.

\paragraph{Non-subdivided star graph}
It turns out that the computation of $\beta_n^E$ can be reduced to the case of $n$ particles on a non-subdivided star graph, so we consider this first.
Let $\bar{S}_{E}$ denote the star graph with $E+1$ vertices and $E$ edges each connecting the central vertex to a single vertex of valency $1$; such a star graph is not sufficiently subdivided for $n>2$ particles. As there are no pairs of disjoint edges (every edge contains the central vertex), there are no contractible cycles.
Therefore, the $n$-particle configuration
space, $\mathcal{D}^{n}(\bar{S}_E)$ is a graph, i.e. a one-dimensional cell complex. The number of independent cycles
in $\mathcal{D}^{n}(\bar{S}_E)$, denoted here and in what follows by $\gamma_n^E$,
is given by the first Betti number, $E_n - V_n + 1$, where $E_n$ and $V_n$ are the number of edges and vertices in $\mathcal{D}^{n}(S_E)$. It is
easy to see that  $V_{n}={E+1 \choose n}$
and $E_{n}=E\cdot{E-1 \choose n-1}$. Hence
\begin{gather}
\gamma_{n}^{E}=E {E-1 \choose n-1}-{E+1 \choose n}+1.\label{eq:2-particleY}
\end{gather}

\paragraph{Y-graph }

The simplest case of a sufficiently subdivided graph is a Y-graph where each arm has $n-1$ segments.  As there are no cycles on the Y-graph itself, cycles in the $n$-particle configuration space are generated by two-particle exchanges on the non-subdivided subgraph $\bar{Y}$ comprised of the three segments adjacent to the central vertex. A basis of independent cycles is obtained by taking all possible configurations of the $n-2$ particles amongst the three arms of the Y-graph.  As configurations which differ by shifting particles within the arms of the Y produce homotopic cycles, the number of distinct configurations is the number of partitions of $n-2$ indistinguishable particles amongst three distinguishable boxes, or ${(n-2) +(3-1) \choose n-2} = {n \choose n-2}$.  Therefore,

\begin{gather}
\beta_{n}^{3}= {n \choose n-2}\gamma_2^3=\frac{n(n-1)}{2}.\label{eq:Y}
\end{gather}

\paragraph{Star graph with five arms}

For star graphs with more than three arms, it is necessary to take account of relations between cycles involving two or more moving particles.  With this in mind, we introduce the following terminology: an $(n,m)$-cycle is a cycle of $n$ particles on which $m$ particles move and $(n-m)$ particles remain fixed.

The general case is well illustrated by considering the star graph with $E=5$ arms.  As above, we suppose that each arm of  $S_5$ has $(n-1)$ segments, and is therefore sufficiently subdivided to accommodate $n$ particles.  Let $\bar{S}_5$ denote the non-subdivided subgraph consisting of the five segments adjacent to the central vertex. As there are no cycles on $S_5$, a spanning set for the first homology group of the $n$-particle configuration space is provided by two-particle cycles on the Y's contained in $\bar{S}_5$.  The number of independent two-particle cycles on $\bar{S}_5$ is given by $\gamma_5^2$.  For each of these, we can distribute the remaining $(n-2)$ particles among the five edges of $S_5$  (cycles which differ by shifting particles within an edge  are homotopic).  Therefore, we obtain a spanning set consisting  of ${\beta''}_n^5$ $(n,2)$-cycles, where
\[ {\beta''}_n^5 :=  \binom{n+2}{4} \gamma_2^5 .\]

The preceding discussion of non-subdivided star graphs reveals that there are relations among the cycles in the spanning set.  In particular, a subset of the $(n,2)$-cycles can be replaced by a smaller number of $(n,3)$-cycles.  To see this, consider first the case of three particles on the non-subdivided star graph $\bar{S}_5$.
By definition, the number of independent $(3,3)$-cycles is $\gamma_3^5$.  However, the number of  $(3,2)$-cycles on $\bar{S}_5$ is larger; it is given by $\binom{5}{1} \gamma^4_2$, where the first factor represents the number of positions of the fixed particle, and the second factor represents the number of independent $(2,2)$-cycles on the remaining four edges of $\bar{S}_5$.  It is easily checked that $\gamma_3^5 - \binom{5}{1} \gamma^4_2 = -3$, so that there are three relations amongst the $(3,2)$-cycles on $\bar{S}_5$. 

For each $(3,3)$-cycle on $\bar{S}_5$, there are $\binom{n+1}{4}$ $(n,3)$-cycles on $S_5$; the factor  $\binom{n+1}{4}$ is the number of ways to distribute the $n-3$ fixed  particles on the five edges of $S_5$ outside of $\bar{S}_5$. The corresponding calculation of the number of $(n,2)$-cycles on $S_5$ obtained from $(3,2)$-cycles on $\bar{S}_5$ requires more care. The preceding reasoning would suggest that the number of such $(n,2)$-cycles is given by $\binom{n+1}{4} \binom{5 }{1} \gamma_2^4$.  However,  this expression introduces some double counting.  In particular,  $(n,2)$-cycles for which two of the fixed particles lie  in $\bar{S}_5$ are  counted twice, as each of these two fixed particles is separately regarded as the fixed particle in a $(3,2)$-cycle on ${\bar  S}_5$.  The correct expression is obtained by subtracting the number of doubly counted cycles, i.e.  $\binom{n}{4} \binom{5}{2} \gamma_2^3$. Thus we may replace  this subset of $(n,2)$-cycles by the $(n,3)$-cycles to which they are related to obtain a smaller spanning set with ${\beta'}_n^5 $ elements, where
\[ {\beta'}_n^5 =  {\beta''}_n^5 +  \binom{n+1}{4} \gamma_3^5 - \left( \binom{n+1}{4} \binom{5} {1} \gamma_2^4 - \binom{n}{4} \binom{5} {2}\gamma_2^3\right).\]

Finally, we must account for relations among the $(n,3)$-cycles.  Consider first the case of just four particles on $\bar{S}_5$.  The number of independent $(4,4)$-cycles is $\gamma_4^5$.  The number of $(4,3)$-cycles 
is $\binom{5}{1} \gamma^4_3$, where the first factor represents the number of positions of the fixed particle, and the second factor represents the number of independent $(3,3)$-cycles on the remaining four edges of $\bar{S}_5$.  For each  $(4,4)$-cycle on $\bar{S}_5$, there are $\binom{n}{4}$ $(n,4)$ cycles  on $S_5$. Similarly, for each $(4,3)$-cycle on $\bar{S}_5$, there are $\binom{n}{4}$ $(n,3)$-cycles on $S_5$ (there is no over-counting, as there are no five-particle cycles on $\bar{S}_5$).  Replacing this subset of $(n,3)$-cycles by the $(n,4)$-cycles to which they are related, we get a smaller spanning set of $\beta_n^5$ elements, where
\[ \beta_n^5 =  {\beta'}_n^5 +  \binom{n}{4}\left( \gamma_4^5 - \binom{5}{1} \gamma_3^4\right) = 6\binom{n+2}{4} - 4\binom{n+1}{4} + \binom{n}{4}.\]
As there are no five-particle cycles on $\bar{S}_5$, there are no additional relations, and the resulting spanning set constitutes a basis.

\paragraph{$n$ particles on a star graph with $E$ arms}

The formula in the general case of $E$ edges is obtained following a similar argument.  We start with a spanning set of $\binom{n+E-3}{E-1} \gamma_2^E$ $(n,2)$-cycles on $S_E$.  We then replace a subset of $(n,2)$-cycles by a smaller number of $(n,3)$-cycles, then replace a subset of these $(n,3)$-cycles by a smaller number of $(n,4)$-cycles, and so on, proceeding to $(n,E-1)$-cycles, thereby obtaining a basis. The number of elements in the basis is given by
\begin{equation} \label{eq: first beta expression} \beta_n^E = \sum_{m=2}^{E-1}\left( \binom{n - m + E -1}{E-1} \gamma_m^E + \sum_{j = 1}^{E-m}  (-1)^j \binom{n - m-j + E}{E-1} \binom{E}{j} \gamma_{m-1}^{E-j}\right).\end{equation}
The outer $m$-sum is taken over $(n,m)$-cycles.  The $m$th term is the difference between the number of $(n,m)$-cycles and the number of $(n,m-1)$-cycles to which they are related.  The inclusion-exclusion sum over $j$ compensates for over-counting $(n,m-1)$-cycles with $j$ fixed particles in ${\bar S}_E$.

It  turns out to be convenient to rearrange the sums in \eqref{eq: first beta expression} to obtain the following equivalent expression:
\begin{equation} \label{eq: second beta expression}
\beta_{n}^{E}=\sum_{k=2}^{E-1}{n-k+E-1 \choose E-1} \alpha^E_{k} \end{equation}
where
\begin{equation}\label{eq:alpha_k}
\alpha_{k}^{E}=\sum_{i=0}^{k-2}(-1)^{i}{E \choose i}\cdot\gamma_{k-i}^{E-i}.
\end{equation}
This is because the coefficients $\alpha_k^E$ turn out to have a simple expression.
First, straightforward manipulation yields
\begin{gather}
\alpha_{k}^{E}=\gamma_{k}^{E}-\sum_{i=1}^{k-2}{E \choose i}\alpha_{k-i}^{E-i}.\label{eq:alpha_k1}
\end{gather}
We then have the following:
%
%
%
%
%
%
%
%

\begin{lemma}
\label{lemma4}The coefficients $\alpha_{k}^{E}=(-1)^{k}{E-1 \choose k}$.\end{lemma}
\begin{proof}
We proceed by induction. Direct calculations give $\alpha_{2}={E-1 \choose 2}$.
Assume that $\alpha_{i}^{E}=(-1)^{i}{E-1 \choose i}$ for $i\in\{2,\ldots,k-1\}$
and $k\leq E$. Using this assumption and (\ref{eq:alpha_k1})

\begin{gather*}
\alpha_{k}=\gamma_{k}^{E}-(-1)^{k}\sum_{i=1}^{k-2}(-1)^{i}{E \choose i}{E-i-1 \choose k-i}.
\end{gather*}
Making use of the identity ${r \choose k}=(-1)^{k}{k-r-1 \choose k}$
and Vandermonde's convolution $\sum_{i=0}^{k}{E \choose i}{k-E \choose k-i}=1$, we get

\begin{gather*}
(-1)^{k}\sum_{i=1}^{k-2}(-1)^{i}{E \choose i}{E-i-1 \choose k-i}=\sum_{i=1}^{k-2}{E \choose i}{k-E \choose k-i}\\
=1-(-1)^{k}{E-1 \choose k}+(E-k){E \choose k-1}-{E \choose k}\,.
\end{gather*}
Using (\ref{eq:2-particleY}) for $\gamma_k^E$, we get
\begin{gather*}
\alpha_{k}=(-1)^{k}{E-1 \choose k}+E{E-1 \choose k-1}-{E+1 \choose k}-(E-k){E \choose k-1}+{E \choose k}\,.
\end{gather*}
Expanding ${E+1 \choose k}={E \choose k}+{E \choose k-1}$ and straightforward manipulations show
\begin{gather*}
\alpha_{k}=(-1)^{k}{E-1 \choose k},
\end{gather*}
which completes the argument.\qed
\end{proof}
\noindent By Lemma \ref{lemma4}
\begin{gather*}
\beta_{n}^{E}=\sum_{k=2}^{E-1}{n-k+E-1 \choose E-1}\cdot\alpha_{k}=\sum_{k=2}^{E-1}\left(-1\right)^{k}{E-1 \choose k}{n-k+E-1 \choose E-1}\\
=\sum_{k=2}^{E-1}\left(-1\right)^{k}{E-1 \choose k}{n-k+E-1 \choose n-k}=(-1)^{n}\sum_{k=2}^{E-1}{E-1 \choose k}{-E \choose n-k}\,.
\end{gather*}
By Vandermonde's convolution

\begin{gather*}
\sum_{k=0}^{E-1}{E-1 \choose k}{-E \choose n-k}=\sum_{k=0}^{n}{E-1 \choose k}{-E \choose n-k}={-1 \choose n}=(-1)^{n}.
\end{gather*}
Therefore

\begin{gather*}
\beta_{n}^{E}=1-{n+E-1 \choose E-1}+{n+E-2 \choose E-1}\left(E-1\right).
\end{gather*}
Notice that ${n+E-1 \choose E-1}={n+E-2 \choose E-1}+{n+E-2 \choose E-2}$
and thus

\begin{gather}
\beta_{n}^{E}={n+E-2 \choose E-1}\left(E-2\right)-{n+E-2 \choose E-2}+1.\label{eq:star-n}
\end{gather}

\begin{figure}[h]
\begin{center}~~~~~~~~~\includegraphics[scale=0.6]{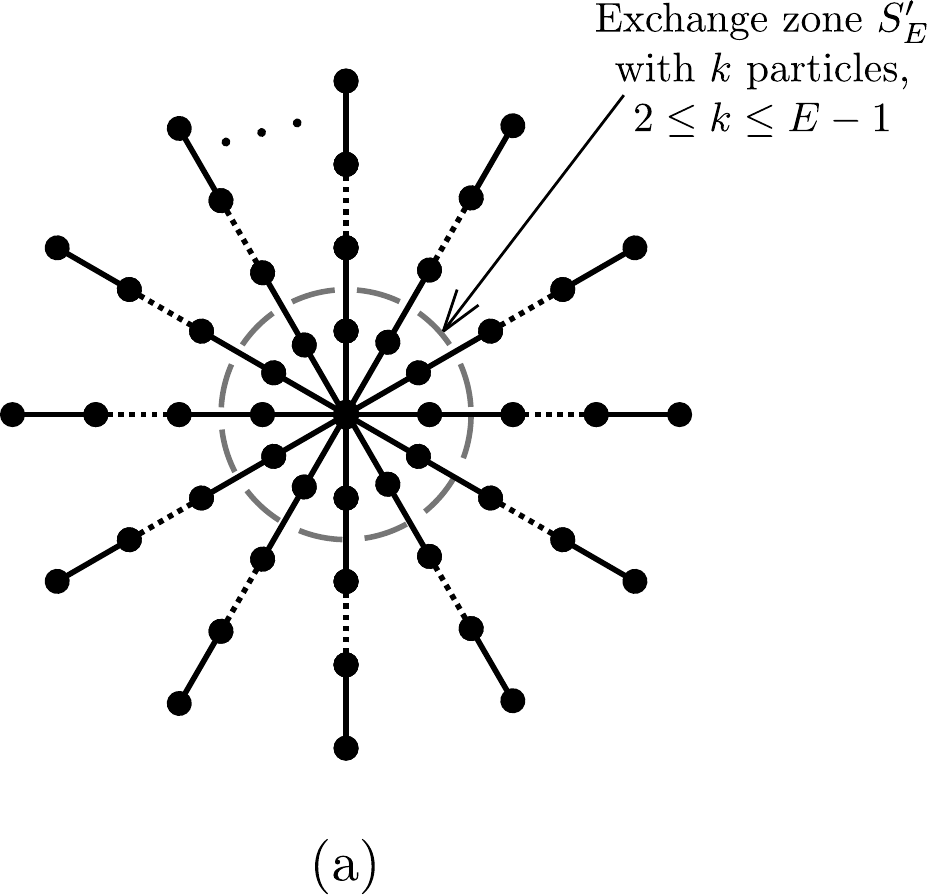}~~~~~~~~~~~\includegraphics[scale=0.6]{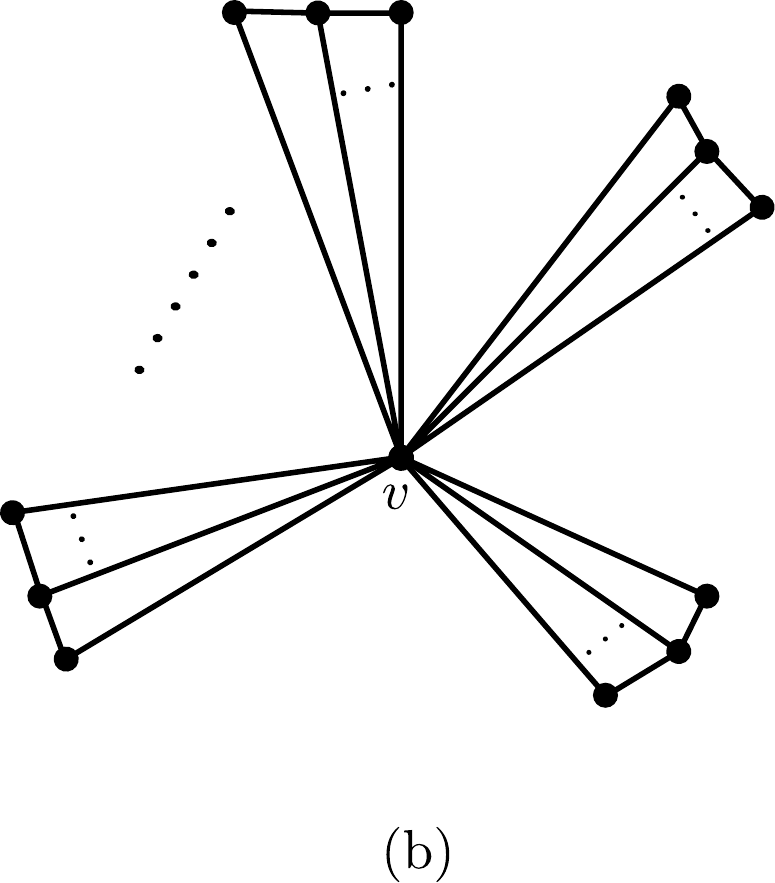}\end{center}

\caption{\label{fig:The-star and fan}(a) The star graph with $E$ arms and
$n$ particles. Each arm has $n$ vertices. The exchange zone $S_{E}^{\prime}$
can accommodate $2$, $3$,...,$E-1$ particles. (b) The fan graph
$F$.}
\end{figure}

\noindent Note finally that in contrast with $2$-connected graphs,
formula (\ref{eq:star-n}) indicates a strong dependence of the
quantum statistics on the number of particles, $n$.

\subsection{The fan graphs }

Following the argument presented in section \ref{sub:One-connected-graphs}
in order to treat a one-vertex cut $v$ we need to count the number
of the independent Y-phases which are lost due to the removal of $v$. As in Section \ref{sub:One-connected-graphs}, let $\mu = \mu(v)$ denote the number of connected components following the deletion of $v$, and denote these components by $\Gamma_1, \ldots, \Gamma_\mu$.  For Y-cycles with edges in three distinct components, the number  of independent phases, $\beta_n^\mu$, is given by the expression (\ref{eq:star-n}) for star graphs,
\begin{gather}
\beta_{n}^{\mu}={n+\mu-2 \choose \mu-1}\left(\mu-2\right)-{n+\mu-2 \choose \mu-2}+1.\label{eq:fan-star}
\end{gather}
We must also determine the number of independent Y-cycles with two edges in the same component $\Gamma_{i}$, denoted  $\gamma_n(v)$ .

Let us first consider a simple example, namely
 the graphs shown in figures \ref{fig:fan-ex}(a)
and \ref{fig:fan-ex}(b). Assume there are three particles. We 
calculate $\gamma_3(v)$ as follows. 
The $Y$ subgraphs we are interested in are denoted
by dashed lines and are $Y_{1}$ and $Y_{2}$ respectively. Note that
each of them contributes three phases corresponding to different
positions of the third particle $\{\phi_{Y_{1}}^{A},\phi_{Y_{1}}^{B},\phi_{Y_{1}}^{C},\phi_{Y_{2}}^{A},\phi_{Y_{2}}^{B},\phi_{Y_{2}}^{C}\}$.
They are, however, not independent. To see this, note that using Lemma
\ref{aspect1-1} we can write
\begin{gather*}
\phi_{c,3}=\phi_{Y_{1}}^{A}+\phi_{Y_{1}}^{B}+\phi_{c,1}^{B,B^{\prime}},\,\,\,\phi_{c,3}=\phi_{Y_{2}}^{A}+\phi_{Y_{2}}^{C}+\phi_{c,1}^{C,C^{\prime}}\,,\\
\phi_{c,2}^{B}=\phi_{Y_{1}}^{B}+\phi_{c,1}^{B,B^{\prime}},\,\,\,\phi_{c,2}^{B}=\phi_{Y_{2}}^{B}+\phi_{c,1}^{B,C}\,,\\
\phi_{c,2}^{C}=\phi_{Y_{1}}^{C}+\phi_{c,1}^{B,C},\,\,\,\phi_{c,2}^{C}=\phi_{Y_{2}}^{C}+\phi_{c,1}^{C,C^{\prime}}\,.
\end{gather*}
The phase $\phi_{c,3}$ is not lost when $v$ is cut. 
On the other hand, the five phases
\begin{gather}
\{\phi_{c,1}^{C,C^{\prime}},\,\phi_{c,1}^{B,B^{\prime}},\,\phi_{c,1}^{B,C},\,\phi_{c,2}^{B},\,\phi_{c,2}^{C}\},
\end{gather}
are lost. The knowledge of them and $\phi_{c}^{3}$ determines all
six $\phi_{Y}$ phases. Therefore, $\gamma_3(v)$ is the number of $1$ and $2$-particle exchanges on cycle
$c$ (which is $5$) rather than the number of $Y$ phases (which
is $6$).

\begin{figure}[h]
\begin{center}~~~~~~~~~~~~\includegraphics[scale=0.6]{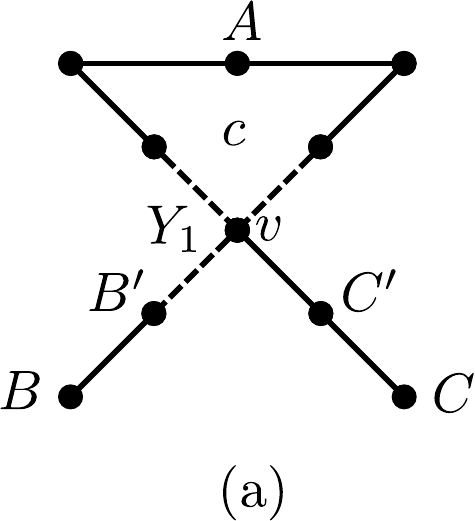}~~~~~~~~~~~~~~~~~~\includegraphics[scale=0.6]{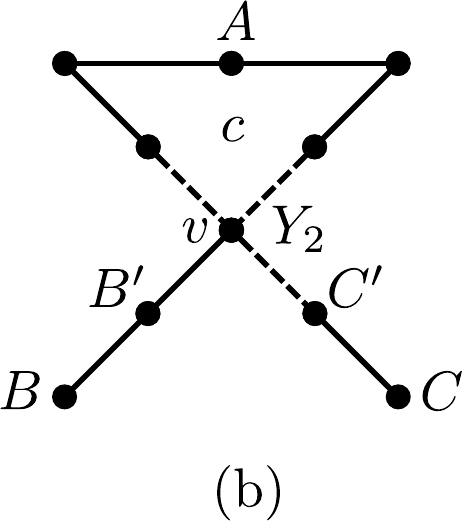}\end{center}

\caption{\label{fig:fan-ex}The $Y$ subgraphs (a) $Y_{1}$ and (b)$Y_{2}$.}
\end{figure}

For the general case, let $\nu_i$ denote the number of edges at $v$ which belong to $\Gamma_i$.  Since the $\Gamma_i$ are connected, there exist $\nu_i - 1$ independent cycles in $\Gamma_i$ which connect these edges.  Denote these by $C_{i,1},\ldots, C_{1,\nu_1-1}$.  Fan graphs (see Fig~\ref{fig:The-star and fan} (b)) provide the simplest realization.  Using arguments similar to those in the above example, one can show that Y-cycles with two edges in the same component  can be expressed in terms of two sets of cycles. The first set contains cycles which are wholly contained in just one of the connected components. These cycles are not lost when $v$ is cut, and therefore do not contribute to $\gamma_n(v)$.  The second type of cycle is characterised as follows:  Consider a partition $\{n_i\}_{i=1}^{\mu}$ of the particles amongst the components $\Gamma_i$.  For each partition, we can construct cycles where all of the particles in $\Gamma_i$ -- assuming $\Gamma_i$ contains at least one particle, i.e.~that $n_i > 0$ -- are taken to move once around  $C_{i,j}$ while the other particles remain fixed.  Excluding the cases in which all of the particles belong to a single component,
 the number of such cycles is given by the following sum over partitions $n_1 +\cdots + n_\mu = n$:
\[ \gamma_n(v) =  \sum_{\scriptstyle
n_1,\ldots,n_\mu = 0 
\\ \atop \scriptstyle
n_1 + \cdots + n_\mu = n}^n \
\sum_{\scriptstyle
i = 1   
\\ \atop \scriptstyle
0 < n_i < n}^\mu (\nu_i - 1).\]
Noting that
\[ \sum_{\scriptstyle
i = 1   
\\ \atop \scriptstyle
0 < n_i < n}^\mu  =
\sum_{i = 1}^\mu \ - \
\sum_{\scriptstyle
i = 1   
\\ \atop \scriptstyle  n_i = 0}^\mu
\ - \
\sum_{\scriptstyle
i = 1  
\\ \atop \scriptstyle n_i = n}^\mu\]
and $\sum_{i = 1}^\mu (\nu_i - 1) = \nu - \mu$,
we readily obtain
\[  \gamma_n(v) = \left( \binom{n+\mu -1}{n} - \binom{n+\mu-2}{n} - 1\right)(\nu-\mu) = \left( \binom{n+\mu -2}{n-1} - 1\right)(\nu-\mu).\]
Hence the  number of the phases lost when $v$ is cut is given by
\begin{gather}
N_{1}(v,n)= \beta^\mu_n + \gamma_n(v) = {n+\mu-2 \choose \mu-1}\left(\nu-2\right)-{n+\mu-2 \choose \mu-2}-\left(\nu-\mu-1\right).\label{eq:one-connected-n}
\end{gather}

\paragraph{The final formula for $H_1(\mathcal{D}^{n}(\Gamma))$}

By the repeated application of the one-vertex cuts the resulting components
of $\Gamma$ become finally $2$-connected graphs. Let $v_{1},\ldots,v_{l}$
be the set of cut vertices such that components $\Gamma_{v_{i},k}$
are $2$-connected. Making use of formula (\ref{eq:2-connected})
we write

\begin{gather}
H_{1}(\mathcal{D}^{n}(\Gamma))=\mathbb{Z}^{\beta(\Gamma)+N_{1}+N_{2}+N_{3}}\oplus\mathbb{Z}_{2}^{N_{3}^{\prime}},\label{eq:2-particle-final-1}
\end{gather}
 where $N_{1}=\sum_{i}N_{1}(v_{i},n)$, the coefficients $N_{1}(v_i,n)$
are given by (\ref{eq:one-connected-n}) and $N_{2}$, $N_{3}$, $N_{3}^{\prime}$
are defined as in section \ref{sec:Two-particle-quantum-statistics}.

\section{Gauge potentials for $2$-connected graphs}

In this section we give a prescription for the $n$-particle topological gauge potential on $\mathcal{D}^n(\Gamma)$ in terms of the $2$-particle topological gauge potential. For $2$-connected graphs all choices of $n$-particle topological gauge potentials on $\mathcal{D}^n(\Gamma)$  are realized by this prescription. The discussion is divided into three parts: i) separation of a $2$-particle topological gauge potential into AB and quantum statistics components, ii) topological gauge potentials for 2-particles on a subdivided graph,  iii) $n$-particle topological gauge potentials.

We start with some relevant background. Assume as previously that $\Gamma$ is sufficiently subdivided. Recall that directed edges or $1$-cells of $\mathcal{D}^n(\Gamma)$ are of the form $v_1\times\ldots\times v_{n-1}\times e$ up to permutations, where $v_j$ are vertices of $\Gamma$ and $e=j\rightarrow k$ is an edge of $\Gamma$ whose endpoints are not $\{v_1,\ldots, v_{n-1}\}$. For simplicity we will use the following notation
\[
\{v_1,\ldots,v_{n-1},j\rightarrow k\}:=v_1\times\ldots\times v_{n-1}\times e.
\]
An $n$-particle gauge potential is a function $\Omega^{(n)}$ defined on the directed edges of $\mathcal{D}^n(\Gamma)$ with the values in $\mathbb{R}_{+}$ modulo $2\pi$ such that
\begin{equation}\label{asym}
\Omega^{(n)}(\{v_1,\ldots,v_{n-1},k\rightarrow j\})=-\Omega^{(n)}(\{v_1,\ldots,v_{n-1},j\rightarrow k\}).
\end{equation}
In order to define $\Omega$ on linear combinations of directed edges we extend (\ref{asym}) by linearity.

For a given gauge potential, $\Omega^{(n)}$ the sum of its values calculated on the directed edges of an oriented cycle $C$ will be called the flux of $\Omega$ through $C$ and denoted $\Omega(C)$. Two gauge potentials $\Omega_1^{(n)}$ and $\Omega_2^{(n)}$ are called equivalent if for any oriented cycle $C$ the fluxes $\Omega_1^{(n)}(C)$ and $\Omega_2^{(n)}(C)$ are equal modulo $2\pi$.

The $n$-particle gauge potential $\Omega^{(n)}$ is called a {\it topological gauge potential} if for any contractible oriented cycle $C$ in $\mathcal{D}^n(\Gamma)$ the flux $\Omega^{(n)}(C)=0\,\mathrm{mod}\,2\pi$. It is thus clear that equivalence classes of topological gauge potentials are in 1-1 correspondence with the equivalence classes in $H_1(\mathcal{D}^n(\Gamma))$.

\paragraph{Pure Aharonov-Bhom  and pure quantum statistics topological gauge potentials}
Let $\Gamma$ be a graph with $V$ vertices. We say that a 2-particle gauge potential $\Omega^{(2)}_{AB}$ is a {\it pure Aharonov-Bohm  gauge potential} if and only if
\begin{equation}
\label{ eq: AB gauge potential}
\Omega^{(2)}_{AB}(\{i, j\rightarrow k\}) = \omega^{(1)}(j \rightarrow k), \text { for all distinct vertices $i,j,k$ of $\Gamma$}.
\end{equation}
Here $\omega^{(1)}$ can be regarded as a gauge potential on $\Gamma$. Thus, for a pure AB gauge potential, the phase associated with one particle moving from $j$ to $k$ does not depend on where the other particle is. We say that a 2-particle gauge potential $\Omega^{(2)}_{S}$ is a {\it pure statistics gauge potential} if and only if
\begin{equation}
\label{ eq: AB gauge potential}
\sum_{\scriptstyle i\atop \scriptstyle i \neq j,k} \Omega^{(2)}_{S}(\{i, j\rightarrow k\}) = 0, \text { for all distinct vertices $j,k$ of $G$}.
\end{equation}
That is, the phase associated with one particle moving from $j$ to $k$ averaged over all possible positions of the other particle is zero.  It is clear that an arbitrary gauge potential $\Omega^{(2)}$ has a unique decomposition into a pure AB and pure statistics gauge potentials, i.e.
\begin{equation}
\label{eq: decomposition }
\Omega^{(2)} = \Omega^{(2)}_{AB} + \Omega^{(2)}_{S},
\end{equation}
where
\begin{equation}
\label{eq: decomposition 2}
 \Omega^{(2)}_{AB}(\{i, j\rightarrow k\}) = \frac{1}{V-2} \sum_{\scriptstyle p\atop \scriptstyle p \neq j,k} \Omega^{(2)}(\{p,j\rightarrow k\}), \quad \Omega^{(2)}_S =\Omega^{(2)} - \Omega^{(2)}_{AB}.
\end{equation}
It is straightforward to verify that if $\Omega^{(2)}$ is a topological gauge potential, then so are $\Omega^{(2)}_{AB}$ and $\Omega^{(2)}_S$, and vice versa. Moreover, one can easily check that $\Omega_{AB}^{(2)}$ vanishes on any Y-cycle of $\mathcal{D}^2(\Gamma)$. Note, however, that for a given cycle $C$ of $\Gamma$ the AB-phase, $\phi^{v}_{C,1}$ considered in the previous sections is not $\Omega^{(2)}_{AB}(v\times C)$ but rather $\Omega^{(2)}(v\times C)$ as AB-phases can depend on the position of the stationary particle.

\paragraph{Gauge potential for a subdivided 2-particle graph}\label{sec: gauge potential}
Let $\Gbar$ be a graph with vertices $\Vcalbar = \{1,\ldots, \Vbar\}$. Let $\Omegabar^{(2)}$ be a gauge potential on $\mathcal{D}^2(\Gbar)$.

We assume that $\Omegabar^{(2)}$ is topological, that is, for every pair of disjoint edges of $\Gbar$, $i\leftrightarrow k$ and $j\leftrightarrow l$ we have
\begin{equation}
\label{eq: relation 0 }
\Omegabar^{(2)}(i,j\rightarrow l) + \Omegabar^{(2)}(l,i\rightarrow k) +\Omegabar^{(2)}(k,l\rightarrow j) + \Omegabar^{(2)}(j,k\rightarrow i) = 0.
\end{equation}
Assume we add a vertex to $\Gbar$ by subdividing an edge.  Let $p$ and $q$ denote the vertices of this edge, and denote the new graph by $\Gamma$ and  the added vertex by $a$. Since subdividing an edge does not change the topology of a graph, it is clear that we can find a gauge potential, $\Omega^{(2)}$, on $\mathcal{D}^2(\Gamma)$ that is, in some sense, equivalent to $\Omegabar^{(2)}$.


For the sake of  completeness, we first give a precise definition of what it means for gauge potentials on $\mathcal{D}^2(\Gamma)$ and $\mathcal{D}^2(\Gbar)$ to be equivalent. Given a path $\Cbar$ on $\mathcal{D}^2(\Gbar)$, we can construct a path $P$ on $\mathcal{D}^2(\Gamma)$ by making the replacements
\begin{align}\label{eq: subs C_0 to C}
\{i,p\rightarrow q\} &\mapsto \{i,p\rightarrow a\rightarrow q\},\nonumber\\
\{i,q\rightarrow p\} &\mapsto \{i,q\rightarrow a\rightarrow p\}.
\end{align}
Similarly, given a path $P$ on $\mathcal{D}^2(\Gamma)$  we can construct a path $\Cbar$ on $\mathcal{D}^2(\Gbar)$ by making the following substitutions:
\begin{align}\label{eq: subs C to C_0}
\{i,p\rightarrow a\rightarrow p\} &\mapsto  \{i,p\},\nonumber\\
\{i,p\rightarrow a\rightarrow q\} &\mapsto \{i,p\rightarrow q\},\nonumber\\
\{i,q\rightarrow a\rightarrow p\} &\mapsto \{i,q\rightarrow p\},\nonumber\\
\{i,q\rightarrow a\rightarrow q\} &\mapsto \{i,q\}.
\end{align}
We say that $\Omega^{(2)}$ and $\Omegabar^{(2)}$ are equivalent if
\begin{equation}
\label{eq: equiv Omega Omega^0 }
 \Omega^{(2)}(P) = \Omegabar^{(2)}(\Cbar)
\end{equation}
whenever $P$ and $\Cbar$ are related as above.

Next we give an explicit prescription for $\Omega^{(2)}$.  For edges in $\mathcal{D}^2(\Gamma)$ that do not involve vertices on the subdivided edge, we take $\Omega^{(2)}$ to coincide with $\Omegabar^{(2)}$.  That is,  for $i, j, k$ all distinct from $p,a,q$, we take
\begin{equation}
\label{eq: Omega ijk }
\Omega^{(2)}(\{i,j\rightarrow k\}) = \Omegabar^{(2)}(\{i,j\rightarrow k\}).
\end{equation}
As $p$ and $q$ are not adjacent on $\Gamma$, we take
\begin{equation}
\label{eq: Omega ijk }
\Omega^{(2)}(\{i,p\rightarrow q\}) = 0.
\end{equation}
For edges on $\mathcal{D}^2(\Gamma)$ involving the subdivided segments $p\rightarrow a$ and $a \rightarrow q$,  we require that $\Omega^{(2)}(\{i,p\rightarrow a\})$ and  $\Omega^{(2)}(\{i,a\rightarrow q\})$ add up to give the  phase  $\Omegabar^{(2)}(i,p\rightarrow q)$ on the original edge.  The partitioning of the original phase between the subdivided segments amounts to a choice of gauge.  For definiteness, we will take the phases on the two halves of the subdivided edge to be the same, so that
\begin{equation} \label{eq: Omega p,q-> a}
\Omega^{(2)}(\{i, p\rightarrow a\})  =
\Omega^{(2)}(\{i, a\rightarrow q\})  = \frac12 \Omegabar^{(2)}(\{i,p\rightarrow q\}). 
\end{equation}

It remains to determine $\Omega^{(2)}$ for edges of $C_2(G)$ on which the stationary particle sits at the new vertex $a$.  This follows from requiring that $\Omega^{(2)}$ satisfy the relations
\begin{align}
\label{eq: relations}
\Omega^{(2)}(\{a,i\rightarrow j\}) + \Omega^{(2)}(\{j,a\rightarrow p\}) + \Omega^{(2)}(\{p,j\rightarrow i\}) + \Omega^{(2)}(\{i,p\rightarrow a\}) &= 0,\nonumber\\
\Omega^{(2)}(\{a,i\rightarrow j\}) + \Omega^{(2)}(\{j,a\rightarrow q\}) + \Omega^{(2)}(\{q,j\rightarrow i\}) + \Omega^{(2)}(\{i,q\rightarrow a\}) &= 0.
\end{align}
From  (\ref{eq: Omega p,q-> a}) and the antisymmetry property $\Omega^{(2)}(\{i,j\rightarrow k\}) = -\Omega(\{i,k\rightarrow j\})$, along with the relations \eqref{eq: relation 0 } satisfied by $\Omegabar^{(2)}$, it follows that these conditions are equivalent, and both are satisfied by taking
\begin{equation}\label{eq: Omega a i-> j}
    \Omega^{(2)}(a, i\rightarrow j)  = \half \left(\Omegabar^{(2)}(p,i\rightarrow j) + \Omegabar^{(2)}(q,i\rightarrow j)\right).
\end{equation}
Finally, when $i$ or $j$ coincide with one of the vertices $p$ or $q$ the expression should be
\begin{equation}\label{eq: Omega a p-> j}
    \Omega^{(2)}(\{a, q\rightarrow j\})  = \left(\Omegabar^{(2)}(\{p,q\rightarrow j\}) + \half \Omegabar^{(2)}(\{j,q\rightarrow p\})\right).
\end{equation}
It is then straightforward to verify that $\Omega^{(2)}(P) = \Omegabar^{(2)}(\Cbar)$ whenever $P$ and $\Cbar$ are related as in \eqref{eq: subs C_0 to C} and \eqref{eq: subs C to C_0} and that $\Omega^{(2)}$ is a topological gauge potential.

\paragraph{Construction of $n$-particle topological gauge potential}

Let $\Omegabar^{(2)}$ be a gauge potential on $\mathcal{D}^2(\Gbar)$.  By repeatedly applying the procedure from the previous paragraph, we can construct an equivalent gauge potential $\Omega^{(2)}$ on $\mathcal{D}^2(\Gamma)$, where $\Gamma$ is a sufficiently subdivided version of $\Gbar$, in which $n-2$ vertices are added to each edge of $\Gbar$.
We resolve $\Omega^{(2)}$ into its AB and statistics components $\Omega^{(2)}_{AB}$ and $\Omega^{(2)}_S$, as in \eqref{eq: decomposition }.  Suppose the pure AB component is described by the gauge potential $\omega^{(1)}$ on $\Gamma$.  We define the $n$-particle gauge potential, $\Omega^{(n)} $, on $\mathcal{D}^{n}(\Gamma)$ as follows. Given $(n+1)$ vertices of $\Gamma$, denoted $\{ v_1,\ldots, v_{n-1},i,j\}$, with $i\sim j$, we take
\begin{equation}
\label{eq: Omega^n }
\Omega^{(n)} \left(
\{ v_1,\ldots,v_{n-1}, i\rightarrow j\}
\right)
= \omega^{(1)}(i\rightarrow j) +
\sum_{r=1}^{n-1} \Omega^{(2)}_S(\{v_r, i\rightarrow j\}).
\end{equation}
That is, the phase associated with the one-particle move $i\rightarrow j$ is the sum of the AB-phase $\omega^{(1)}(i,j)$ and the two-particle statistics phases $\Omega_S^{(2)}(\{v_r, i\rightarrow j\})$ summed over the positions of the other particles.

Given that $\Omega^{(2)}$ is a topological gauge potential, let us verify that $\Omega^{(n)} $ is a topological gauge potential. Let $i\rightarrow k$ and $j\rightarrow l$ be distinct edges of $\Gamma$, and let $\{v_1,\ldots, v_{n-2}\}$ denote $(n-2)$ vertices of $\Gamma$ that  are distinct from $i$, $j$, $k$, $l$.  We need to verify if
\begin{gather*}
\Omega^{(n)} \left( \{v_1,\ldots, v_{n-2},i, j\rightarrow l\}\right) +
\Omega^{(n)} \left( \{v_1,\ldots, v_{n-2},l, i\rightarrow k\}\right) + \\
+\Omega^{(n)} \left( \{v_1,\ldots, v_{n-2},k, l\rightarrow j\}\right) +
\Omega^{(n)} \left( \{v_1,\ldots, v_{n-2},j, k\rightarrow i\}\right)=0.
\end{gather*}
Using \eqref{eq: Omega^n } it reduces to
\begin{small}
\begin{gather*}
\omega^{(1)} (i\rightarrow k) + \omega^{(1)} (k\rightarrow i) + \omega^{(1)} (j \rightarrow l) + \omega^{(1)} (l\rightarrow k) + \\
+
\left( \sum_{r = 1}^{n-2} \Omega^{(2)}_S(\{v_r, j\rightarrow l\}) +  \Omega^{(2)}_S(\{i, j\rightarrow l\}) \right) +
\left( \sum_{r = 1}^{n-2} \Omega^{(2)}_S(\{v_r, i\rightarrow k\}) +  \Omega^{(2)}_S(\{l, i\rightarrow k\}) \right)+ \\
+\left( \sum_{r = 1}^{n-2} \Omega^{(2)}_S(\{v_r, l\rightarrow j\}) +  \Omega^{(2)}_S(\{k, l\rightarrow j\}) \right) +
\left( \sum_{r = 1}^{n-2} \Omega^{(2)}_S(\{v_r, k\rightarrow i\}) +  \Omega^{(2)}_S(\{j, k\rightarrow i\}) \right). \\
\end{gather*}
\end{small}
Next, using the antisymmetry property $\Omega^{(2)}_S(\{v_r, i\rightarrow k\}) = -\Omega^{(2)}_S(\{v_r, k\rightarrow i\})$ and the fact that $\Omega^{(2)}_S$ is a topological gauge potential we get\begin{small}
\begin{gather*}
\sum_{r=1}^{n-2}  \left(\Omega^{(2)}_S(\{v_r, j\rightarrow l\}) + \Omega^{(2)}_S(\{v_r, l\rightarrow j\})\right) +
\left(\Omega^{(2)}_S(\{v_r, i\rightarrow k\}) + \Omega^{(2)}_S(\{v_r, k\rightarrow i\})\right) +\\
+ \Omega^{(2)}_S(\{i, j\rightarrow l\}) +  \Omega^{(2)}_S(\{l, i\rightarrow k\}) +  \Omega^{(2)}_S(\{k, l\rightarrow j\}) + \Omega^{(2)}_S(\{j, k\rightarrow i\})  = 0.
\end{gather*}\end{small}

Therefore, the gauge potential defined by \eqref{eq: Omega^n } is topological. Equivalence classes of n-particle topological gauge potentials are essentially elements of the first homology group $H_1(\mathcal{D}^2(\Gamma))$. By Theorem~\ref{thm: 5} the equivalence classes in $H_1(\mathcal{D}^n(\Gamma))$ are in 1-1 correspondence with equivalence classes in $H_1(\mathcal{D}^2(\Gamma))$. Hence, for $2$-connected graphs all choices of $n$-particle topological gauge potential on $\mathcal{D}^n(\Gamma)$  can be realized  by \eqref{eq: Omega^n }. Finally, note  that, as explained in \cite{JHJKJR}, having an $n$-particle topological gauge potential one can easily construct a tight-binding Hamiltonian which supports quantum statistics represented by it (see \cite{JHJKJR} for more details).

\section*{Acknowledgments}

We would like to thank Ram Band for helpful discussions. AS is supported by a University of Bristol
Postgraduate Research Scholarship and Polish Ministry of Science and Higher Education grant no. N N202
085840.  JMH would like to thank Bristol University for their hospitality during his sabbatical supported by the Baylor University research leave and
summer sabbatical programs during which
some of the work was carried out.

\section*{Appendix}

We present an argument which shows that the $n$-particle cycles given in sections \ref{sub:An-over-complete-basis} and
\ref{sub:An-over-complete-basis-2} form an over-complete spanning set of the first homology group $H_{1}(\mathcal{D}^{n}(\Gamma))$.  The argument follows the characterization of the fundamental group using discrete Morse theory by Farley and Sabalka \cite{FS05,FS08,FS12} or alternatively the characterization of the discrete Morse function for the $n$-particle graph \cite{S12}.  Here, however, we present the central idea in a way that does not assume a familiarity with discrete Morse theory in order to remain accessible.  For a rigorous proof we refer to the articles cited above.

Given a sufficiently subdivided graph $\Gamma$ we identify some maximal spanning subtree $T$ in $\Gamma$; $T$ is obtained by omitting exactly $\beta_1 (\Gamma)$ of the edges in $\Gamma$ such that $T$ remains connected but contains no loops.  The tree can then be drawn in the plane to fix an orientation.  A single vertex of degree $1$ in $T$ is identified as the root and the vertices of $T$ are labeled $1,2,\dots,|V|$ starting with $1$ for the root and labeling each vertex in turn traveling from the root around the boundary of $T$ clockwise, see figure \ref{fig:graphtotree}.

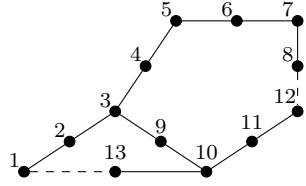
\begin{figure}[tbh]
\begin{center}
\begin{tikzpicture}[scale=0.2]
    \tikzstyle{every node}=[draw,circle,fill=black,minimum size=4pt,
                            inner sep=0pt]
    \draw (0,0) node (1) [label=120:$1$] {}
        -- (3,2) node (2) [label=120:$2$] {}
        -- (6,4) node (3) [label=120:$3$] {}
        -- (8,7) node (4) [label=120:$4$] {}
        -- (10,10) node (5) [label=120:$5$] {}
        -- (14,10) node (6) [label=120:$6$] {}
        -- (18,10) node (7) [label=120:$7$] {}
        -- (18,7) node (8) [label=120:$8$] {};
    \draw (3) -- (9,2) node (9) [label=90:$9$] {}
        -- (12,0) node (10) [label=90:$10$] {}
        -- (15,2) node (11) [label=90:$11$] {}
        -- (18,4) node (12) [label=120:$12$] {};
    \draw (10) -- (6,0) node (13) [label=90:$13$] {};

    \draw[dashed] (1) -- (13);
    \draw[dashed] (8) -- (12);
\end{tikzpicture}
\end{center}
\caption{\label{fig:graphtotree}A sufficiently subdivided graph for $3$ particles, edges in a maximal spanning tree are shown with solid lines and edges omitted to obtain the tree are shown with dashed lines.  Vertices are labeled following the boundary of the tree clockwise from the root vertex $1$.}
\end{figure}

To characterize a basis of $n$-particle cycles for the first homology group we fix a root configuration $\vx_0=\{1,2,\dots,n\}$ where the particles are lined up as close to the root as possible, see figure \ref{fig:Ybasis}(a).  The tree $T$ is used to establish a set of contactable paths between $n$-particle configurations on the graph (a discrete vector field). Given an $n$-particle configuration $\vx=\{v_1,\dots,v_n\}$ on the graph a path from $\vx$ to $\vx_0$ is a sequence of one-particle moves, where a single particle hops to an adjacent vacant vertex with the remaining $n-1$ particles remaining fixed.  This is a $1$-cell $\{v_1,\dots,v_{n-1},u\rightarrow v\}$ where $u$ and $v$ are the locations of the moving particle.  The labeling of the vertices in the tree provides a discrete vector field on the configuration space.  A particle moves according to the vector field if $n+1\rightarrow n$, i.e. the particle moves towards the root along the tree.  This allows a particle to move through a non-trivial vertex (a vertex of degree $\geq 3$) if the particle is coming from the direction clockwise from the direction of the root.  To define a flow that takes any configuration back to $\vx_0$ we also define a set of priorities at the non-trivial vertices that avoids $n$-particle paths crossing.  A particle may also move onto a non-trivial vertex $u$ according to the vector field if the $1$-cell $\{v_1,\dots,v_{n-1},u\rightarrow v\}$ does not contain a vertex $v_j$ with $v<v_j<u$; i.e. moving into a nontrivial vertex particles give way (yield) to the right.  So a particle can only move into the nontrivial vertex if there are no particles on branches of the graph between the branch the particle is on and the root direction clockwise from the root.  With this set of priorities it is clear that a path (sequence of $1$-cells) exists that takes any configuration $\vx$ to $\vx_0$ using only $1$-cells in the discrete vector field.  Equivalently by reversing the direction of edges in $1$-cells so that particles move away from the root we can move particles from the reference configuration $\vx_0$ to any configuration $\vx$ according to the flow.  As $n$-particle paths following this discrete flow do not cross these paths are contractible; equivalently we can always choose a basis of loops so that the phase around closed loops combining paths following the discrete flow is zero.

It remains to find a basis for the cycles that use $1$-cells not in the discrete vector field.  We see now that there are only two types of $1$-cells that are excluded; those where the edge $u\leftrightarrow v$ is one of the $\beta_1 (\Gamma)$ edges omitted from $\Gamma$ to construct $T$, and those where a particle moves through a non-trivial vertex out of order - without giving way to the right.

We first consider a $1$-cell $c_{u\rightarrow v}=\{v_1, \dots,v_{n-1}, u\rightarrow v\}$ where $u\leftrightarrow v$ is an omitted edge.  Such a $1$-cell is naturally associated with a cycle where the particles move from $\vx_0$ to $\{v_1,\dots,v_{n-1},u\}$ according to the vector field, then follow $c_{u\rightarrow v}$ and finally move back from $\{v_1,\dots,v_{n-1},v\}$ to $\vx_0$ according to the vector field.  These $n$-particle cycles are typically the AB-cycles where one particle moves around a loop in $\Gamma$ with the other particles at a given configuration.  We saw in section \ref{sub:An-over-complete-basis} that while the phase associated with an AB-cycle can depend on the position of the other particles, these phases can be parameterized by only $\beta_1(\Gamma)$ independent parameters; one parameter for those cycles using each omitted edge.

We now consider, instead, cycles that include a $1$-cell $c=\{v_1, \dots,v_{n-1},u \rightarrow v)\}$ where a particle moves out of order at a nontrivial vertex.  Again each such $1$-cell is naturally associated to a cycle $C$ through $\vx_0$ where the particle moves according to the vector field except when it uses the $1$-cell $c$.  Such a cycle is shown in figure \ref{fig:Ybasis}.

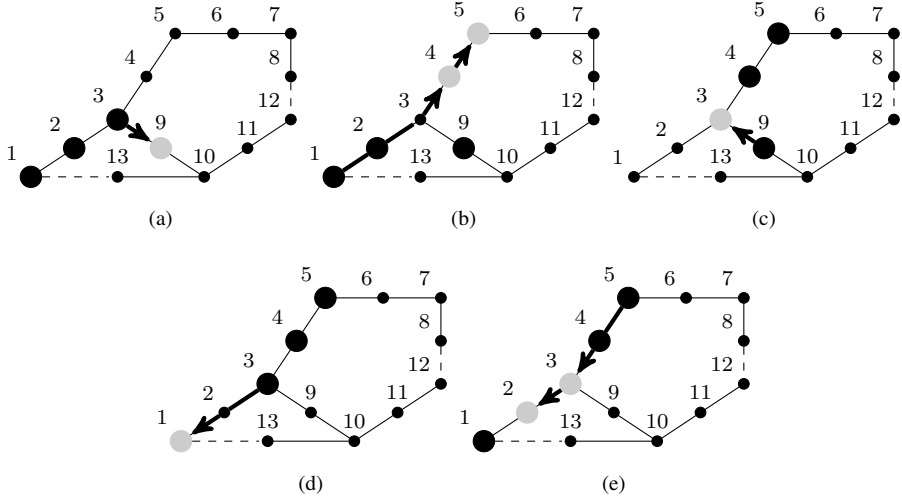
\begin{figure}[tbh]
\begin{center}
\begin{tikzpicture}[scale=0.19,>=stealth',shorten >=1pt]
    \tikzstyle{graph node}=[draw,circle,fill=black,minimum size=4pt,inner sep=0pt]

    \tikzstyle{blue node}=[draw,circle,fill=black,minimum size=8pt,inner sep=0pt]
    \tikzstyle{green node}=[draw,circle,fill=black,minimum size=8pt,inner sep=0pt]
    \tikzstyle{red node}=[draw,circle,fill=black,minimum size=8pt,inner sep=0pt]
    \tikzstyle{fblue node}=[draw=black!20,circle,fill=black!20,minimum size=8pt,inner sep=0pt]
    \tikzstyle{fgreen node}=[draw=black!20,circle,fill=black!20,minimum size=8pt,inner sep=0pt]
    \tikzstyle{fred node}=[draw=black!20,circle,fill=black!20,minimum size=8pt,inner sep=0pt]

    \draw (9,-3) node {(a)};

    \draw (0,0) node[blue node] (1) [label=120:$1$] {}
        -- (3,2) node[green node] (2) [label=120:$2$] {}
        -- (6,4) node[red node] (3) [label=120:$3$] {}
        -- (8,7) node[graph node] (4) [label=120:$4$] {}
        -- (10,10) node[graph node] (5) [label=120:$5$] {}
        -- (14,10) node[graph node] (6) [label=120:$6$] {}
        -- (18,10) node[graph node] (7) [label=120:$7$] {}
        -- (18,7) node[graph node] (8) [label=120:$8$] {};
    \draw (3) -- (9,2) node[fred node] (9) [label=90:$9$] {}
        -- (12,0) node[graph node] (10) [label=90:$10$] {}
        -- (15,2) node[graph node] (11) [label=90:$11$] {}
        -- (18,4) node[graph node] (12) [label=120:$12$] {};
    \draw (10) -- (6,0) node[graph node] (13) [label=90:$13$] {};

    \draw[dashed] (1) -- (13);
    \draw[dashed] (8) -- (12);

    \draw[black, ultra thick,->] (3) -- (9);
\end{tikzpicture}
\begin{tikzpicture}[scale=0.19,>=stealth',shorten >=1pt]
    \tikzstyle{graph node}=[draw,circle,fill=black,minimum size=4pt,
                            inner sep=0pt]

    \tikzstyle{blue node}=[draw,circle,fill=black,minimum size=8pt,inner sep=0pt]
    \tikzstyle{green node}=[draw,circle,fill=black,minimum size=8pt,inner sep=0pt]
    \tikzstyle{red node}=[draw,circle,fill=black,minimum size=8pt,inner sep=0pt]
    \tikzstyle{fblue node}=[draw=black!20,circle,fill=black!20,minimum size=8pt,inner sep=0pt]
    \tikzstyle{fgreen node}=[draw=black!20,circle,fill=black!20,minimum size=8pt,inner sep=0pt]
    \tikzstyle{fred node}=[draw=black!20,circle,fill=black!20,minimum size=8pt,inner sep=0pt]

    \draw (9,-3) node {(b)};

    \draw (0,0) node[blue node] (1) [label=120:$1$] {}
        -- (3,2) node[green node] (2) [label=120:$2$] {}
        -- (6,4) node[graph node] (3) [label=120:$3$] {}
        -- (8,7) node[fblue node] (4) [label=120:$4$] {}
        -- (10,10) node[fgreen node] (5) [label=120:$5$] {}
        -- (14,10) node[graph node] (6) [label=120:$6$] {}
        -- (18,10) node[graph node] (7) [label=120:$7$] {}
        -- (18,7) node[graph node] (8) [label=120:$8$] {};
    \draw (3) -- (9,2) node[red node] (9) [label=90:$9$] {}
        -- (12,0) node[graph node] (10) [label=90:$10$] {}
        -- (15,2) node[graph node] (11) [label=90:$11$] {}
        -- (18,4) node[graph node] (12) [label=120:$12$] {};
    \draw (10) -- (6,0) node[graph node] (13) [label=90:$13$] {};

    \draw[dashed] (1) -- (13);
    \draw[dashed] (8) -- (12);

    \draw[black, ultra thick,->] (4) -- (5);
    \draw[black, ultra thick,->] (3) -- (4);
    \draw[black, ultra thick] (1) -- (2) -- (3);
\end{tikzpicture}
\begin{tikzpicture}[scale=0.19,>=stealth',shorten >=1pt]
    \tikzstyle{graph node}=[draw,circle,fill=black,minimum size=4pt,
                            inner sep=0pt]

    \tikzstyle{blue node}=[draw,circle,fill=black,minimum size=8pt,inner sep=0pt]
    \tikzstyle{green node}=[draw,circle,fill=black,minimum size=8pt,inner sep=0pt]
    \tikzstyle{red node}=[draw,circle,fill=black,minimum size=8pt,inner sep=0pt]
    \tikzstyle{fblue node}=[draw=black!20,circle,fill=black!20,minimum size=8pt,inner sep=0pt]
    \tikzstyle{fgreen node}=[draw=black!20,circle,fill=black!20,minimum size=8pt,inner sep=0pt]
    \tikzstyle{fred node}=[draw=black!20,circle,fill=black!20,minimum size=8pt,inner sep=0pt]

    \draw (9,-3) node {(c)};

    \draw (0,0) node[graph node] (1) [label=120:$1$] {}
        -- (3,2) node[graph node] (2) [label=120:$2$] {}
        -- (6,4) node[fred node] (3) [label=120:$3$] {}
        -- (8,7) node[blue node] (4) [label=120:$4$] {}
        -- (10,10) node[green node] (5) [label=120:$5$] {}
        -- (14,10) node[graph node] (6) [label=120:$6$] {}
        -- (18,10) node[graph node] (7) [label=120:$7$] {}
        -- (18,7) node[graph node] (8) [label=120:$8$] {};
    \draw (3) -- (9,2) node[red node] (9) [label=90:$9$] {}
        -- (12,0) node[graph node] (10) [label=90:$10$] {}
        -- (15,2) node[graph node] (11) [label=90:$11$] {}
        -- (18,4) node[graph node] (12) [label=120:$12$] {};
    \draw (10) -- (6,0) node[graph node] (13) [label=90:$13$] {};

    \draw[dashed] (1) -- (13);
    \draw[dashed] (8) -- (12);

    \draw[black, ultra thick,->] (9) -- (3);
\end{tikzpicture}
\end{center}
\begin{center}
\begin{tikzpicture}[scale=0.19,>=stealth',shorten >=1pt]
    \tikzstyle{graph node}=[draw,circle,fill=black,minimum size=4pt,
                            inner sep=0pt]

    \tikzstyle{blue node}=[draw,circle,fill=black,minimum size=8pt,inner sep=0pt]
    \tikzstyle{green node}=[draw,circle,fill=black,minimum size=8pt,inner sep=0pt]
    \tikzstyle{red node}=[draw,circle,fill=black,minimum size=8pt,inner sep=0pt]
    \tikzstyle{fblue node}=[draw=black!20,circle,fill=black!20,minimum size=8pt,inner sep=0pt]
    \tikzstyle{fgreen node}=[draw=black!20,circle,fill=black!20,minimum size=8pt,inner sep=0pt]
    \tikzstyle{fred node}=[draw=black!20,circle,fill=black!20,minimum size=8pt,inner sep=0pt]
    \tikzstyle{plain node}=[draw]

    \draw (9,-3) node {(d)};

    \draw (0,0) node[fred node] (1) [label=120:$1$] {}
        -- (3,2) node[graph node] (2) [label=120:$2$] {}
        -- (6,4) node[red node] (3) [label=120:$3$] {}
        -- (8,7) node[blue node] (4) [label=120:$4$] {}
        -- (10,10) node[green node] (5) [label=120:$5$] {}
        -- (14,10) node[graph node] (6) [label=120:$6$] {}
        -- (18,10) node[graph node] (7) [label=120:$7$] {}
        -- (18,7) node[graph node] (8) [label=120:$8$] {};
    \draw (3) -- (9,2) node[graph node] (9) [label=90:$9$] {}
        -- (12,0) node[graph node] (10) [label=90:$10$] {}
        -- (15,2) node[graph node] (11) [label=90:$11$] {}
        -- (18,4) node[graph node] (12) [label=120:$12$] {};
    \draw (10) -- (6,0) node[graph node] (13) [label=90:$13$] {};

    \draw[dashed] (1) -- (13);
    \draw[dashed] (8) -- (12);

    \draw[black, ultra thick, ->] (3) -- (2) -- (1);
\end{tikzpicture}
\begin{tikzpicture}[scale=0.19,>=stealth',shorten >=1pt]
    \tikzstyle{graph node}=[draw,circle,fill=black,minimum size=4pt,
                            inner sep=0pt]

    \tikzstyle{blue node}=[draw,circle,fill=black,minimum size=8pt,inner sep=0pt]
    \tikzstyle{green node}=[draw,circle,fill=black,minimum size=8pt,inner sep=0pt]
    \tikzstyle{red node}=[draw,circle,fill=black,minimum size=8pt,inner sep=0pt]
    \tikzstyle{fblue node}=[draw=black!20,circle,fill=black!20,minimum size=8pt,inner sep=0pt]
    \tikzstyle{fgreen node}=[draw=black!20,circle,fill=black!20,minimum size=8pt,inner sep=0pt]
    \tikzstyle{fred node}=[draw=black!20,circle,fill=black!20,minimum size=8pt,inner sep=0pt]
    \tikzstyle{plain node}=[draw]

    \draw (9,-3) node {(e)};

    \draw (0,0) node[red node] (1) [label=120:$1$] {}
        -- (3,2) node[fblue node] (2) [label=120:$2$] {}
        -- (6,4) node[fgreen node] (3) [label=120:$3$] {}
        -- (8,7) node[blue node] (4) [label=120:$4$] {}
        -- (10,10) node[green node] (5) [label=120:$5$] {}
        -- (14,10) node[graph node] (6) [label=120:$6$] {}
        -- (18,10) node[graph node] (7) [label=120:$7$] {}
        -- (18,7) node[graph node] (8) [label=120:$8$] {};
    \draw (3) -- (9,2) node[graph node] (9) [label=90:$9$] {}
        -- (12,0) node[graph node] (10) [label=90:$10$] {}
        -- (15,2) node[graph node] (11) [label=90:$11$] {}
        -- (18,4) node[graph node] (12) [label=120:$12$] {};
    \draw (10) -- (6,0) node[graph node] (13) [label=90:$13$] {};

    \draw[dashed] (1) -- (13);
    \draw[dashed] (8) -- (12);

    \draw[black, ultra thick,->] (5) -- (4) -- (3);
    \draw[black, ultra thick, ->] (3) -- (2);

\end{tikzpicture}
\end{center}
\caption{\label{fig:Ybasis} An exchange cycle starting from the root configuration $\{1,2,3\}$ and using a single $1$-cell (c) that does not respect the flow at the non-trivial vertex $3$.  Large bold nodes indicate the initial positions of particles and light nodes their final positions.  In paths (a),(b),(d) and (e) particles move according to the vector field.}
\end{figure}

Such a cycle can be broken down into a product of $Y$-cycles in which pairs of particles are exchanged using three arms of the tree connected to the nontrivial vertex $v$ identified by $u,1$ and some $v_j$ where $v_j$ is a vertex in $c$ with $v<v_j<u$.  Figure \ref{fig:Ybasis2} shows a cycle homotopic to the cycle in figure \ref{fig:Ybasis} broken into the product of two $Y$-cycles; paths (a) through (c) and (d) through (e) respectively.  Notice that moving according to the vector field one returns from the initial configuration in figure \ref{fig:Ybasis2}(a) to the root configuration in figure \ref{fig:Ybasis}(a) and similarly one returns from the final configuration in figure \ref{fig:Ybasis2}(e) to the final configuration figure \ref{fig:Ybasis2}(d).  Then by contracting adjacent $1$-cells in the paths where the direction of the edge has been reversed it is straight forward to verify that the cycles in figures \ref{fig:Ybasis} and \ref{fig:Ybasis2} are indeed homotopic.

\begin{figure}[tbh]
\begin{center}

\begin{tikzpicture}[scale=0.19,>=stealth',shorten >=1pt]
    \tikzstyle{graph node}=[draw,circle,fill=black,minimum size=4pt,
                            inner sep=0pt]

    \tikzstyle{blue node}=[draw,circle,fill=black,minimum size=8pt,inner sep=0pt]
    \tikzstyle{green node}=[draw,circle,fill=black,minimum size=8pt,inner sep=0pt]
    \tikzstyle{red node}=[draw,circle,fill=black,minimum size=8pt,inner sep=0pt]
    \tikzstyle{fblue node}=[draw=black!20,circle,fill=black!20,minimum size=8pt,inner sep=0pt]
    \tikzstyle{fgreen node}=[draw=black!20,circle,fill=black!20,minimum size=8pt,inner sep=0pt]
    \tikzstyle{fred node}=[draw=black!20,circle,fill=black!20,minimum size=8pt,inner sep=0pt]

    \draw (9,-3) node {(a)};

    \draw (0,0) node[blue node] (1) [label=120:$1$] {}
        -- (3,2) node[graph node] (2) [label=120:$2$] {}
        -- (6,4) node[graph node] (3) [label=120:$3$] {}
        -- (8,7) node[fblue node] (4) [label=120:$4$] {}
        -- (10,10) node[green node] (5) [label=120:$5$] {}
        -- (14,10) node[graph node] (6) [label=120:$6$] {}
        -- (18,10) node[graph node] (7) [label=120:$7$] {}
        -- (18,7) node[graph node] (8) [label=120:$8$] {};
    \draw (3) -- (9,2) node[red node] (9) [label=90:$9$] {}
        -- (12,0) node[graph node] (10) [label=90:$10$] {}
        -- (15,2) node[graph node] (11) [label=90:$11$] {}
        -- (18,4) node[graph node] (12) [label=120:$12$] {};
    \draw (10) -- (6,0) node[graph node] (13) [label=90:$13$] {};

    \draw[dashed] (1) -- (13);
    \draw[dashed] (8) -- (12);

    \draw[black, ultra thick,->] (1) -- (2) -- (3) -- (4);

\end{tikzpicture}
\begin{tikzpicture}[scale=0.19,>=stealth',shorten >=1pt]
    \tikzstyle{graph node}=[draw,circle,fill=black,minimum size=4pt,
                            inner sep=0pt]

    \tikzstyle{blue node}=[draw,circle,fill=black,minimum size=8pt,inner sep=0pt]
    \tikzstyle{green node}=[draw,circle,fill=black,minimum size=8pt,inner sep=0pt]
    \tikzstyle{red node}=[draw,circle,fill=black,minimum size=8pt,inner sep=0pt]
    \tikzstyle{fblue node}=[draw=black!20,circle,fill=black!20,minimum size=8pt,inner sep=0pt]
    \tikzstyle{fgreen node}=[draw=black!20,circle,fill=black!20,minimum size=8pt,inner sep=0pt]
    \tikzstyle{fred node}=[draw=black!20,circle,fill=black!20,minimum size=8pt,inner sep=0pt]
    \tikzstyle{plain node}=[draw]

    \draw (9,-3) node {(b)};

    \draw (0,0) node[fred node] (1) [label=120:$1$] {}
        -- (3,2) node[graph node] (2) [label=120:$2$] {}
        -- (6,4) node[graph node] (3) [label=120:$3$] {}
        -- (8,7) node[blue node] (4) [label=120:$4$] {}
        -- (10,10) node[green node] (5) [label=120:$5$] {}
        -- (14,10) node[graph node] (6) [label=120:$6$] {}
        -- (18,10) node[graph node] (7) [label=120:$7$] {}
        -- (18,7) node[graph node] (8) [label=120:$8$] {};
    \draw (3) -- (9,2) node[red node] (9) [label=90:$9$] {}
        -- (12,0) node[graph node] (10) [label=90:$10$] {}
        -- (15,2) node[graph node] (11) [label=90:$11$] {}
        -- (18,4) node[graph node] (12) [label=120:$12$] {};
    \draw (10) -- (6,0) node[graph node] (13) [label=90:$13$] {};

    \draw[dashed] (1) -- (13);
    \draw[dashed] (8) -- (12);

    \draw[black, ultra thick, ->] (9) -- (3) -- (2) -- (1);
\end{tikzpicture}
\begin{tikzpicture}[scale=0.19,>=stealth',shorten >=1pt]
    \tikzstyle{graph node}=[draw,circle,fill=black,minimum size=4pt,
                            inner sep=0pt]

    \tikzstyle{blue node}=[draw,circle,fill=black,minimum size=8pt,inner sep=0pt]
    \tikzstyle{green node}=[draw,circle,fill=black,minimum size=8pt,inner sep=0pt]
    \tikzstyle{red node}=[draw,circle,fill=black,minimum size=8pt,inner sep=0pt]
    \tikzstyle{fblue node}=[draw=black!20,circle,fill=black!20,minimum size=8pt,inner sep=0pt]
    \tikzstyle{fgreen node}=[draw=black!20,circle,fill=black!20,minimum size=8pt,inner sep=0pt]
    \tikzstyle{fred node}=[draw=black!20,circle,fill=black!20,minimum size=8pt,inner sep=0pt]
    \tikzstyle{plain node}=[draw]

    \draw (9,-3) node {(c)};

    \draw (0,0) node[red node] (1) [label=120:$1$] {}
        -- (3,2) node[graph node] (2) [label=120:$2$] {}
        -- (6,4) node[graph node] (3) [label=120:$3$] {}
        -- (8,7) node[blue node] (4) [label=120:$4$] {}
        -- (10,10) node[green node] (5) [label=120:$5$] {}
        -- (14,10) node[graph node] (6) [label=120:$6$] {}
        -- (18,10) node[graph node] (7) [label=120:$7$] {}
        -- (18,7) node[graph node] (8) [label=120:$8$] {};
    \draw (3) -- (9,2) node[fblue node] (9) [label=90:$9$] {}
        -- (12,0) node[graph node] (10) [label=90:$10$] {}
        -- (15,2) node[graph node] (11) [label=90:$11$] {}
        -- (18,4) node[graph node] (12) [label=120:$12$] {};
    \draw (10) -- (6,0) node[graph node] (13) [label=90:$13$] {};

    \draw[dashed] (1) -- (13);
    \draw[dashed] (8) -- (12);

    \draw[black, ultra thick, ->] (4) -- (3) -- (9);
\end{tikzpicture}
\end{center}
\begin{center}
\begin{tikzpicture}[scale=0.19,>=stealth',shorten >=1pt]
    \tikzstyle{graph node}=[draw,circle,fill=black,minimum size=4pt,inner sep=0pt]
    \tikzstyle{blue node}=[draw,circle,fill=black,minimum size=8pt,inner sep=0pt]
    \tikzstyle{green node}=[draw,circle,fill=black,minimum size=8pt,inner sep=0pt]
    \tikzstyle{red node}=[draw,circle,fill=black,minimum size=8pt,inner sep=0pt]
    \tikzstyle{fblue node}=[draw=black!20,circle,fill=black!20,minimum size=8pt,inner sep=0pt]
    \tikzstyle{fgreen node}=[draw=black!20,circle,fill=black!20,minimum size=8pt,inner sep=0pt]
    \tikzstyle{fred node}=[draw=black!20,circle,fill=black!20,minimum size=8pt,inner sep=0pt]
    \tikzstyle{plain node}=[draw]

    \draw (9,-3) node {(d)};

    \draw (0,0) node[red node] (1) [label=120:$1$] {}
        -- (3,2) node[fblue node] (2) [label=120:$2$] {}
        -- (6,4) node[graph node] (3) [label=120:$3$] {}
        -- (8,7) node[graph node] (4) [label=120:$4$] {}
        -- (10,10) node[green node] (5) [label=120:$5$] {}
        -- (14,10) node[graph node] (6) [label=120:$6$] {}
        -- (18,10) node[graph node] (7) [label=120:$7$] {}
        -- (18,7) node[graph node] (8) [label=120:$8$] {};
    \draw (3) -- (9,2) node[blue node] (9) [label=90:$9$] {}
        -- (12,0) node[graph node] (10) [label=90:$10$] {}
        -- (15,2) node[graph node] (11) [label=90:$11$] {}
        -- (18,4) node[graph node] (12) [label=120:$12$] {};
    \draw (10) -- (6,0) node[graph node] (13) [label=90:$13$] {};

    \draw[dashed] (1) -- (13);
    \draw[dashed] (8) -- (12);

    \draw[black, ultra thick, ->] (9) -- (3) -- (2);
\end{tikzpicture}
\begin{tikzpicture}[scale=0.19,>=stealth',shorten >=1pt]
    \tikzstyle{graph node}=[draw,circle,fill=black,minimum size=4pt,inner sep=0pt]
    \tikzstyle{blue node}=[draw,circle,fill=black,minimum size=8pt,inner sep=0pt]
    \tikzstyle{green node}=[draw,circle,fill=black,minimum size=8pt,inner sep=0pt]
    \tikzstyle{red node}=[draw,circle,fill=black,minimum size=8pt,inner sep=0pt]
    \tikzstyle{fblue node}=[draw=black!20,circle,fill=black!20,minimum size=8pt,inner sep=0pt]
    \tikzstyle{fgreen node}=[draw=black!20,circle,fill=black!20,minimum size=8pt,inner sep=0pt]
    \tikzstyle{fred node}=[draw=black!20,circle,fill=black!20,minimum size=8pt,inner sep=0pt]
    \tikzstyle{plain node}=[draw]

    \draw (9,-3) node {(e)};

    \draw (0,0) node[red node] (1) [label=120:$1$] {}
        -- (3,2) node[blue node] (2) [label=120:$2$] {}
        -- (6,4) node[graph node] (3) [label=120:$3$] {}
        -- (8,7) node[graph node] (4) [label=120:$4$] {}
        -- (10,10) node[green node] (5) [label=120:$5$] {}
        -- (14,10) node[graph node] (6) [label=120:$6$] {}
        -- (18,10) node[graph node] (7) [label=120:$7$] {}
        -- (18,7) node[graph node] (8) [label=120:$8$] {};
    \draw (3) -- (9,2) node[fgreen node] (9) [label=90:$9$] {}
        -- (12,0) node[graph node] (10) [label=90:$10$] {}
        -- (15,2) node[graph node] (11) [label=90:$11$] {}
        -- (18,4) node[graph node] (12) [label=120:$12$] {};
    \draw (10) -- (6,0) node[graph node] (13) [label=90:$13$] {};

    \draw[dashed] (1) -- (13);
    \draw[dashed] (8) -- (12);

    \draw[black, ultra thick, ->] (5) -- (4) -- (3) -- (9);

\end{tikzpicture}
\begin{tikzpicture}[scale=0.19,>=stealth',shorten >=1pt]
    \tikzstyle{graph node}=[draw,circle,fill=black,minimum size=4pt,inner sep=0pt]
    \tikzstyle{blue node}=[draw,circle,fill=black,minimum size=8pt,inner sep=0pt]
    \tikzstyle{green node}=[draw,circle,fill=black,minimum size=8pt,inner sep=0pt]
    \tikzstyle{red node}=[draw,circle,fill=black,minimum size=8pt,inner sep=0pt]
    \tikzstyle{fblue node}=[draw=black!20,circle,fill=black!20,minimum size=8pt,inner sep=0pt]
    \tikzstyle{fgreen node}=[draw=black!20,circle,fill=black!20,minimum size=8pt,inner sep=0pt]
    \tikzstyle{fred node}=[draw=black!20,circle,fill=black!20,minimum size=8pt,inner sep=0pt]
    \tikzstyle{plain node}=[draw]

    \draw (9,-3) node {(f)};

    \draw (0,0) node[red node] (1) [label=120:$1$] {}
        -- (3,2) node[blue node] (2) [label=120:$2$] {}
        -- (6,4) node[graph node] (3) [label=120:$3$] {}
        -- (8,7) node[graph node] (4) [label=120:$4$] {}
        -- (10,10) node[fblue node] (5) [label=120:$5$] {}
        -- (14,10) node[graph node] (6) [label=120:$6$] {}
        -- (18,10) node[graph node] (7) [label=120:$7$] {}
        -- (18,7) node[graph node] (8) [label=120:$8$] {};
    \draw (3) -- (9,2) node[green node] (9) [label=90:$9$] {}
        -- (12,0) node[graph node] (10) [label=90:$10$] {}
        -- (15,2) node[graph node] (11) [label=90:$11$] {}
        -- (18,4) node[graph node] (12) [label=120:$12$] {};
    \draw (10) -- (6,0) node[graph node] (13) [label=90:$13$] {};

    \draw[dashed] (1) -- (13);
    \draw[dashed] (8) -- (12);

    \draw[black, ultra thick, ->] (2) -- (3) -- (4) -- (5);
\end{tikzpicture}
\end{center}
\caption{\label{fig:Ybasis2} Examples of paths that form $Y$-cycles in the over complete basis; large bold nodes indicate the initial positions of particles on the path and light nodes the final position a particle moves to. (a),(b) and (c) together form a $Y$-cycle, exchanging two particles at the non-trivial vertex $3$, similarly (c),(d) and (e) also form a $Y$-cycle.  Paths (a) through (e) together in order is a cycle homotopic to the exchange cycle starting from the root configuration shown in figure \ref{fig:Ybasis}.}
\end{figure}

Given a cycle $C$ from $\vx_0$ associated with a $1$-cell $\mathcal{c}$ that does not respect the ordering at a nontrivial vertex to obtain a factorization of $C$ as a product of $Y$-cycles one need only start from $c$ and follow $C$ until it is necessary to move a third particle.  Instead of moving the third particle close the path to make a $Y$-cycle, which requires moving only one of the two particles moved so far.  Then retrace ones steps to rejoin $C$ and move the third particle through the nontrivial vertex again close a $Y$-cycle and repeat.  As any permutation can be written as the product of exchanges any such cycle $C$ can be factored as a product of $Y$-cycles.

Finally, as any $n$-particle cycle can be written as a closed sequence of $1$-cells and between $1$-cells we can add contactable paths according to the vector field without changing the phase associated with a cycle, we see that the AB-cycles and the cycles associated with $Y$ subgraphs centered at the nontrivial vertices form a basis for the $n$-particle cycles.  Clearly this spanning set will, in general, be over-complete as many relations between these cycles exist in a typical graph, in fact the full discrete Morse theory argument shows that all such relations are determined by critical $2$-cells \cite{FS05}.

\end{document}